\documentclass[11pt,twoside]{article} 

\usepackage{amsmath,amssymb,bm,xspace}
\usepackage{times}
\usepackage{graphicx}
\usepackage{latexsym}
\usepackage[tight,footnotesize]{subfigure}
\usepackage{pdfsync}				            

\usepackage{jmlr2e}

\ShortHeadings{Sequence Memory in Recurrent Networks}{Charles, Yin and Rozell}

\firstpageno{1}

\begin{document}

\title{Distributed Sequence Memory of Multidimensional Inputs in Recurrent Networks}%

\author{\name Adam S. Charles \email adamsc@princeton.edu \\
       \addr Princeton Neuroscience Institute\\
       Princeton University\\
       Princeton, NJ 08544, USA
       \AND
       \name Dong Yin \email dongyin@berkeley.edu \\
       \addr Department of Electrical Engineering and Computer Sciences\\
       University of California, Berkeley\\
       Berkeley, CA 94720-1776, USA
       \AND
       \name Christopher J. Rozell \email crozell@gatech.edu \\
       \addr School of Electrical and Computer Engineering\\
       Georgia Institute of Technology\\
       Atlanta, GA 30332-0250, USA}

\editor{}
	   
\maketitle


\begin{abstract}%

Recurrent neural networks (RNNs) have drawn interest from machine learning researchers because of their effectiveness at preserving past inputs for time-varying data processing tasks.  To understand the success and limitations of RNNs, it is critical that we advance our analysis of their fundamental memory properties.  We focus on echo state networks (ESNs), which are RNNs with simple memoryless nodes and random connectivity. In most existing analyses, the short-term memory (STM) capacity results conclude that the ESN network size must scale linearly with the input size for unstructured inputs. The main contribution of this paper is to provide general results characterizing the STM capacity for linear ESNs with multidimensional input streams when the inputs have common low-dimensional structure: sparsity in a basis or significant statistical dependence between inputs. In both cases, we show that the number of nodes in the network must scale linearly with the information rate and poly-logarithmically with the input dimension. The analysis relies on advanced applications of random matrix theory and results in explicit non-asymptotic bounds on the recovery error. Taken together, this analysis provides a significant step forward in our understanding of the STM properties in RNNs.
\end{abstract}


\newpage

\section{Introduction}
\label{sec:intro}

Recurrent neural networks (RNNs) have drawn interest from researchers because of their effectiveness at processing sequences of data~\citep{JAE:2001,lukovsevivcius2012practical,hinaut2014exploring}.  While deep networks have shown remarkable performance improvements at task such as image classification, RNNs have recently been successfully employed as layers in conventional deep neural networks to expand these tools into tasks with time-varying data~\citep{sukhbaatar2015end,gregor2015draw,graves2013speech,bashivan2016learningeeg}. This inclusion is becoming increasingly important as neural networks are being applied to a growing variety of inherently temporal high-dimensional data, such as video~\citep{donahue2015long}, audio~\citep{graves2013speech}, EEG data~\citep{bashivan2016learningeeg}, two-photon calcium imaging ~\citep{apthorpe2016automatic}. Despite the growing use of both deep and recurrent networks, theory characterizing the properties of such networks remain relatively unexplored. For deep neural networks, much of the computational power is often attributed to flexibility in learned representations~\citep{mallat2016understanding,vardan2016convolutional,patel2015probabilistic}. The power of RNNs, however, is tied to the ability of the recurrent network dynamics to act as a distributed memory substrate, preserving information about past inputs to leverage temporal dependencies for data processing tasks such as classification and prediction.  To understand the success and limitations of RNNs, it is critical that we advance our analysis of the fundamental memory properties of these network structures.

There are many types of recurrent network structures that have been employed in machine learning applications, each with varying complexity in the network elements and the training procedures.  In this paper we will focus on RNN structures known as echo state networks (ESNs).  These networks have discrete time continuous-valued nodes $\ensuremath{\bm{{x}}[n]}\xspace\in\ensuremath{\mathbb{R}}\xspace^\ensuremath{M}\xspace$ that evolve at time $n$ in response to the inputs $\ensuremath{\bm{{s}}[n]}\xspace\in\ensuremath{\mathbb{R}}\xspace^\ensuremath{L}\xspace$ according to the dynamics:
\begin{gather}
        \ensuremath{\bm{{x}}[n+1]}\xspace = f(\ensuremath{\bm{W}}\xspace\ensuremath{\bm{{x}}[n]}\xspace + \bm{Z}\ensuremath{\bm{{s}}[n]}\xspace + \ensuremath{\bm{\epsilon}[n]}\xspace),     \label{eqn:genNNdyn} 
\end{gather}
where $\ensuremath{\bm{W}}\xspace\in\ensuremath{\mathbb{R}}\xspace^{\ensuremath{M}\xspace\times\ensuremath{M}\xspace}$ is the connectivity matrix defining the recurrent dynamics, $\bm{Z}\in\ensuremath{\mathbb{R}}\xspace^{\ensuremath{M}\xspace\times\ensuremath{L}\xspace}$ is the weight vector describing how the input drives the network, $f(\cdot): \ensuremath{\mathbb{R}}\xspace^\ensuremath{M}\xspace\rightarrow\ensuremath{\mathbb{R}}\xspace^\ensuremath{M}\xspace$ is an element-wise nonlinearity evaluated at each node and $\ensuremath{\bm{\epsilon}[n]}\xspace\in\ensuremath{\mathbb{R}}\xspace^\ensuremath{M}\xspace$ represents the error due to potential system imperfections~\citep{JAE:2001,wilson1972excitatory,amari1972characteristics,sompolinsky1988chaos,MAA:2002}.  In an ESN, the connectivity matrix $\ensuremath{\bm{W}}\xspace$ is random and untrained, while the simple individual nodes have a single state variable with no memory.  This is in contrast to approaches such as long short-term memory units~\citep{sak2014long,lipton2016lstm,kalchbrenner2016gridlstm} which have individual nodes with complex memory properties.  As with many other recent papers, we will also focus on linear networks where $f(\cdot)$ is the identity function~\citep{JAE:2001,JAE:2004,WHI:2004,SOM:2008,GAN:2010,charles2014short,latham2013}.

The memory capacity of these networks has been studied in both the machine learning and computational neuroscience literature.  In the approach of interest, the short-term memory (STM) of a network is characterized by quantifying the relationship between the transient network activity and the recent history of the exogenous input stream driving the network~\citep{JAE:2004,MAA:2002,GAN:2010,latham2013,verstraeten2007experimental,WHI:2004,lukovsevivcius2009reservoir,buonomano2009state,charles2014short}.  Note that this is in contrast to alternative approaches that characterize long-term memory in RNNs through quantifying the number of distinct network attractors that can be used to stably remember input patterns with the asymptotic network state.  In the vast majority of the existing theoretical analysis of STM, the results conclude that networks with \ensuremath{M}\xspace nodes can only recover inputs of length $\ensuremath{N}\xspace \leq \ensuremath{M}\xspace$~\citep{WHI:2004,latham2013} when the inputs are unstructured. 

However, in any machine-learning problem of interest, the input statistical structure is precisely what we intend to exploit to accomplish meaningful tasks.  For one example, many signals are well-known to admit a sparse representation in a transform dictionary~\citep{ELA:2008,DAV:2006}.  In fact, some classes of deep neural networks have been designed to induce sparsity at higher layers that may serve as inputs into the recurrent layers~\citep{lecun2010convolutional,kavukcuoglu2010learning}.   For another example, a collection of time-varying input streams (e.g., pixels or image features in a video stream) are often heavily correlated.  In the specific case of single input streams ($\ensuremath{L}\xspace=1$) with inputs that are $\ensuremath{K}\xspace$-sparse in a basis, recent work~\citep{charles2014short} has shown that the STM capacity can scale as favorably as $\ensuremath{M}\xspace = \mathcal{O}\left(\ensuremath{K}\xspace\log^{\gamma}(\ensuremath{N}\xspace)\right) \leq \ensuremath{N}\xspace$, where $\gamma \geq 1$ is a constant. In other words, the memory capacity can scale linearly with the information rate in the signal and only logarithmically with the signal dimension, resulting in the potential for recovery of inputs of length $\ensuremath{N}\xspace \gg \ensuremath{M}\xspace$.  Unfortunately, existing analyses~\citep{JAE:2001,WHI:2004,GAN:2010,charles2014short} are generally specific to the restricted case of single time-series inputs ($\ensuremath{L}\xspace=1$) or unstructured inputs~\citep{verstraeten2010memory}.

Conventional wisdom is that structured inputs should lead to much higher STM capacity, though this has never been addressed with strong analysis in the general case of ESNs with multidimensional input streams.  The main contribution of this paper is to provide general results characterizing the STM capacity for linear randomly connected ESNs with multidimensional input streams when the inputs are either sparse in a basis or have significant statistical dependence (with no sparsity assumption).  In both cases, we show that the number of nodes in the network must scale linearly with the information rate and poly-logarithmically with the total input dimension.  The analysis relies on advanced applications of random matrix theory, and results in non-asymptotic analysis of explicit bounds on the recovery error.  Taken together, this analysis provides a significant step forward in our understanding of the STM properties in RNNs.  While this paper is primarily focused on network structures in the context of RNNs in machine learning, these results also provide foundation for the theoretical understanding of recurrent network structures in biological neural networks, as well as the memory properties in other network structures with similar dynamics (e.g., opinion dynamics in social networks). 


\section{Background and Related Work}


\subsection{Short Term Memory in Recurrent Networks}

Many approaches have been used to analyze the STM of randomly connected networks, including nonlinear networks~\citep{sompolinsky1988chaos,massar2013mean,faugeras2009constructive,rajan2010stimulus,galtier2014local,wainrib2015context} and linear networks~\citep{JAE:2001,JAE:2004,WHI:2004,SOM:2008,GAN:2010,charles2014short,latham2013} with both discrete-time and continuous-time dynamics. These methods can be broadly be classified as either correlation-based methods~\citep{WHI:2004,SOM:2008} or uniqueness methods~\citep{JAE:2001,MAA:2002,JAE:2004,charles2014short,legenstein2007edge,busing2010connectivity}. Correlation methods focus on quantifying the correlation between the network state and recent network inputs. In these studies, the STM is defined as the time of the oldest input where the correlation between the network state and that input remains above a given threshold~\citep{WHI:2004,SOM:2008}. These methods have mostly been applied to discrete-time systems, and have resulted in bounds on the STM that scale linearly with the number of nodes (i.e. $\ensuremath{M}\xspace > \ensuremath{N}\xspace$).   

In contrast, uniqueness methods instead aim to show that  different network states correspond to unique input sequences (i.e. the network dynamics are bijective).\footnote{We note that uniqueness-based methods imply recovery-based methods, modulo a recovery algorithm, as often recovery guarantees are based on some semblance of a bijection.} For uniqueness methods, the STM is defined as the longest input length where this input-network state bijection still holds. These methods have been used under the term \emph{separability property} for continuous-time liquid state machines~\citep{MAA:2002,vapnik1971uniform,legenstein2007edge,latham2013,busing2010connectivity} and under the term \emph{echo-state property} for discrete-time ESNs~\citep{JAE:2001,yildiz2012re,buehner2006tighter,manjunath2013echo}.   The echo-state property is the method most related to the approach we take here, and essentially derives the maximum length of the input signal such that the resulting network states remain unique. While this property guarantees a bijection between inputs and network states, it does not take into account input signal structure, does not capture the robustness of the mapping, and does not provide guarantees for stably recovering the input from the network activity.



\subsection{Compressed Sensing}
\label{sec:CS}

The compressed sensing literature and its recent extensions include many tools for studying the effects of random matrices applied to low-dimensional signals.  Specifically, in the basic compressed sensing problem we desire to recover the  signal $\ensuremath{\bm{{s}}}\xspace\in\ensuremath{\mathbb{R}}\xspace^\ensuremath{N}\xspace$ from \ensuremath{M}\xspace measurements\footnote{In the case of RNNs, the network node values act as the measurements of our system, prompting the use of \ensuremath{M}\xspace as the number of measurements in this section} generated from a random linear measurement operator,
\begin{gather}
        \ensuremath{\bm{{x}}}\xspace = \ensuremath{\mathcal{{A}}\left( \ensuremath{\bm{{s}}}\xspace \right)}\xspace + \ensuremath{\bm{\epsilon}}\xspace,  \label{eqn:measEQ}
\end{gather}
where $\ensuremath{\bm{\epsilon}}\xspace\in\ensuremath{\mathbb{R}}\xspace^{\ensuremath{M}\xspace}$ represents the potential measurement errors. Typically, \ensuremath{\bm{{s}}}\xspace is assumed to have low-dimensional structure and recovery is performed via a convex optimization program.  The most common example is a sparsity model where \ensuremath{\bm{{s}}}\xspace can be represented as
\begin{gather*}
	\ensuremath{\bm{{s}}}\xspace = \ensuremath{\bm{\Psi}}\xspace\ensuremath{\bm{{a}}}\xspace \nonumber,
\end{gather*}
where $\ensuremath{\bm{\Psi}}\xspace\in\ensuremath{\mathbb{R}}\xspace^{\ensuremath{N}\xspace\times\ensuremath{N}\xspace}$ is a transform matrix and $\ensuremath{\bm{{a}}}\xspace\in\ensuremath{\mathbb{R}}\xspace^{\ensuremath{N}\xspace}$ is the sparse coefficient representation of \ensuremath{\bm{{s}}}\xspace with $\ensuremath{K}\xspace \ll \ensuremath{N}\xspace$ of its entries non-zero. Under this sparsity assumption, the coefficient representation is recoverable if the linear operator \ensuremath{\mathcal{{A}}}\xspace satisfies the restricted isometry property (RIP) that guarantees uniqueness of the compressed measurements.  Specifically, we say that \ensuremath{\mathcal{{A}}}\xspace satisfies the RIP(2\ensuremath{K}\xspace,\ensuremath{\delta}\xspace) if for every 2\ensuremath{K}\xspace-sparse signal \ensuremath{\bm{{s}}}\xspace, the following condition is satisfied:
\begin{gather*}
	\ensuremath{C}\xspace\left(1-\ensuremath{\delta}\xspace\right) \leq \ensuremath{\left|\left| {\ensuremath{\mathcal{{A}}}\xspace{\ensuremath{\bm{{s}}}\xspace}} \right|\right|}\xspace_2^2/\ensuremath{\left|\left| {\ensuremath{\bm{{s}}}\xspace} \right|\right|}\xspace_2^2 \leq \ensuremath{C}\xspace\left(1+\ensuremath{\delta}\xspace\right), \nonumber 
\end{gather*}
where $0<\ensuremath{\delta}\xspace<1$ and $C>0$ is a positive constant. When \ensuremath{\mathcal{{A}}}\xspace satisfies the RIP(2\ensuremath{K}\xspace,\ensuremath{\delta}\xspace) the sparse coefficients \ensuremath{\bm{{a}}}\xspace can be recovered by solving an $\ell_1$-norm based optimization function
\begin{gather}
	\ensuremath{\bm{{a}}}\xspace = \arg\min_{\ensuremath{\bm{{a}}}\xspace}\ensuremath{\left|\left| {\ensuremath{\bm{{a}}}\xspace} \right|\right|}\xspace_1 \quad \mbox{such that} \quad \ensuremath{\left|\left| {\ensuremath{\bm{{x}}}\xspace - \ensuremath{\mathcal{{A}}\left( \ensuremath{\bm{\Psi}}\xspace\ensuremath{\bm{{a}}}\xspace \right)}\xspace} \right|\right|}\xspace_2\leq \ensuremath{\left|\left| {\ensuremath{\bm{\epsilon}}\xspace} \right|\right|}\xspace_2, \label{eqn:l1min}
\end{gather}
up to a reconstruction error given by
\begin{gather}
	\ensuremath{\left|\left| {\ensuremath{\bm{{s}}}\xspace - \widehat{\ensuremath{\bm{{s}}}\xspace}} \right|\right|}\xspace_2 \leq \alpha \ensuremath{\left|\left| {\ensuremath{\bm{\epsilon}}\xspace} \right|\right|}\xspace_2 + \beta\frac{\ensuremath{\left|\left| {\ensuremath{\bm{\Psi}}\xspace^T \left(\ensuremath{\bm{{s}}}\xspace - \ensuremath{\bm{{s}}}\xspace_{\ensuremath{K}\xspace}\right)} \right|\right|}\xspace_1}{\sqrt{\ensuremath{K}\xspace}}, 	\label{eqn:DecRec}
\end{gather}
where $\alpha$ and $\beta$ are constants~\citep{CAN:2006b}. The first term of this recovery error bound depends on the norm of the measurement error \ensuremath{\bm{\epsilon}}\xspace, while the second term depends on the $\ell_1$ difference between the true signal and the best \ensuremath{K}\xspace-sparse approximation of the true vector ($\ensuremath{\bm{{s}}}\xspace_{\ensuremath{K}\xspace}$). This term essentially measures how closely the signal matches the sparsity model. The $\ell_1$ optimization program in~\eqref{eqn:l1min} required for recovery can be solved by many efficient algorithms, including neurally plausible architectures~\citep{ROZ:2008,BAL:2012,shapero.13,charles2012lca}. 

When the data of interest is a matrix $\ensuremath{\bm{{S}}}\xspace\in\ensuremath{\mathbb{R}}\xspace^{\ensuremath{L}\xspace\times\ensuremath{N}\xspace}$, other low-dimensional models have also been explored. For example, as an alternative to a sparsity assumption, the successful low-rank model assumes that there are correlations between rows and columns such that \ensuremath{\bm{{S}}}\xspace has rank $\ensuremath{R}\xspace < \min\{\ensuremath{L}\xspace,\ensuremath{N}\xspace\}$.  We can then write the decomposition of the matrix as
\begin{gather*}
	\ensuremath{\bm{{S}}}\xspace = \ensuremath{\bm{Q}}\xspace\ensuremath{\bm{V}^{\ast}}\xspace \nonumber, \label{eqn:LRdef}
\end{gather*}
where $\ensuremath{\bm{Q}}\xspace\in\ensuremath{\mathbb{R}}\xspace^{\ensuremath{L}\xspace\times\ensuremath{R}\xspace}$ and $\ensuremath{\bm{V}^{\ast}}\xspace\in\ensuremath{\mathbb{R}}\xspace^{\ensuremath{R}\xspace\times\ensuremath{N}\xspace}$. There is a rich and growing literature dedicated to establishing guarantees for recovering low-rank matrices from incomplete measurements. Due to the difficulty of establishing a general matrix-RIP property for observations of a matrix~\citep{recht2010guaranteed}, the guarantees in this literature more commonly use the optimality conditions for specific optimization procedures to show that the resulting solution has bounded error with high probability. The most common optimization program used for low-rank matrix recovery is the nuclear norm minimization,
\begin{gather}
	\ensuremath{\bm{{S}}}\xspace = \arg\min_{\ensuremath{\bm{{S}}}\xspace}\ensuremath{\left|\left| {\ensuremath{\bm{{S}}}\xspace} \right|\right|}\xspace_{\ast} \quad \mbox{such that} \quad \ensuremath{\left|\left| {\ensuremath{\bm{{x}}}\xspace - \ensuremath{\mathcal{{A}}\left( \ensuremath{\bm{{S}}}\xspace \right)}\xspace} \right|\right|}\xspace_2\leq \ensuremath{\left|\left| {\ensuremath{\bm{\epsilon}}\xspace} \right|\right|}\xspace_2, \label{eqn:nucnorm0}
\end{gather}
where the nuclear norm $\ensuremath{\left|\left| {\ensuremath{\bm{{S}}}\xspace} \right|\right|}\xspace_{\ast}$ is defined as the sum of the singular values of \ensuremath{\bm{{S}}}\xspace~\citep{candes2010power,candes2010matrix,recht2010guaranteed,chen2004recovering,fazel2002thesis,singer2010uniqueness,toh2010accelerated,liu2009interior,jaggi2010simple}. This optimization procedure is similar to the $\ell_1$-regularized optimization of Equation~\eqref{eqn:l1min}, however the nuclear-norm induces sparsity in the singular values rather than the matrix entries directly. 

The solution to Equation~\eqref{eqn:nucnorm0} can be shown to satisfy performance guarantees via the dual-certificate approach~\citep{candes2010matrix,ahmed2013compressive}. This technique is a proof by construction and shows that a dual certificate (i.e., a vector whose projections into and out of the space spanned by the singular vectors of \ensuremath{\bm{{S}}}\xspace are bounded) exists. Showing that such a certificate exists demonstrates that Equation~\eqref{eqn:nucnorm0} converges to a valid solution and is key to deriving accuracy bounds~\citep{candes2010matrix,ahmed2013compressive}. Specifically, if the dual certificate exists, then  the solution to Equation~\eqref{eqn:nucnorm0} satisfies the recovery bound
\begin{gather}
	\ensuremath{\left|\left| {\widehat{\ensuremath{\bm{{S}}}\xspace} - \ensuremath{\bm{{S}}}\xspace } \right|\right|}\xspace_F \leq \left(4\sqrt{\min(\ensuremath{N}\xspace,\ensuremath{L}\xspace)\frac{2\ensuremath{N}\xspace\ensuremath{L}\xspace + \ensuremath{M}\xspace}{\ensuremath{M}\xspace}} + 2\right)\epsilon, \label{eqn:NucNormErr}
\end{gather}
where the Forbenius norm $\|\cdot\|_F^2$ is defined as the sum of the squares of all the matrix entries. This bound demonstrates that perfect recovery is achievable in the case where there is no error ($\epsilon=0$). We note that alternate optimization programs with similar guarantees have been proposed in the literature for inferring low-rank matrices (i.e.~\citealp{ahmed2013compressive}), but we will focus on nuclear norm optimization approaches due to the extensive literature on nuclear-norm solvers and the connections to sparse vector inference.


\subsection{STM Capacity via the RIP}

The ideas and tools from the compressed sensing literature have recently been used to show that \emph{a-priori} knowledge of the input sparsity can lead to improvements recovery-based STM capacity results for ESNs.  For a single input stream under a sparsity assumption, \cite{GAN:2010} analyzed an annealed version of the network dynamics to show that the network memory capacity can be larger than the network size. Building on this observation,~\cite{charles2014short} provided an analysis of the exact network dynamics in an ESN (for the single input case of $L=1$), yielding precise bounds on a network's STM capacity captured in the following theorem:
\begin{theorem}\emph{(Theorem 4.1.1,~\citealp{charles2014short})}
	\label{thm:STMwithZ}
        Suppose $\ensuremath{N}\xspace \ge \ensuremath{M}\xspace$, $\ensuremath{N}\xspace \ge \ensuremath{K}\xspace$, $\ensuremath{N}\xspace \ge O(1)$,\footnote{We use $O(1)$ notation to indicate that a variable is a finite constant.} and $L=1$. Let $\bm{U}$ be any unitary matrix of eigenvectors (containing complex conjugate pairs) of the connectivity matrix \ensuremath{\bm{W}}\xspace and for $\ensuremath{M}\xspace$ an even integer, denote the eigenvalues of \ensuremath{\bm{W}}\xspace by $\{e^{j w_m}\}_{m = 1}^{\ensuremath{M}\xspace}$.  Let the first $\ensuremath{M}\xspace/2$ eigenvalues ($\{e^{j w_m}\}_{m = 1}^{\ensuremath{M}\xspace/2}$) be chosen uniformly at random on the complex unit circle (i.e., $\{w_m\}_{m=1}^{\ensuremath{M}\xspace/2}$ is uniformly distributed over $[0, 2\pi)$) and the other $\ensuremath{M}\xspace/2$ eigenvalues as the complex conjugates of these values. Furthermore, let the entries of the input weights \ensuremath{\bm{{z}}}\xspace be i.i.d.\ zero-mean Gaussian random variables with variance $\frac{1}{\ensuremath{M}\xspace}$. Given RIP conditioning \ensuremath{\delta}\xspace and failure probability $\ensuremath{N}\xspace^{-\log^4 \ensuremath{N}\xspace} \le \eta \le \frac{1}{e}$, if 
\begin{gather*}
	\ensuremath{M}\xspace \geq C\frac{\ensuremath{K}\xspace}{\ensuremath{\delta}\xspace^2}\mu^2\left(\ensuremath{\bm{\Psi}}\xspace\right)\log^{5}\left(\ensuremath{N}\xspace\right) \log(\eta^{-1}), \nonumber 
\end{gather*}
then for a universal constant $C$, with probability $1-\eta$ the mapping of length-$\ensuremath{N}\xspace$ input sequences into $\ensuremath{M}\xspace$ network state variables satisfies the RIP($2\ensuremath{K}\xspace$, $\ensuremath{\delta}\xspace$).
\end{theorem}

\begin{figure*}[t]
	\centering
	\includegraphics[width=0.95\textwidth]{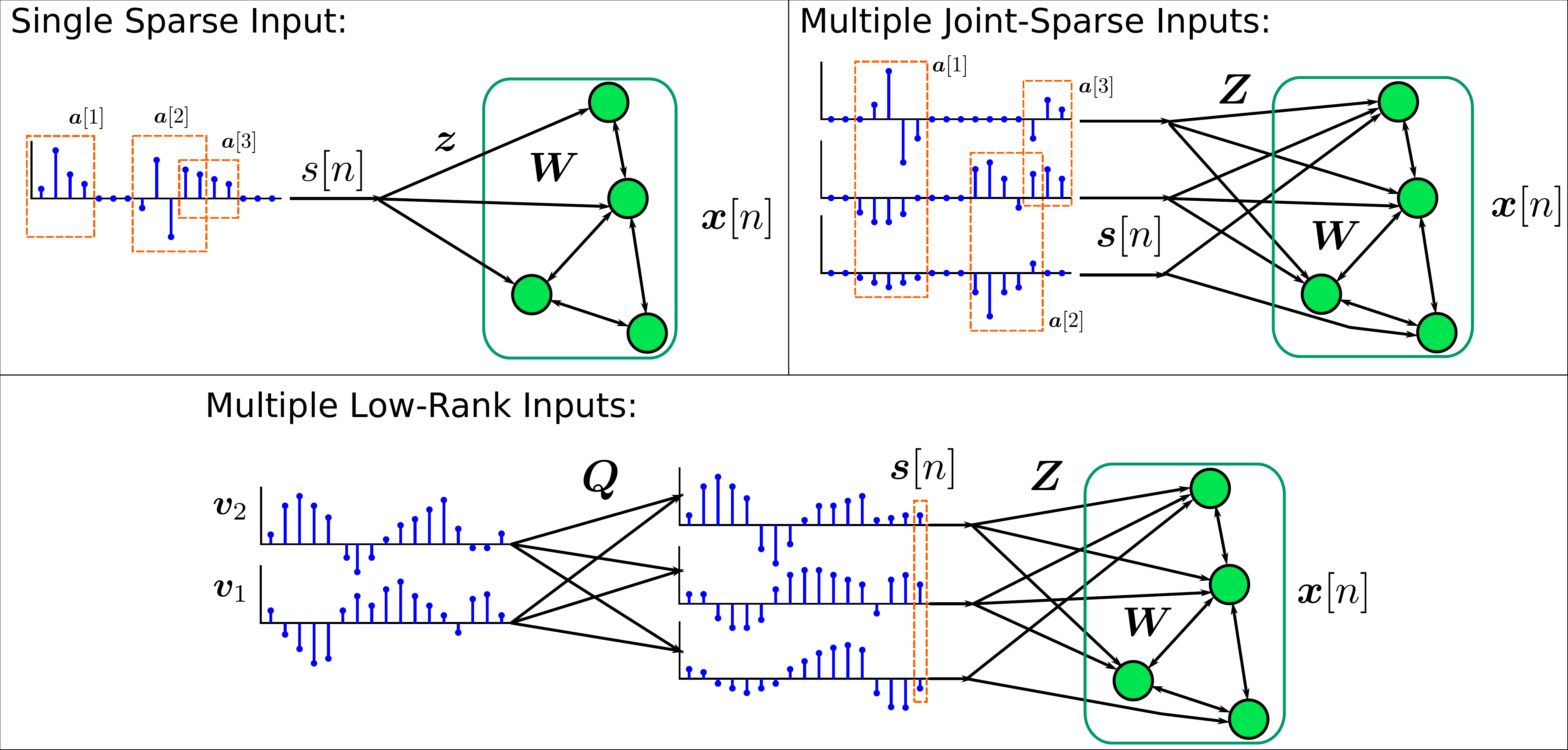}
        \caption{Echo-state networks can efficiently store inputs with a variety of low-dimensional structures. Top-Left: With a single input stream, the coefficients \ensuremath{\bm{{a}}}\xspace represent chunks of activity in the input stream \ensuremath{{s}[n]}\xspace (shown in red boxes). The raw input stream values then drive the network via Equation~\eqref{eqn:genNNdyn}, resulting in a transient network state \ensuremath{\bm{{x}}[N]}\xspace that encodes the input stream. Top-Right: In the case of multiple inputs that are jointly sparse, each coefficient can now represent a chunk of activity both across time and across input streams (as depicted by the red boxes). Bottom: When the multiple input streams are instead low-rank, each input stream is instead described by a linear combination of prototypical vectors $\bm{v}_k$. The matrix $\bm{Q}$ represents how the prototypical vectors are combined in order to obtain the input sequences fed into the network.}
	\label{fig:NNbasic}
\end{figure*}

This theorem proves a rigorous and non-asymptotic bound on the length of the input that can be robustly extracted from the network nodes. By showing the RIP property on the network dynamics, the recovery bound given in Equation~\eqref{eqn:DecRec} establishes the recovery performance for any \ensuremath{N}\xspace-length, \ensuremath{K}\xspace-sparse signal from the resulting network state at time $N$. In short, the number of required nodes scales linearly with the information rate of the signal (i.e., the sparsity level) and poly-logarithmically with the length of the input.  The coherence factor $\mu^2(\bm{\Psi})$, defined as
\begin{gather*}
	\mu\left(\ensuremath{\bm{\Psi}}\xspace\right) = \max_{n=1,\hdots,N}\sup_{t\in[0,2\pi]} \left|\sum_{m = 0}^{N-1}\ensuremath{\bm{\Psi}}\xspace_{m,n}e^{-jtm} \right|, \nonumber
\end{gather*}
expresses the \emph{types} of sparsity that are efficiently stored in the network. Essentially this coherence factor is large (on the order of $\sqrt{\ensuremath{N}\xspace}$) for inputs sparse in the Fourier basis, and is very low (essentially a small constant) for inputs that are sparse in bases different from the Fourier bases (e.g. wavelet transforms).   For the extreme case of Fourier-sparse inputs, the number of nodes must again exceed the number of inputs.  When this coherence is low and $\ensuremath{K}\xspace \ll \ensuremath{M}\xspace$, this bound is a clear improvement over existing results as it allows for $\ensuremath{N}\xspace > \ensuremath{M}\xspace$.  However, this result is restricted to single input streams with one type of low-dimensional structure.  The current paper addresses the much more general problem of multidimensional inputs and other types of low-dimensional structure.


\section{STM for Multi-Input Networks}
\label{sec:stm}

In this work we will use the tools of random matrix theory to establish STM capacity results for recurrent networks under the general conditions of multiple simultaneous input streams and a variety of low-dimensional models. The temporal evolution of the linear network with  multiple inputs is similar to the previous ESN definition, with the main difference being that the input at each time-step $\ensuremath{\bm{{s}}}\xspace[\ensuremath{n}\xspace]\in\ensuremath{\mathbb{R}}\xspace^\ensuremath{L}\xspace$ is a length \ensuremath{L}\xspace vector that drives the network through a feed-forward matrix $\bm{Z}\in\ensuremath{\mathbb{R}}\xspace^{\ensuremath{M}\xspace\times\ensuremath{L}\xspace}$ rather than a feed-forward vector,
\begin{eqnarray}
	\ensuremath{\bm{{x}}[\ensuremath{n}\xspace]}\xspace & = & \ensuremath{\bm{W}}\xspace\ensuremath{\bm{{x}}[\ensuremath{n}\xspace-1]}\xspace + \sum_{l=1}^{\ensuremath{L}\xspace}\ensuremath{\bm{{z}}}\xspace_{l}s_{l}[n] + \widetilde{\bm{\epsilon}}[\ensuremath{n}\xspace] \nonumber \\
        & = & \ensuremath{\bm{W}}\xspace\ensuremath{\bm{{x}}[\ensuremath{n}\xspace-1]}\xspace + \bm{Z}\ensuremath{\bm{{s}}}\xspace[\ensuremath{n}\xspace] + \widetilde{\bm{\epsilon}}[\ensuremath{n}\xspace].  \label{eqn:sysdyn2}
\end{eqnarray}
We denote the columns of $\bm{Z}$ as $\ensuremath{\bm{{z}}}\xspace_l$ to separately notate the vectors mapping each input stream.  

We can write the current network state as a linear function of the inputs by iterating Equation~\eqref{eqn:sysdyn2},
\begin{gather*}
        \ensuremath{\bm{{x}}[\ensuremath{N}\xspace]}\xspace = \sum_{k=1}^{N}{ \bm{W}^{N-k}\bm{Z}\bm{s}[k] } + \bm{\epsilon}, \label{eqn:syssum1} \nonumber
\end{gather*}
where the error term $\bm{\epsilon} = \sum_{k=1}^\ensuremath{N}\xspace \bm{W}^{N-k}\widetilde{\bm{\epsilon}}[k]$ is the accumulated error, and then rewriting sum as a matrix-vector multiply, 
\begin{eqnarray*}
        \ensuremath{\bm{{x}}[\ensuremath{N}\xspace]}\xspace & = & \left[\bm{Z}, \ensuremath{\bm{W}}\xspace\bm{Z}, \cdots, \ensuremath{\bm{W}}\xspace^{\ensuremath{N}\xspace-1}\bm{Z} \right]\left[\ensuremath{\bm{{s}}}\xspace^T[\ensuremath{N}\xspace], \ensuremath{\bm{{s}}}\xspace^T[\ensuremath{N}\xspace-1], \cdots, \ensuremath{\bm{{s}}}\xspace^T[1] \right]^T + \bm{\epsilon} .
\end{eqnarray*}
Depending on the signal statistics in question, we will find it convenient in some cases to express the network dynamics in terms of a linear operator applied to an input matrix, i.e. 
\begin{gather*}
        \ensuremath{\bm{{x}}[\ensuremath{N}\xspace]}\xspace = \ensuremath{\mathcal{{A}}\left( \ensuremath{\bm{{S}}}\xspace \right)}\xspace + \bm{\epsilon}, \nonumber 
\end{gather*}
where $\ensuremath{\bm{{S}}}\xspace = \left[\ensuremath{\bm{{s}}}\xspace^T[\ensuremath{N}\xspace], \ensuremath{\bm{{s}}}\xspace^T[\ensuremath{N}\xspace-1], \cdots, \ensuremath{\bm{{s}}}\xspace^T[1] \right]^T$. In other cases, we find it more convenient to reorganize the columns into an effective measurement matrix applied to a vector of inputs. By defining the eigen-decomposition of $\ensuremath{\bm{W}}\xspace = \bm{U}\bm{D}\bm{U}^{-1}$, we can re-write the dynamics process as 
\begin{eqnarray*}
        \ensuremath{\bm{{x}}[\ensuremath{N}\xspace]}\xspace & = & \bm{U}\left[\bm{D}^0\bm{U}^{-1}\bm{Z}, \bm{D}\bm{U}^{-1}\bm{Z}, \cdots, \bm{D}^{\ensuremath{N}\xspace-1}\bm{U}^{-1}\bm{Z} \right]\left[\ensuremath{\bm{{s}}}\xspace^T[\ensuremath{N}\xspace], \ensuremath{\bm{{s}}}\xspace^T[\ensuremath{N}\xspace-1], \cdots, \ensuremath{\bm{{s}}}\xspace^T[1] \right]^T + \bm{\epsilon} . \nonumber 
\end{eqnarray*}
To simplify this expression, we can reorganize the columns of the linear operator (and the rows of the vector of inputs) such that all the inputs corresponding to the $l^{th}$ input vector $\bm{z}_l$ create a single block. The $k^{th}$ row of the $l^{th}$ block of out matrix is now represented by $\bm{D}^{k-1}\bm{U}^{-1}\bm{z}_l$, which can be written as $\ensuremath{\widetilde{\bm{Z}}}\xspace_l\bm{d}_{k-1}$, where $\ensuremath{\widetilde{\bm{Z}}}\xspace_l=\bm{U}^{-1}\bm{z}_l$ and $\bm{d}_{k-1}$ is the vector of the diagonal elements of $\bm{D}$ raised to the $(k-1)$ power. We can more concisely by defining the matrix $\ensuremath{\bm{F}}\xspace$ consisting of the eigenvalues of \ensuremath{\bm{W}}\xspace raised to different powers (i.e. $\ensuremath{\bm{F}}\xspace_{i,j} = d_i^{j-1}$), resulting in the expression
\begin{gather}
        \ensuremath{\bm{{x}}[\ensuremath{N}\xspace]}\xspace = \ensuremath{\bm{U}}\xspace\left[\ensuremath{\widetilde{\bm{Z}}}\xspace_1\ensuremath{\bm{F}}\xspace, \ensuremath{\widetilde{\bm{Z}}}\xspace_2\ensuremath{\bm{F}}\xspace, \cdots, \ensuremath{\widetilde{\bm{Z}}}\xspace_{\ensuremath{L}\xspace}\ensuremath{\bm{F}}\xspace \right]\left[\ensuremath{\bm{{s}}}\xspace_1^T, \ensuremath{\bm{{s}}}\xspace^T_2, \cdots, \ensuremath{\bm{{s}}}\xspace^T_\ensuremath{L}\xspace \right]^T + \bm{\epsilon} = \bm{A}\widetilde{\ensuremath{\bm{{s}}}\xspace} + \bm{\epsilon}.   \label{eqn:matvec2}
\end{gather}
Since the eigenvalues of \ensuremath{\bm{W}}\xspace here are restricted to reside on the unit circle, we note that \ensuremath{\bm{F}}\xspace is a Vandermonde matrix whose rows are Fourier basis vectors. From Equation~\eqref{eqn:matvec2} we see that the current state is simply the sum of \ensuremath{L}\xspace compressed input streams, where the compression for each block essentially performs the same compression as for a single stream, but modulated by the different feed-forward vectors $\ensuremath{\bm{{z}}}\xspace_l$. 

\subsection{Sparse Multiple Inputs}

To begin, we consider the direct extension of previous results based on sparsity models to the multi-input setting.  In this setting we consider the model where the composite of all input signals is sparse in a basis $\ensuremath{\bm{\Psi}}\xspace\in\ensuremath{\mathbb{R}}\xspace^{\ensuremath{N}\xspace\ensuremath{L}\xspace\times \ensuremath{N}\xspace\ensuremath{L}\xspace}$ so that $\widetilde{\ensuremath{\bm{{s}}}\xspace} = \ensuremath{\bm{\Psi}}\xspace\widetilde{\ensuremath{\bm{{a}}}\xspace}$. This means that each signal stream can be written as $\ensuremath{\bm{{s}}}\xspace_l = \sum_{k=1}^\ensuremath{L}\xspace \ensuremath{\bm{\Psi}}\xspace^{l,k}\ensuremath{\bm{{a}}}\xspace_k$ where $\ensuremath{\bm{\Psi}}\xspace^{l,k}$ is the $\{l,k\}^{th}$ $\ensuremath{N}\xspace\times \ensuremath{N}\xspace$ block of $\widetilde{\ensuremath{\bm{\Psi}}\xspace}$. This signal model captures dependencies between input streams because a given coefficient can influence multiple channels. While in many application the basis $\ensuremath{\bm{\Psi}}\xspace$ is pre-specified (i.e. wavelet decomposition in image processing;~\citealp{christopoulos2000jpeg2000}), these bases can also be learned from exemplar data via dictionary learning algorithms~\citep{OLS:1996,ELA:KSVD}. This sparsity model can be a useful model for signals of interest, such as video signals, where similar sparse decompositions have been used for action recognition~\citep{guha2012learning} and video categorization~\citep{chiang2013learning}. With this model, we will use a generalized notion of the coherence parameter used in~\citep{charles2014short}:
\begin{gather}
	\mu_S\left(\ensuremath{\bm{\Psi}}\xspace\right) = \max_{l,k=1,\hdots,\ensuremath{L}\xspace}\max_{n=1,\hdots,\ensuremath{N}\xspace}\sup_{t\in[0,2\pi]} \frac{\left|\sum_{m = 0}^{N-1}\ensuremath{\bm{\Psi}}\xspace^{l,k}_{m,n}e^{-jtm} \right|}{\|\ensuremath{\bm{\Psi}}\xspace_m^{l,k} \|_2}. \label{eqn:multico}
\end{gather}
In this case, each $\ensuremath{N}\xspace\times\ensuremath{N}\xspace$ block must be different from the Fourier basis to achieve high STM capacity. This restriction is reasonable, since if a single sub-block of \ensuremath{\bm{\Psi}}\xspace was coherent with the Fourier basis, then at least one input stream could be sparse in a Fourier-like basis and hence would be unrecoverable. Using this network and signal model, we obtain the following theorem on the stability of the network representation:

\begin{theorem}
	\label{thm:STMmulti}
	Suppose $\ensuremath{N}\xspace\ensuremath{L}\xspace \ge \ensuremath{M}\xspace$, $\ensuremath{N}\xspace \ge \ensuremath{K}\xspace$ and $\ensuremath{N}\xspace \ge O(1)$. Let $\ensuremath{\bm{U}}\xspace$ be any unitary matrix of eigenvectors (containing complex conjugate pairs) and the entries of $\bm{Z}$ be i.i.d. zero-mean Gaussian random variables with variance $\frac{1}{\ensuremath{M}\xspace}$. For $\ensuremath{M}\xspace$ an even integer, denote the eigenvalues of \ensuremath{\bm{W}}\xspace by $\{e^{j w_m}\}_{m = 1}^{\ensuremath{M}\xspace}$.  Let the first $\ensuremath{M}\xspace/2$ eigenvalues ($\{e^{j w_m}\}_{m = 1}^{\ensuremath{M}\xspace/2}$) be chosen uniformly at random on the complex unit circle (i.e., we chose $\{w_m\}_{m=1}^{\ensuremath{M}\xspace/2}$ uniformly at random from $[0, 2\pi)$) and the other $\ensuremath{M}\xspace/2$ eigenvalues as the complex conjugates of these values. For a given RIP conditioning $\ensuremath{\delta}\xspace$, failure probability $\ensuremath{N}\xspace^{-\log^4 \ensuremath{N}\xspace} \le \eta \le \frac{1}{e}$, and coherence $\mu_S\left(\ensuremath{\bm{\Psi}}\xspace\right)$ as defined as in Equation~\eqref{eqn:multico}, if 
\begin{gather*}
	\ensuremath{M}\xspace \geq C\frac{\ensuremath{K}\xspace}{\ensuremath{\delta}\xspace^2}\mu_S^2\left(\ensuremath{\bm{\Psi}}\xspace\right)\log^{5}\left(\ensuremath{N}\xspace\ensuremath{L}\xspace\right) \log(\eta^{-1}), \nonumber
\end{gather*}
then $\ensuremath{\bm{{A}}}\xspace$ satisfies RIP-$(2\ensuremath{K}\xspace,\ensuremath{\delta}\xspace)$ with probability exceeding $1-\eta$ for a universal constant $C$.
\end{theorem}
The proof of Theorem~\ref{thm:STMmulti} is provided in Appendix~\ref{app:STMmulti}. Note that when \ensuremath{L}\xspace = 1, Theorem~\ref{thm:STMmulti} reduces to Theorem~\ref{thm:STMwithZ}. In this result we see that that the number of nodes relies only linearly on the underlying dimensionality (\ensuremath{K}\xspace) and poly-logarithmically on the total size of the input $(\ensuremath{N}\xspace\ensuremath{L}\xspace)$. This means that under favorable coherence and sparsity conditions on the input, the network can again have STM capacities that are higher than the number of nodes in the network. Specifically, showing that \ensuremath{\bm{{A}}}\xspace satisfies the RIP property, Theorem~\ref{thm:STMmulti} ensures that standard recovery guarantees from the sparse inference literature hold. In particular, any \ensuremath{K}\xspace-sparse input is recoverable from the network state at time $N$ up to the error bound of Equation~\eqref{eqn:DecRec} by solving the $\ell_1$-regularized least-squares optimization of Equation~\eqref{eqn:l1min}. 


\subsection{Low Rank Multiple Inputs}

Next we consider the case of a very different type of low-dimensional structure where the input signals are correlated but not necessarily sparse.  Specifically,  in this setting we assume that the inputs arise from a process where $\ensuremath{R}\xspace$ prototypical signals combine linearly to form the various input streams. Such a signal structure could arise, for instance, due to correlations between input streams at spatially neighboring locations. A number of interesting applications display such correlations, including important measurement modalities in neuroscience (e.g. two-photon calcium imaging;~\citealp{RN2,maruyama2014detecting} and neural electrophysiological recordings;~\citealp{berenyi2014large,ahmed2013compressive}), and remote sensing applications (e.g. hyperspectral imagery;~\citealp{zhang2014hyperspectral,veganzones2016hyperspectral}). The applicability of RNNs and machine learning methods to data well described by this low-rank model is also increasingly relevant as there is increasing interest in applying neural network techniques to such data, either for detection~\citep{apthorpe2016automatic}, classification~\citep{chen2016deep,chen2014deep}, or as samplers via variational auto-encoders~\citep{gao2016linear}. In this case, we can write out the input matrix in the reduced form $\ensuremath{\bm{{S}}}\xspace = \ensuremath{\bm{Q}}\xspace\ensuremath{\bm{V}^{\ast}}\xspace$, where $\ensuremath{\bm{V}^{\ast}}\xspace\in\ensuremath{\mathbb{R}}\xspace^{\ensuremath{R}\xspace\times\ensuremath{N}\xspace}$ is the matrix whose rows may represent environmental causes generating the data and $\ensuremath{\bm{Q}}\xspace\in\ensuremath{\mathbb{R}}\xspace^{\ensuremath{L}\xspace\times\ensuremath{R}\xspace}$ represents the mixing matrix that defines the input stream. We will assume both $\ensuremath{L}\xspace \geq \ensuremath{R}\xspace$ and $\ensuremath{N}\xspace \geq \ensuremath{R}\xspace$, meaning that $\ensuremath{\bm{{S}}}\xspace$ is low-rank.  With this model we use a definition of coherence given by:
\begin{gather}
        \mu_L^2 = \ensuremath{R}\xspace^{-1}\sup_{ \omega\in[0,2\pi] }\ensuremath{\left|\left| {\ensuremath{\bm{V}^{\ast}}\xspace\ensuremath{\bm{{f}}}\xspace_{\omega}} \right|\right|}\xspace_2^2. \label{eqn:lowrankco}
\end{gather}
where $\bm{f}_{\omega} = [1, e^{-j\omega},\cdots e^{-j(N-1)\omega}]^T$ is the Fourier vector with frequency $\omega$. This coherence parameter mirrors the coherence used for the sparse-input case. As $\mu_S$ measured the similarity between the measurement vectors and the sparsity basis \ensuremath{\bm{\Psi}}\xspace, $\mu_L$ measures the similarity between the measurements and the left singular vectors of the measured matrix. The intuition here is that measurements that align with the left singular vectors are unlikely to measure significant information about the \ensuremath{\bm{{S}}}\xspace. 

To analyze the STM of the network dynamics with respect to low-rank signal statistics, we leverage the dual certificate approach~\citep{candes2010matrix,candes2011probabilistic,ahmed2013compressive} to derive the following theorem,
\begin{theorem}
	\label{thm:STMlowrank}
	Suppose $\ensuremath{N}\xspace\ensuremath{L}\xspace \ge \ensuremath{M}\xspace$, $\ensuremath{N}\xspace \ge \ensuremath{R}\xspace$, $\ensuremath{N}\xspace \ge O(1)$ and $\ensuremath{L}\xspace \ge O(1)$. Let \ensuremath{\bm{{z}}}\xspace be i.i.d. zero-mean Gaussian random variables with variance $\frac{1}{\ensuremath{M}\xspace}$. For $\ensuremath{M}\xspace$ an even integer, denote the eigenvalues of \ensuremath{\bm{W}}\xspace by $\{e^{j w_m}\}_{m = 1}^{\ensuremath{M}\xspace}$.  Let the first $\ensuremath{M}\xspace/2$ eigenvalues ($\{e^{j w_m}\}_{m = 1}^{\ensuremath{M}\xspace/2}$) be chosen uniformly at random on the complex unit circle (i.e., we chose $\{w_m\}_{m=1}^{\ensuremath{M}\xspace/2}$ uniformly at random from $[0, 2\pi)$) and the other $\ensuremath{M}\xspace/2$ eigenvalues as the complex conjugates of these values.   For a given coherence $\mu_L$ as defined as in Equation~\eqref{eqn:lowrankco}, if 
		\begin{gather*}
			\ensuremath{M}\xspace \geq c\ensuremath{R}\xspace\left(\ensuremath{N}\xspace + \mu_L^2\ensuremath{L}\xspace\right)\log^3(\ensuremath{L}\xspace\ensuremath{N}\xspace), \nonumber 
		\end{gather*}
                then, with probability at least $1-O((\ensuremath{L}\xspace\ensuremath{N}\xspace)^{1-\beta}$, the minimization in Equation~\eqref{eqn:nucnorm0} recovers the rank-$\ensuremath{R}\xspace$ input matrix \ensuremath{\bm{{S}}}\xspace up to the error bound in Equation~\eqref{eqn:NucNormErr}.  
\end{theorem}
The proof of Theorem~\ref{thm:STMlowrank} is in Appendix~\ref{app:lowrank} and follows a golfing scheme to find an inexact dual certificate. In fact, we note that since our architecture is extremely similar mathematically to the architecture in~\citep{ahmed2013compressive}, our proof is also very similar. The main difference is that due to the unbounded nature of our distributions (i.e. the feed-forward vectors $\bm{Z}$ are Gaussian random variables) and the fact that our Fourier vectors are on the unit circle (rather than gridded), we can consider our proof as a generalization of the proof in~\citep{ahmed2013compressive}. 

Theorem~\ref{thm:STMlowrank} is qualitatively similar to Theorem~\ref{thm:STMmulti} in the way the STM capacity scales.  In this case, the bound still scales linearly with the information rate as captured by the number of elements in the left and right matrices that compose \ensuremath{\bm{{S}}}\xspace: $\ensuremath{R}\xspace\ensuremath{N}\xspace + \ensuremath{R}\xspace\ensuremath{L}\xspace$. Interestingly, due to the left singular vectors interacting with the measurement operator first, the coherence term only affects the portion of the bound related to the number of elements in \ensuremath{\bm{Q}}\xspace. Additionally, as before, the number of total inputs $\ensuremath{L}\xspace\ensuremath{N}\xspace$ only impacts the bound poly-logarithmically. 


\section{Simulation}
\label{sec:sims}

To empirically verify that these theoretical STM scaling laws are representative of the empirical behavior, we generated a number of random networks and evaluated the recovery of (sparse or low-rank) input sequences in the presence of noise. For each simulation we generate a $\ensuremath{M}\xspace \times \ensuremath{M}\xspace$ random orthogonal connectivity matrix $\bm{W}$\footnote{Orthogonal connectivity matrices were obtained by running an orthogonalization procedure on a random Gaussian matrix.} and a $\ensuremath{M}\xspace\times\ensuremath{L}\xspace$ random Gaussian feed-forward matrix $\bm{Z}$. In both cases we fixed the number of inputs to $\ensuremath{L}\xspace = 40$ and the number of time-steps to $\ensuremath{N}\xspace = 100$ while varying the network size \ensuremath{M}\xspace and underlying dimensionality of the input (i.e., the sparsity level or the input matrix rank). For the sparse input simulations, inputs were chosen with a uniformly random support pattern with random Gaussian values on the support. For low-rank simulations, the right singular vectors were chosen to be Gaussian random vectors, and the left singular values were chosen at random from a number of different basis sets. 

In Figure~\ref{fig:SPsims} we show the relative mean-squared error of the input recovery as a function of the sparsity-to-network size ratio $\rho = \ensuremath{K}\xspace/\ensuremath{M}\xspace$ and the network size-to-input ratio $\gamma = \ensuremath{M}\xspace/\ensuremath{N}\xspace\ensuremath{L}\xspace$. Each pixel value represents the average recovery relative mean-squared error (rMSE), as calculated by 
\begin{gather*}
        \mbox{RMSE} = \frac{\left\|\widehat{\bm{s}} - \bm{s}\right\|_2^2}{\left\|\bm{s}\right\|_2^2}, \nonumber
\end{gather*}
over 20 randomly generated trials with a noise level of $\|\bm{\epsilon}\|_2\approx 0.01$. We show results for recovery of three different types of sparse signals: signals sparse in the canonical basis, signals sparse in a Haar wavelet basis, and signals sparse in a discrete cosine transform (DCT) basis. As our theory predicts, for canonical- and Haar wavelet-sparse signals the network has very favorable STM capacity results.  The fact that the capacity achieves $\ensuremath{M}\xspace < \ensuremath{N}\xspace\ensuremath{L}\xspace$ is demonstrated by the area left of the $\ensuremath{M}\xspace = \ensuremath{N}\xspace\ensuremath{L}\xspace$ point ($\gamma=1$) where the signal is recovered with high accuracy. Likewise, for the DCT-sparse signals we find that the inputs are never recovered well for any $\ensuremath{M}\xspace < \ensuremath{N}\xspace\ensuremath{L}\xspace$.  This behavior is also predicted by our theory because of the unfavorable coherence properties of the DCT basis. 

\begin{figure*}[t]
        \centering
        \subfigure[RMSE: No basis]{%
        \includegraphics[width=0.25\textwidth]{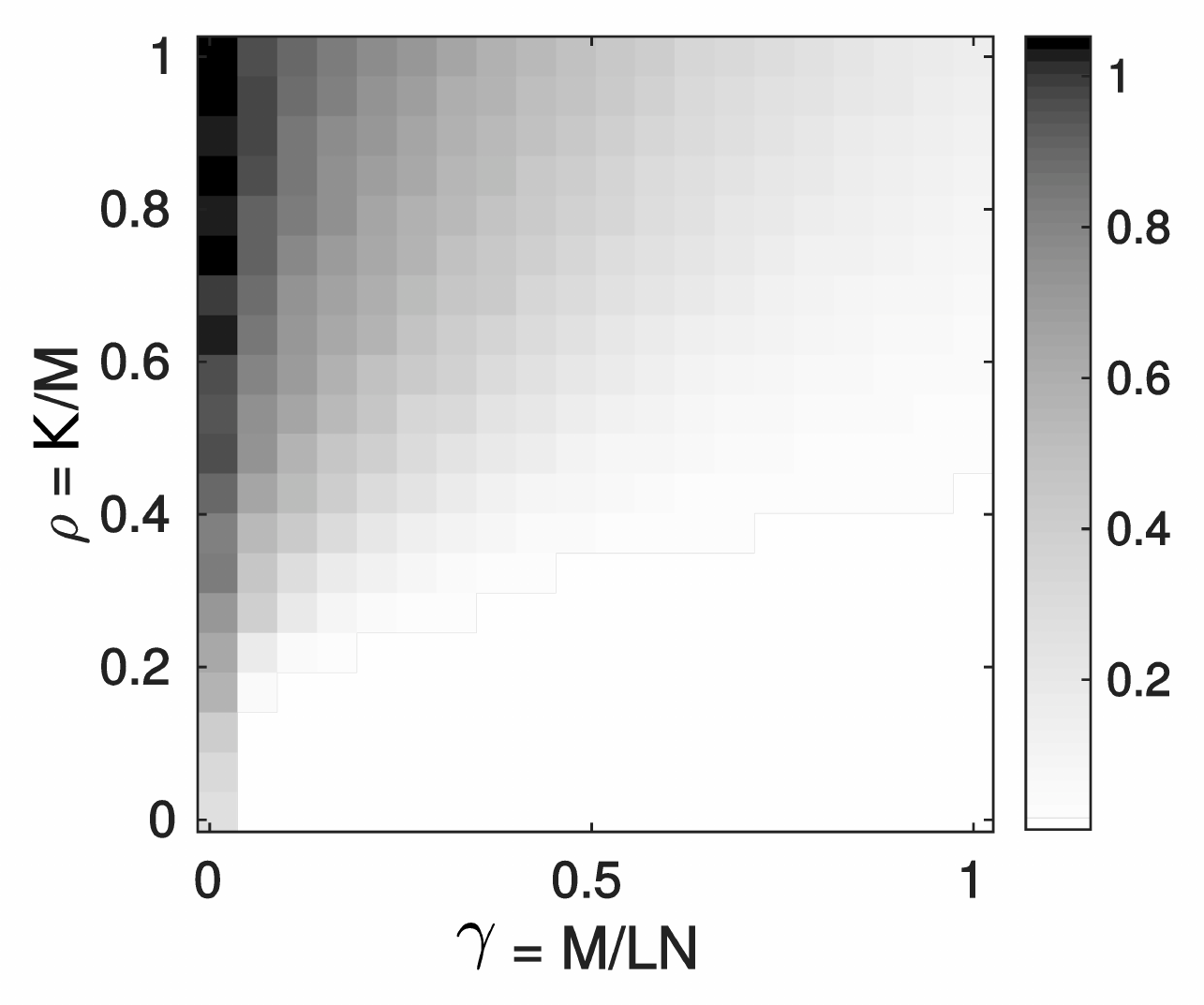}
        \label{fig:SconRMSE}}
        \quad
        \subfigure[RMSE: Haar wavelet basis]{%
        \includegraphics[width=0.25\textwidth]{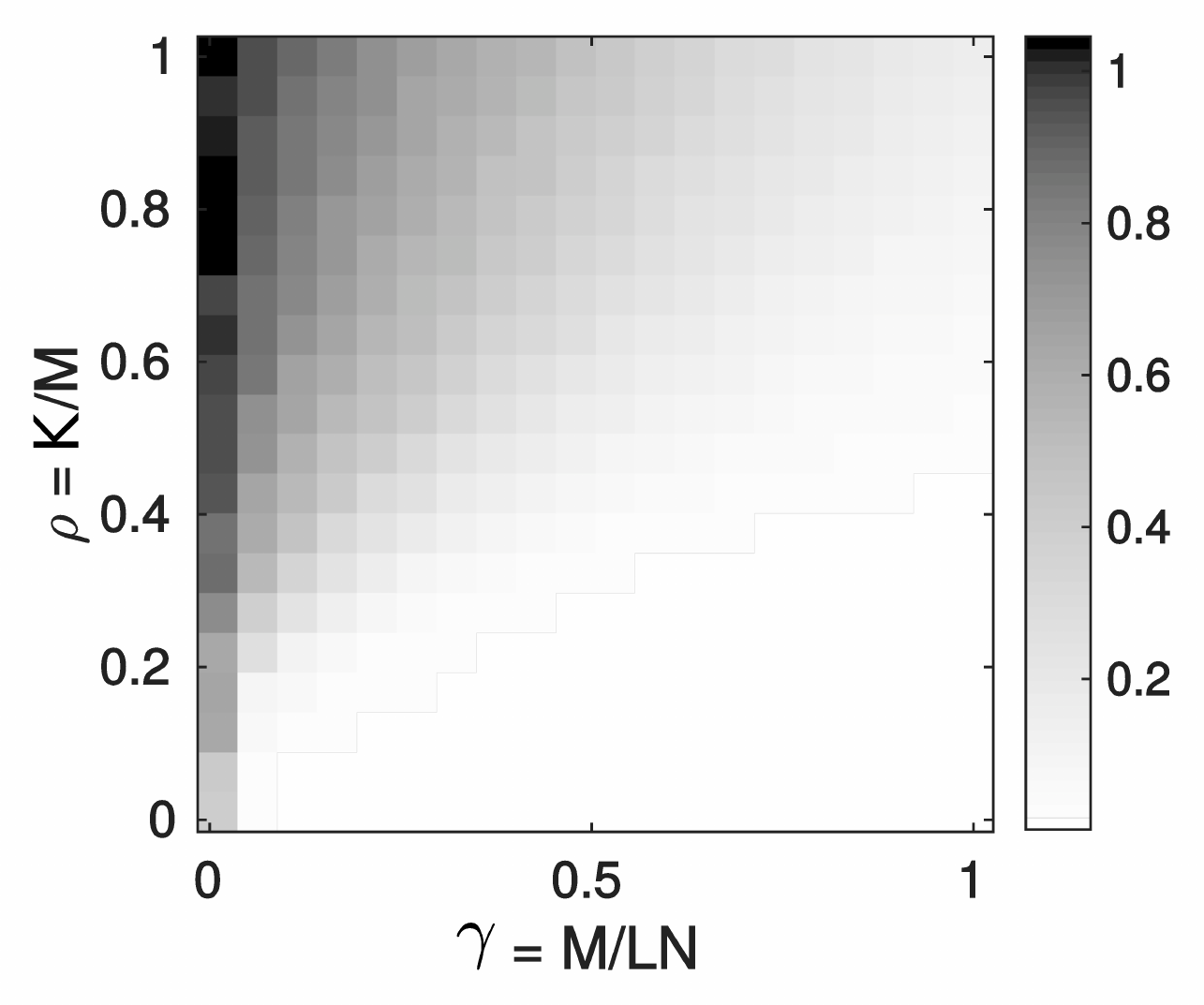}
        \label{fig:ShaarRMSE}}
        \quad
        \subfigure[RMSE: DCT basis]{%
        \includegraphics[width=0.25\textwidth]{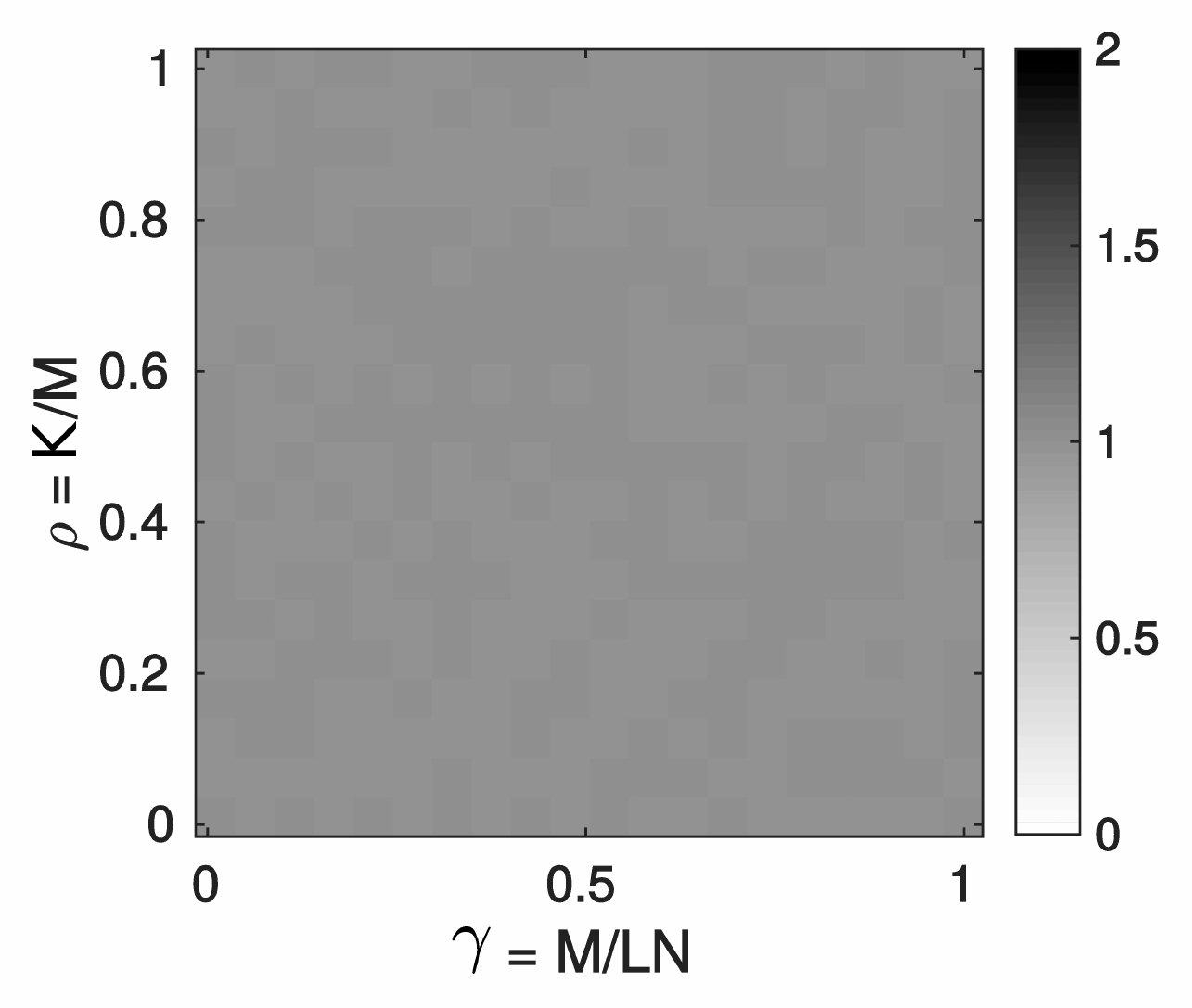}
        \label{fig:SdctRMSE}}
        \caption{ESNs can have high STM capacity for multidimensional sparse inputs. Relative mean-squared error (rMSE) of the recovery for canonical sparse inputs~\subref{fig:SconRMSE} and Haar wavelet-sparse inputs~\subref{fig:ShaarRMSE} is very low for a range of sparsity and network sizes satisfying $\ensuremath{M}\xspace < \ensuremath{L}\xspace\ensuremath{N}\xspace$. The rMSE for DCT-sparse inputs~\subref{fig:SdctRMSE}, as predicted by our theoretical results, remains high (approximately 100\% error). }
        \label{fig:SPsims} 
\end{figure*}

For the low-rank trials we see that recovery of low-rank inputs for a range of $\ensuremath{M}\xspace < \ensuremath{L}\xspace\ensuremath{N}\xspace$ is possible as predicted by the theoretical results.  As with the sparse input case we consider three types of low-rank inputs. Instead of changing the sparsity basis, however, we change the right singular vectors \ensuremath{\bm{V}}\xspace of the low-rank input matrix \ensuremath{\bm{{S}}}\xspace. We explore the cases where the elements of \ensuremath{\bm{V}}\xspace are chosen from the canonical basis, the haar basis and the DCT basis. These results are shown in Figure~\ref{fig:LRsims} with plots similar to those in Figure~\ref{fig:SPsims}, but with only showing the range $\gamma<0.5$ to reduce computational time.  As our theory predicts, the recovery of inputs with canonical- and Haar right singular vectors is more accurate for a larger range of $\rho,\gamma$ pairs than the inputs with DCT right singular vectors. 

\begin{figure*}[t]
        \centering
        \subfigure[RMSE: No basis]{
        \includegraphics[width=0.25\textwidth]{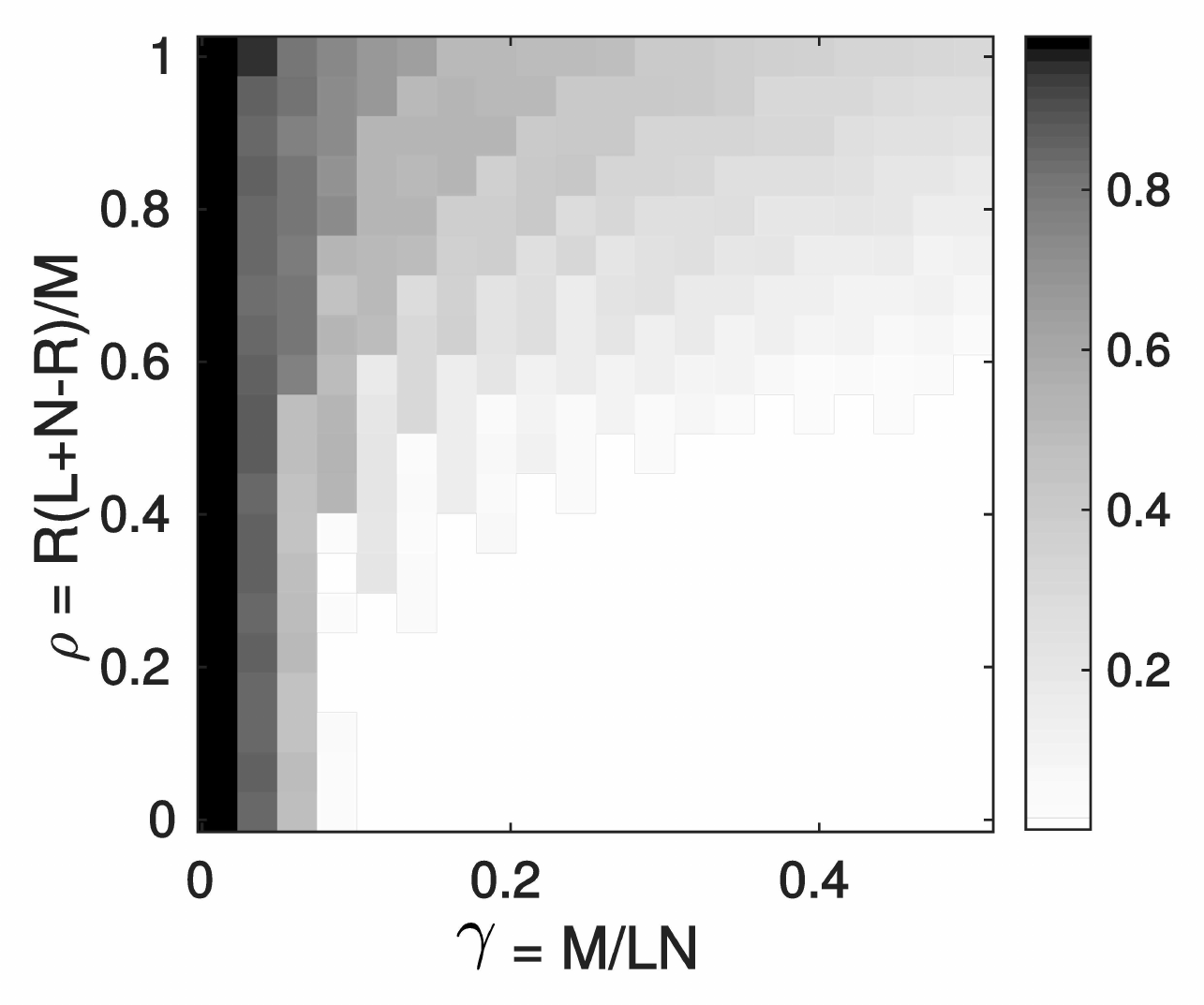}
        \label{fig:LRconRMSE}}
        \quad
        \subfigure[RMSE: Haar wavelet basis]{
                \includegraphics[width=0.25\textwidth]{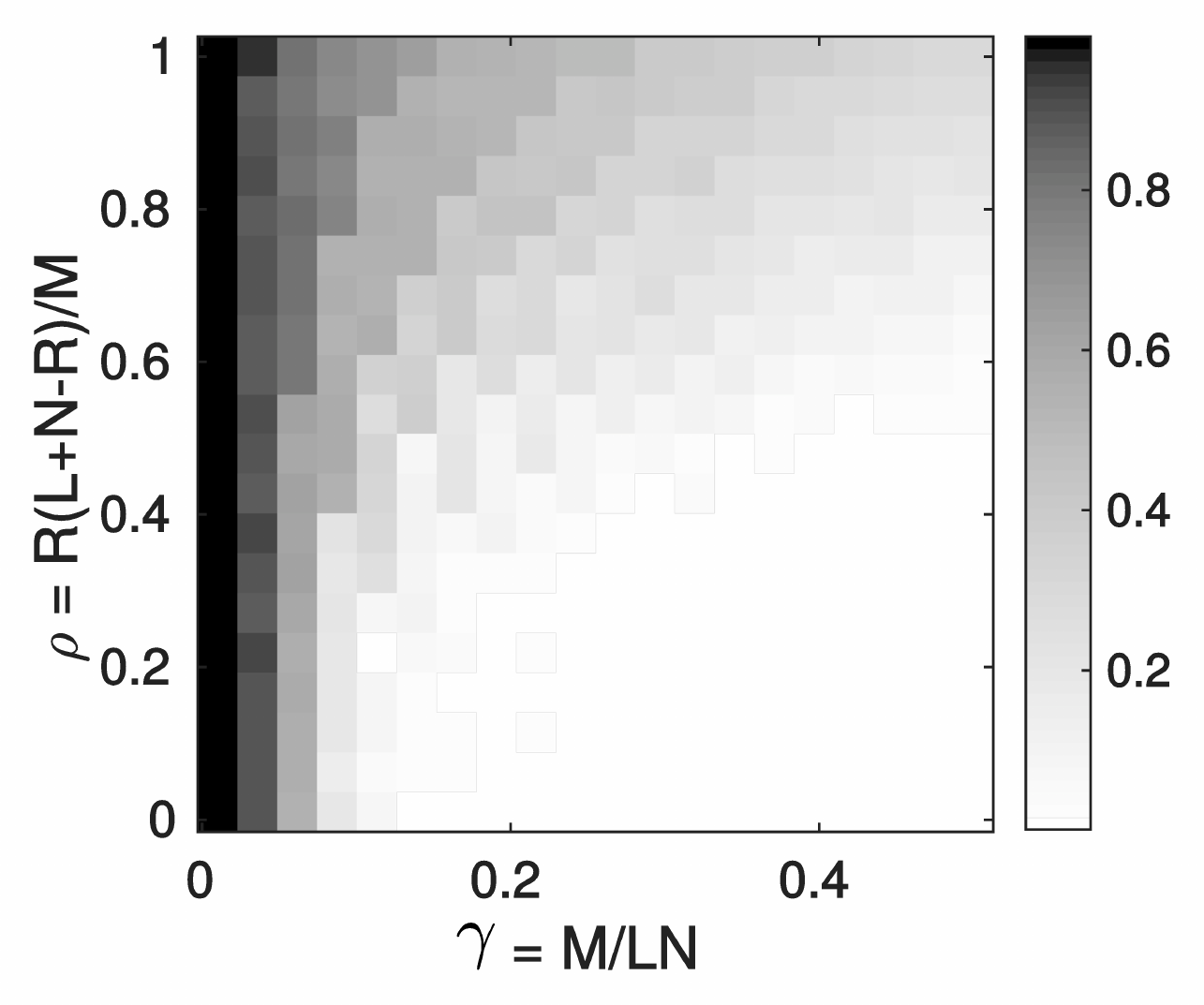}
        \label{fig:LRhaarRMSE}}
        \quad
        \subfigure[RMSE: DCT basis]{
        \includegraphics[width=0.25\textwidth]{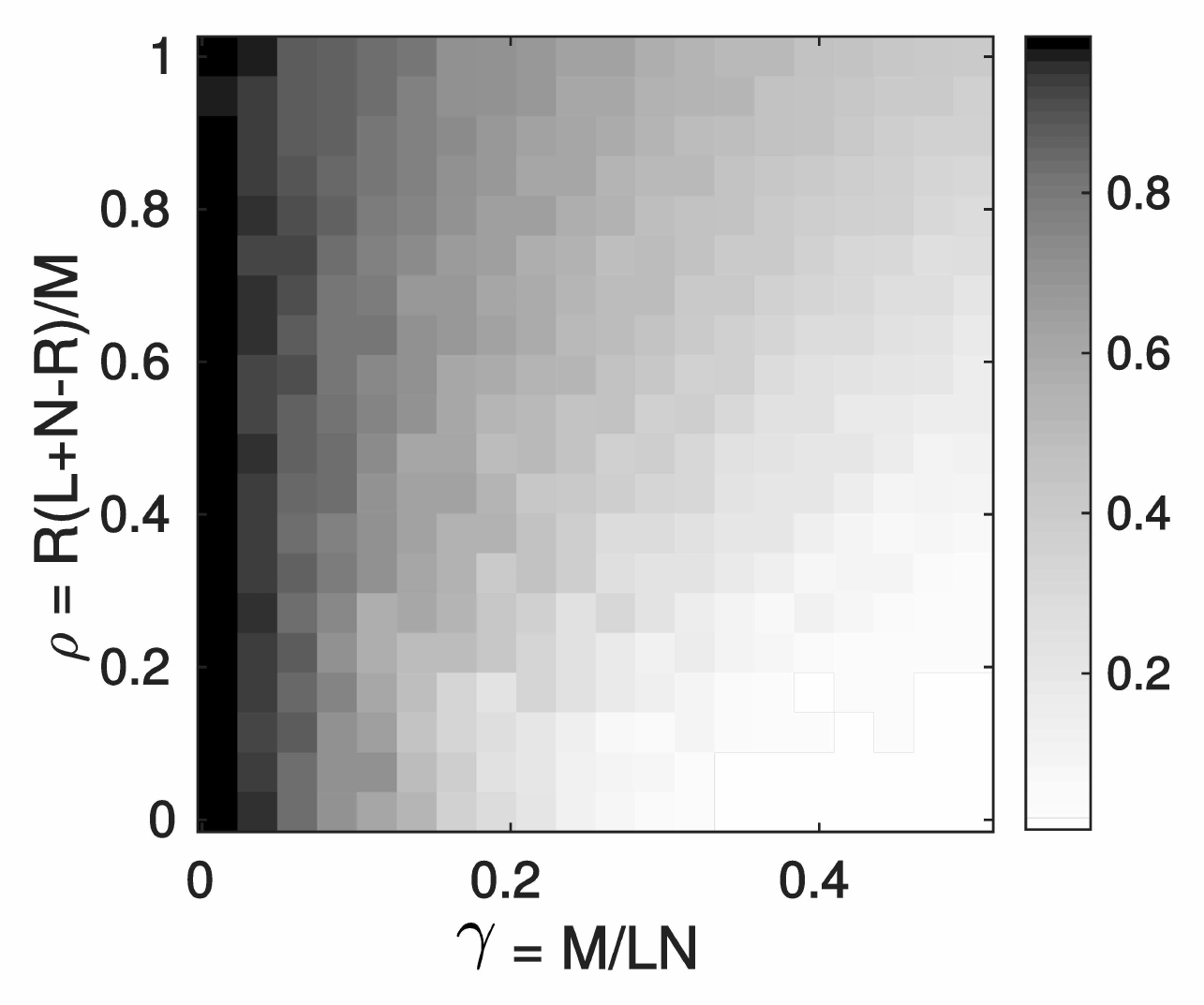}
        \label{fig:LRdctRMSE}}
        \caption{ESNs can have high STM capacity for multidimensional inputs with low-rank structure. The rMSE of the input recovery for inputs with canonical right singular vectors~\subref{fig:LRconRMSE} and Haar wavelets for right singular vectors~\subref{fig:LRhaarRMSE} is very low for a range of input rank and network sizes satisfying $\ensuremath{M}\xspace < \ensuremath{L}\xspace\ensuremath{N}\xspace$. The rMSE for inputs with DCT right singular vectors~\subref{fig:LRdctRMSE} remains high, a behavior predicted by our theoretical results.}
        \label{fig:LRsims}
\end{figure*}

\section{Conclusions}

Determining the fundamental limits of memory in recurrent networks is critical to understanding their behavior in machine learning tasks. In this work we show that randomly connected echo-state networks can exploit the low-dimensional structure in multidimensional input streams to achieve very high short-term memory capacity. Specifically, we show non-asymptotic bounds on recovery error for input sequences that have underlying low-dimensional statistics described by either joint sparsity or  low-rank properties (with no sparsity). For multiple sparse inputs, we find that the network size must be linearly proportional to the input sparsity and only logarithmically dependent on the total input size  (a combination of the input length and number of inputs). For inputs with low-rank structure, we find a similar dependency where the network size depends linearly on the underlying dimension of the inputs (the product of the input rank with the input length and the number of inputs) and logarithmically on the total input size. Both results continue to demonstrate that ESNs can have STM capacities much larger than the network size.

These results are a significant (conceptual and technical) generalization over previous work that provided theoretical guarantees in the case of a single sparse input~\citep{charles2014short}.  While the linear ESN structure is a simplified model, rigorous analysis of these networks has remained elusive due to the recurrence itself.  These results isolate the properties of the transient dynamics due to the recurrent connections, and may provide one foundation for which to explore the analysis of other complex network attributes such as nonlinearities and spiking properties. We also note that knowledge of how well neural networks compresses structured signals could indicate methods to pick the size of recurrent network layers. Specifically, if a task is thought to require a certain time-frame of an input signal, the overall sparsity (or rank) of the signal in that time-frame can be used in conjunction with the length of that time-frame to give a lower bound for the required number of nodes in the recurrent layer of the network. 

While the current paper is restricted to orthogonal connectivity matrices, previous work~\citep{charles2014short} has shown that a number of network structures can satisfy these criteria, including some types of small-world network topologies. Additionally, we explored here low-rank and sparse inputs separately. The methods we have used to prove Theorem~\ref{thm:STMlowrank}, however, have also been used to analyze recovery signals with other related structures (e.g. matrices that can be decomposed into the sum of a sparse and low-rank matrix) from compressive measurements~\citep{candes2011tight}. Our bounds presented here therefore could open up avenues for similar analysis of other low-dimensional signal classes. 

While the context of this paper is focused on the role of recurrent networks as a tool in machine learning tasks, these results may also lead to a better understanding of the STM properties in other networked systems (e.g., social networks, biological networks, distributed computing, etc.).  With respect to the literature relating recurrent ESN and liquid-state machines to working memory in biological systems, the notion of sparsity has much of the same flavor as the concept of \emph{chunking}~\citep{gobet2001chunking}. Chunking is the concept of humans learning to remember items in highly correlated groups rather than remembering items individually as a way of artificially increasing their working memory. Similarly, the use of sparsity bases allow RNNs to `chunk' items according to the basis elements. Thus, each basis counts only as one item (the true underlying sparsity) and the network needs only store these elements rather than storing every input separately.


\subsection*{Acknowledgments}
The authors are grateful to S. Bahmani, A. Ahmed and J. Romberg for valuable discussions related to this work. This work was supported in part by ONR grant N00014-15-1-2731 and NSF grants CCF-0830456 and CCF-1409422.

\renewcommand{\thesection}{A}
\section{Appendix}

\subsection{RIP for Multiple Gaussian Feed-forward Vectors}
\label{app:STMmulti}

In this appendix we prove Theorem~\ref{thm:STMmulti}, showing that the matrix representing the network evolution with \ensuremath{L}\xspace inputs and \emph{i.i.d.} Gaussian feed-forward vectors satisfies the RIP. Recall that we have $\bm{x}[N] = \bm{A}\widetilde{\ensuremath{\bm{{s}}}\xspace}$, where $\bm{A}\in\mathbb{R}^{\ensuremath{M}\xspace\times \ensuremath{N}\xspace\ensuremath{L}\xspace}$ is derived in Section~\ref{sec:stm} and that $\widetilde{\ensuremath{\bm{{s}}}\xspace}$ (the vectorization of $\bm{S}$) is sparse with respect to the basis $\Psi$, meaning that there is a $K$-sparse signal $\bm{a}$ such that $\widetilde{\ensuremath{\bm{{s}}}\xspace} =\Psi \bm{a}$. Similar to~~\citep{charles2014short}, this proof is based on showing conditions on \ensuremath{M}\xspace such that \ensuremath{\bm{{A}}}\xspace satisfies the RIP with respect to \ensuremath{\bm{\Psi}}\xspace, i.e. 
\begin{gather*}
(1-\ensuremath{\delta}\xspace)\ensuremath{\left|\left| {\bm{a}} \right|\right|}\xspace_2^2\le \ensuremath{\left|\left| {\ensuremath{\bm{{A}}}\xspace\ensuremath{\bm{\Psi}}\xspace\bm{a}} \right|\right|}\xspace_2^2 \le (1+\ensuremath{\delta}\xspace)\ensuremath{\left|\left| {\bm{a}} \right|\right|}\xspace_2^2, \nonumber
\end{gather*}
holds with high probability for all \ensuremath{K}\xspace-sparse $\bm{a}$. This is equivalent to bounding the following probability of the event
\begin{gather} 
	\ensuremath{\left|\left| {(\ensuremath{\bm{{A}}}\xspace\ensuremath{\bm{\Psi}}\xspace)^{H}\ensuremath{\bm{{A}}}\xspace\ensuremath{\bm{\Psi}}\xspace - I} \right|\right|}\xspace_\ensuremath{K}\xspace \le \ensuremath{\delta}\xspace. \label{eqn:eqvnt_bound}
\end{gather}
where the norm $\ensuremath{\left|\left| {\ensuremath{\bm{{A}}}\xspace} \right|\right|}\xspace_\ensuremath{K}\xspace$ is defined as 
\[ \ensuremath{\left|\left| {\ensuremath{\bm{{A}}}\xspace} \right|\right|}\xspace_K:=\sup_{\bm{y}\ is\ K-sparse}\frac{\bm{y}^{H}\ensuremath{\bm{{A}}}\xspace\bm{y}}{\ensuremath{\left|\left| {\bm{y}} \right|\right|}\xspace_2^2}.\]

First we bound the expectation of $\ensuremath{\left|\left| {(\ensuremath{\bm{{A}}}\xspace\ensuremath{\bm{\Psi}}\xspace)^{H}\ensuremath{\bm{{A}}}\xspace\ensuremath{\bm{\Psi}}\xspace - I} \right|\right|}\xspace_\ensuremath{K}\xspace$, and use the result to bound the tail probability of the event~\eqref{eqn:eqvnt_bound}.

\begin{proof}
	\subsubsection{Expectation}
First we let
\begin{gather*}
	\widehat{\ensuremath{\bm{{A}}}\xspace} = \left[ \begin{matrix} \ensuremath{\widetilde{\bm{Z}}}\xspace_1\ensuremath{\bm{F}}\xspace & \ensuremath{\widetilde{\bm{Z}}}\xspace_2F & \cdots & \ensuremath{\widetilde{\bm{Z}}}\xspace_\ensuremath{L}\xspace\ensuremath{\bm{F}}\xspace \end{matrix}\right]\ensuremath{\bm{\Psi}}\xspace, \nonumber
\end{gather*}
Since $(\ensuremath{\bm{{A}}}\xspace\ensuremath{\bm{\Psi}}\xspace)^{H}\ensuremath{\bm{{A}}}\xspace\ensuremath{\bm{\Psi}}\xspace=\widehat{\ensuremath{\bm{{A}}}\xspace}^{H}\widehat{\ensuremath{\bm{{A}}}\xspace}$, then for any $\bm{a}\in\ensuremath{\mathbb{R}}\xspace^{\ensuremath{N}\xspace\ensuremath{L}\xspace}$, $\ensuremath{\left|\left| {\ensuremath{\bm{{A}}}\xspace\ensuremath{\bm{\Psi}}\xspace \bm{a}} \right|\right|}\xspace_2=\ensuremath{\left|\left| {\widehat{\ensuremath{\bm{{A}}}\xspace}\bm{a}} \right|\right|}\xspace_2$. Therefore we only need to prove that
\[ \ensuremath{\left|\left| {\widehat{\ensuremath{\bm{{A}}}\xspace}^{H}\widehat{\ensuremath{\bm{{A}}}\xspace}-\bm{I}} \right|\right|}\xspace_\ensuremath{K}\xspace\le \ensuremath{\delta}\xspace, \]
holds with high probability when \ensuremath{M}\xspace is large enough. Let $V_i^{\rm H}$ denote the $i$th row of $\widehat{\ensuremath{\bm{{A}}}\xspace}$, i.e.,
\[ V_i^{\rm H}=\sum_{l=1}^{\ensuremath{L}\xspace}\tilde{\bm{z}}_{i,l}\ensuremath{\bm{F}}\xspace_{i}^{H}\ensuremath{\bm{\Psi}}\xspace_i. \]

Let $\bm{B}_1=\sum_{i=1}^{\ensuremath{M}\xspace/2}V_iV_i^{H}-\bm{I}/2$ and $\bm{B}_2=\sum_{i=\ensuremath{M}\xspace/2+1}^{\ensuremath{M}\xspace}V_iV_i^{H}-\bm{I}/2$; it is easy to check that $\bm{B}_1$ and $\bm{B}_2$ are complex conjugates, and that $\widehat{\ensuremath{\bm{{A}}}\xspace}^{H}\widehat{\ensuremath{\bm{{A}}}\xspace}-\bm{I}=\bm{B}_1+\bm{B}_2$. Therefore $\ensuremath{\left|\left| {\widehat{\ensuremath{\bm{{A}}}\xspace}^{H}\widehat{\ensuremath{\bm{{A}}}\xspace}-\bm{I}} \right|\right|}\xspace_\ensuremath{K}\xspace\le 2\ensuremath{\left|\left| {\bm{B}_1} \right|\right|}\xspace_\ensuremath{K}\xspace$, and we only need to show that $\ensuremath{\left|\left| {\bm{B}_1} \right|\right|}\xspace_\ensuremath{K}\xspace \le \ensuremath{\delta}\xspace$ with high probability when \ensuremath{M}\xspace is large enough. We can easily see that when $i \neq j$, $1 \le i,j\le \ensuremath{M}\xspace/2$, $V_i^{H}$ and $V_j^{H}$ are independent.

First we show that $\ensuremath{\mathbb{E}\left[{\ensuremath{\left|\left| {\bm{B}_1} \right|\right|}\xspace_\ensuremath{K}\xspace}\right]}\xspace$ is small with high probability when \ensuremath{M}\xspace is large enough. We then show that $\ensuremath{\left|\left| {\bm{B}_1} \right|\right|}\xspace_\ensuremath{K}\xspace$ is concentrated around its mean with high probability when \ensuremath{M}\xspace is large enough. By Lemma 6.7 in \cite{RAH:2010}, for a Rademacher sequence $\epsilon_i$, $i=1,...,\ensuremath{M}\xspace/2$,  we have
\begin{gather*}
	\ensuremath{\mathbb{E}\left[{\ensuremath{\left|\left| {\bm{B}_1} \right|\right|}\xspace_\ensuremath{K}\xspace}\right]}\xspace  = \ensuremath{\mathbb{E}\left[{ \ensuremath{\left|\left| {\sum_{i=1}^{\ensuremath{M}\xspace/2}(V_iV_i^{H}-\frac{1}{\ensuremath{M}\xspace}\bm{I})} \right|\right|}\xspace_\ensuremath{K}\xspace}\right]}\xspace \le  2\ensuremath{\mathbb{E}\left[{\ensuremath{\left|\left| {\sum_{i=1}^{\ensuremath{M}\xspace/2}\epsilon_iV_iV_i^{H}} \right|\right|}\xspace_\ensuremath{K}\xspace}\right]}\xspace. \nonumber
\end{gather*}

We can now apply Lemma 8.2 from \cite{RAH:2010}, giving us
\begin{eqnarray}
	\ensuremath{\mathbb{E}\left[{\ensuremath{\left|\left| {\bm{B}_1} \right|\right|}\xspace_\ensuremath{K}\xspace}\right]}\xspace & \leq & 2\ensuremath{\mathbb{E}\left[{\ensuremath{\left|\left| {\sum_{i=1}^{\ensuremath{M}\xspace/2}\epsilon_iV_iV_i^{H}} \right|\right|}\xspace_\ensuremath{K}\xspace}\right]}\xspace \nonumber\\
	&\leq & 2\mathbb{E}\left[\mathbb{E}\left[C_0 V_{\max} \sqrt{K}\log{(100K)}\sqrt{\log{(4NL)}\ln{(5M)}} \right.\right.\nonumber \\
        & & \qquad\qquad\qquad \left.\left.\cdot\sqrt{\left|\left| {\sum_{i=1}^{M/2}V_iV_i^{H}} \right|\right|}_K \ | V_i, i=1,...,M/2 \right]\right] \nonumber\\
	&\leq & \sqrt{C_1\ensuremath{K}\xspace\log^4(\ensuremath{N}\xspace\ensuremath{L}\xspace)}\sqrt{\ensuremath{\mathbb{E}\left[{V_{\max}^2}\right]}\xspace\ensuremath{\mathbb{E}\left[{\ensuremath{\left|\left| {\bm{B}_1} \right|\right|}\xspace_\ensuremath{K}\xspace+\frac{1}{2}}\right]}\xspace}, \label{eqn:Bineq}
\end{eqnarray}
where the last inequality results due to $25\ensuremath{K}\xspace\leq\ensuremath{N}\xspace\ensuremath{L}\xspace$, $5\ensuremath{M}\xspace\leq4\ensuremath{N}\xspace\ensuremath{L}\xspace$, the  Cauthy-Schwarz inequality and the triangle inequality, and
$V_{\max}=\max_{1\le l \le \ensuremath{M}\xspace/2}\ensuremath{\left|\left| {V_l} \right|\right|}\xspace_\infty.$
Note that the $p$th element of $V_l^{H}$ can be written
\[ V_l^{H}(p)=\sum_{i=1}^\ensuremath{L}\xspace\tilde{\ensuremath{\bm{{z}}}\xspace}_{l,i}\ensuremath{\bm{F}}\xspace_l^{H}\ensuremath{\bm{\Psi}}\xspace_i(p). \]
and since we know that
\[ \sum_{i=1}^L|\ensuremath{\bm{F}}\xspace_l^{H}\ensuremath{\bm{\Psi}}\xspace_i(p)|^2 \le \sum_{i=1}^L\ensuremath{\left|\left| {\ensuremath{\bm{\Psi}}\xspace_i(p)} \right|\right|}\xspace_2^2\mu^2(\ensuremath{\bm{\Psi}}\xspace)=\mu^2(\ensuremath{\bm{\Psi}}\xspace), \]
and
\[ V_{\max}=\max_{\substack{1\le l\le \ensuremath{M}\xspace/2\\1\le p\le \ensuremath{N}\xspace\ensuremath{L}\xspace}}|V_l(p)| \]
we can now use Corollary~\ref{tail_exp_max}. Setting $Q=\ensuremath{M}\xspace\ensuremath{N}\xspace\ensuremath{L}\xspace/2$, $\mu_0=\mu(\ensuremath{\bm{\Psi}}\xspace)$ in Corollary \ref{tail_exp_max} yields
\begin{gather}
	\ensuremath{\mathcal{P}\left[{ V_{\max}^2 > \frac{\mu^2(\ensuremath{\bm{\Psi}}\xspace)}{\ensuremath{M}\xspace}\log{\frac{\ensuremath{M}\xspace\ensuremath{N}\xspace\ensuremath{L}\xspace}{2\eta}} }\right]}\xspace\le \eta,  \label{v_max_tail}
\end{gather}
and
\begin{gather}
	\ensuremath{\mathbb{E}\left[{V_{\max}^2}\right]}\xspace  \le \frac{\mu^2(\ensuremath{\bm{\Psi}}\xspace)}{\ensuremath{M}\xspace}(\log{\frac{\ensuremath{M}\xspace\ensuremath{N}\xspace\ensuremath{L}\xspace}{2}}+1) \le  C_2\frac{\mu^2(\ensuremath{\bm{\Psi}}\xspace)}{\ensuremath{M}\xspace}\log{(\ensuremath{N}\xspace\ensuremath{L}\xspace)}.\label{v_max_exp}
\end{gather}
Returning to the inequality in Equation~\eqref{eqn:Bineq}, considering (\ref{v_max_exp}), we have 
\[\ensuremath{\mathbb{E}\left[{\ensuremath{\left|\left| {\bm{B}_1} \right|\right|}\xspace_\ensuremath{K}\xspace}\right]}\xspace\le a\sqrt{\ensuremath{\mathbb{E}\left[{\ensuremath{\left|\left| {\bm{B}_1} \right|\right|}\xspace_\ensuremath{K}\xspace}\right]}\xspace+1},\] 
where $a=\sqrt{C_1C_2\ensuremath{K}\xspace\log^5{(\ensuremath{N}\xspace\ensuremath{L}\xspace)}\mu^2(\ensuremath{\bm{\Psi}}\xspace)/\ensuremath{M}\xspace}$. Then $\ensuremath{\mathbb{E}\left[{\ensuremath{\left|\left| {\bm{B}_1} \right|\right|}\xspace_\ensuremath{K}\xspace}\right]}\xspace \le \frac{a^2}{2}+a\sqrt{\frac{1}{2}+\frac{a^2}{4}}$. When $a \le 1/2$, we get $\ensuremath{\mathbb{E}\left[{\ensuremath{\left|\left| {\bm{B}_1} \right|\right|}\xspace_\ensuremath{K}\xspace}\right]}\xspace \le a$. Let $0<{\ensuremath{\bm{{A}}}\xspace}^{\prime} \le a\le1/2$, and we can conclude that when
\begin{gather*}
\ensuremath{M}\xspace \ge\frac{C_3\ensuremath{K}\xspace\mu^2{(\ensuremath{\bm{\Psi}}\xspace)}\log^5{(\ensuremath{N}\xspace\ensuremath{L}\xspace)}}{{\ensuremath{\delta}\xspace}^{\prime}}, \nonumber
\end{gather*}
then
\[ \ensuremath{\mathbb{E}\left[{\ensuremath{\left|\left| {\bm{B}_1} \right|\right|}\xspace_\ensuremath{K}\xspace}\right]}\xspace \le {\ensuremath{\delta}\xspace}^{\prime}. \]

\subsubsection{Tail Bound}

Now we study the tail bound of $\ensuremath{\left|\left| {\bm{B}_1} \right|\right|}\xspace_\ensuremath{K}\xspace$. First we construct a second set of random variables $V_l^{\prime}$, which are independent of $V_l$ and are identically distributed as $V_l$. Additionally, we let
\[ \widetilde{\bm{B}}_1=\sum_{i=1}^{\ensuremath{M}\xspace/2}\left(V_iV_i^{H}-V_i^{\prime}V_i^{\prime{H}}\right), \]
and then according to \cite{charles2014short}, there is
\begin{gather}
	\ensuremath{\mathbb{E}\left[{\ensuremath{\left|\left| {\widetilde{\bm{B}}_1} \right|\right|}\xspace_\ensuremath{K}\xspace}\right]}\xspace\le 2\ensuremath{\mathbb{E}\left[{\ensuremath{\left|\left| {\bm{B}_1} \right|\right|}\xspace_\ensuremath{K}\xspace}\right]}\xspace, \label{exp_B_B_tilde}
\end{gather}
\begin{gather}
	\ensuremath{\mathcal{P}\left[{\ensuremath{\left|\left| {\bm{B}_1} \right|\right|}\xspace_\ensuremath{K}\xspace > 2\ensuremath{\mathbb{E}\left[{\ensuremath{\left|\left| {\bm{B}_1} \right|\right|}\xspace_\ensuremath{K}\xspace}\right]}\xspace+u}\right]}\xspace \le 2\ensuremath{\mathcal{P}\left[{\ensuremath{\left|\left| {\widetilde{\bm{B}}_1} \right|\right|}\xspace_\ensuremath{K}\xspace>u}\right]}\xspace.  \label{tail_B_B_tilde}
\end{gather}
Now since we have
\[ \ensuremath{\left|\left| {V_iV_i^{H}-V_i^{\prime}V_i^{\prime{H}}} \right|\right|}\xspace_\ensuremath{K}\xspace \le 2\max\{\ensuremath{\left|\left| {V_iV_i^{H}} \right|\right|}\xspace_\ensuremath{K}\xspace,\ensuremath{\left|\left| {V_i^{\prime}V_i^{\prime{H}}} \right|\right|}\xspace_\ensuremath{K}\xspace\},  \]
and
\[ \ensuremath{\left|\left| {V_iV_i^{H}} \right|\right|}\xspace_\ensuremath{K}\xspace \le \sup_{y\: is\: \ensuremath{K}\xspace-sparse}\ensuremath{\left|\left| {V_i} \right|\right|}\xspace_\infty^2\frac{\ensuremath{\left|\left| {y} \right|\right|}\xspace_1^2}{\ensuremath{\left|\left| {y} \right|\right|}\xspace_2^2} \le \ensuremath{K}\xspace V_{\max}^2, \]
then we know that
$\max_{1\le i\le \ensuremath{M}\xspace/2}\ensuremath{\left|\left| {V_iV_i^{H}} \right|\right|}\xspace_\ensuremath{K}\xspace\le \ensuremath{K}\xspace V_{\max}^2$ and $\max_{1\le i\le \ensuremath{M}\xspace/2}\ensuremath{\left|\left| {V_i^{\prime}V_i^{{\prime}{H}}} \right|\right|}\xspace_\ensuremath{K}\xspace \le \ensuremath{K}\xspace V_{\max}^{\prime2}$, where $V_{\max}^{\prime}=\max_{1\le l\le \ensuremath{M}\xspace/2}\ensuremath{\left|\left| {V_l^{\prime}} \right|\right|}\xspace_\infty$
then by Equation~\eqref{v_max_tail}, we obtain
\[ \ensuremath{\mathcal{P}\left[{\max_{1\le i\le \ensuremath{M}\xspace/2}\ensuremath{\left|\left| {V_iV_i^{H}} \right|\right|}\xspace_\ensuremath{K}\xspace>\frac{\ensuremath{K}\xspace\mu^2(\ensuremath{\bm{\Psi}}\xspace)}{\ensuremath{M}\xspace}\log{\frac{\ensuremath{M}\xspace\ensuremath{N}\xspace\ensuremath{L}\xspace}{2\eta}} }\right]}\xspace\le \ensuremath{\mathcal{P}\left[{V_{\max}^2>\frac{\mu^2(\ensuremath{\bm{\Psi}}\xspace)}{\ensuremath{M}\xspace}\log{\frac{\ensuremath{M}\xspace\ensuremath{N}\xspace\ensuremath{L}\xspace}{2\eta}} }\right]}\xspace\le \eta. \]

Since the probability theorems depend on bounded random variables, we define ${\bf F}$ to denote the following event
\[ {\bf F}=\left\{\max\left\{\max_{1\le i\le \ensuremath{M}\xspace/2}\ensuremath{\left|\left| {V_iV_i^{H}} \right|\right|}\xspace_\ensuremath{K}\xspace,\max_{1\le i\le \ensuremath{M}\xspace/2}\ensuremath{\left|\left| {V_i^{\prime}V_i^{\prime{H}}} \right|\right|}\xspace_\ensuremath{K}\xspace\right\}\le\frac{\ensuremath{K}\xspace\mu^2(\ensuremath{\bm{\Psi}}\xspace)}{\ensuremath{M}\xspace}\log{\frac{\ensuremath{M}\xspace\ensuremath{N}\xspace\ensuremath{L}\xspace}{2\eta}}\right\}, \]
such that $\ensuremath{\mathcal{P}\left[{{\bf F}^C}\right]}\xspace \le 2\eta$. Furthermore, we define ${\bf I}_F$ as the indicator function of ${\bf F}$, and let $\widehat{\bm{B}}_1=\sum_{i=1}^{\ensuremath{M}\xspace/2}\xi_i\left(V_iV_i^{H}-V_i^{\prime}V_i^{\prime{H}}\right){\bf I}_{F}$, where $\xi=\{\xi_i\}$, $i=1,2,...,\ensuremath{M}\xspace/2$ is a Rademacher sequence and independent of $V_i$. The truncated variable $Y_i = \xi_i\left(V_iV_i^{H}-V_i^{\prime}V_i^{\prime{H}}\right){\bf I}_{F}$ has a symmetric distribution and $\ensuremath{\left|\left| {Y_i} \right|\right|}\xspace_\ensuremath{K}\xspace$ is bounded by $B_{\max}:=\frac{2\ensuremath{K}\xspace\mu^2(\ensuremath{\bm{\Psi}}\xspace)}{\ensuremath{M}\xspace}\ln{\frac{\ensuremath{M}\xspace\ensuremath{N}\xspace\ensuremath{L}\xspace}{2\eta}}$. By Proposition 19 in \cite{Tropp2009a}, we have
\begin{gather}
	\ensuremath{\mathcal{P}\left[{\ensuremath{\left|\left| {\widehat{\bm{B}}_1} \right|\right|}\xspace_\ensuremath{K}\xspace>C_4(u\ensuremath{\mathbb{E}\left[{\ensuremath{\left|\left| {\widehat{\bm{B}}_1} \right|\right|}\xspace_\ensuremath{K}\xspace}\right]}\xspace + tB_{\max})}\right]}\xspace\le e^{-u^2}+e^{-t}, \label{tail_hat}
\end{gather}
for all $u,\ t\ge 1$. Following~\cite{Tropp2009a}, we find that 
\begin{gather}
	\ensuremath{\mathcal{P}\left[{\ensuremath{\left|\left| {\widetilde{\bm{B}}_1} \right|\right|}\xspace_\ensuremath{K}\xspace > v}\right]}\xspace\le \ensuremath{\mathcal{P}\left[{\ensuremath{\left|\left| {\widehat{\bm{B}}_1} \right|\right|}\xspace_\ensuremath{K}\xspace>v}\right]}\xspace+\ensuremath{\mathcal{P}\left[{{\bf F}^C}\right]}\xspace,\label{tail_hat_tilde}
\end{gather}
and
\begin{gather}
	\ensuremath{\mathbb{E}\left[{\ensuremath{\left|\left| {\widehat{\bm{B}}_1} \right|\right|}\xspace_\ensuremath{K}\xspace}\right]}\xspace \le \ensuremath{\mathbb{E}\left[{\ensuremath{\left|\left| {\widetilde{\bm{B}}_1} \right|\right|}\xspace_\ensuremath{K}\xspace}\right]}\xspace. \label{exp_hat_tilde}
\end{gather}
By combining Equations~\eqref{tail_hat}, \eqref{tail_hat_tilde}, and \eqref{exp_hat_tilde}, we get
\begin{gather*}
	\ensuremath{\mathcal{P}\left[{\ensuremath{\left|\left| {\widetilde{\bm{B}}_1} \right|\right|}\xspace_\ensuremath{K}\xspace > C_4(u\ensuremath{\mathbb{E}\left[{\ensuremath{\left|\left| {\widetilde{\bm{B}}_1} \right|\right|}\xspace_\ensuremath{K}\xspace}\right]}\xspace+tB_{\max})}\right]}\xspace\le e^{-u^2}+e^{-t}+2\eta. \nonumber
\end{gather*}
In Equation~\eqref{tail_tilde}, let $\eta<1/e$, $u=\sqrt{\log{\eta^{-1}}}$ and $t=\log{\eta^{-1}}$. These values yield
\begin{gather}
	\ensuremath{\mathcal{P}\left[{\ensuremath{\left|\left| {\widetilde{\bm{B}}_1} \right|\right|}\xspace_\ensuremath{K}\xspace>C_4(\sqrt{\log{\eta^{-1}}}\ensuremath{\mathbb{E}\left[{\ensuremath{\left|\left| {\widetilde{\bm{B}}_1} \right|\right|}\xspace_\ensuremath{K}\xspace}\right]}\xspace+(\log{\eta^{-1}})B_{\max})}\right]}\xspace\le 4\eta. \label{tail_tilde}
\end{gather}
Now we can combine Equations~\eqref{exp_B_B_tilde}, \eqref{tail_B_B_tilde}, and \eqref{tail_tilde} to get
\begin{gather*}
	\ensuremath{\mathcal{P}\left[{\ensuremath{\left|\left| {\bm{B}_1} \right|\right|}\xspace_\ensuremath{K}\xspace > 2\ensuremath{\mathbb{E}\left[{\ensuremath{\left|\left| {\bm{B}_1} \right|\right|}\xspace_\ensuremath{K}\xspace}\right]}\xspace+2C_4\sqrt{\log{\eta^{-1}}}\ensuremath{\mathbb{E}\left[{\ensuremath{\left|\left| {\bm{B}_1} \right|\right|}\xspace_\ensuremath{K}\xspace}\right]}\xspace+C_4(\log{\eta^{-1}})B_{\max}}\right]}\xspace \le 8\eta. \nonumber
\end{gather*}
By choosing
\begin{equation}
\ensuremath{M}\xspace\ge\frac{C_3\ensuremath{K}\xspace\mu^2{(\ensuremath{\bm{\Psi}}\xspace)}\log^5{(\ensuremath{N}\xspace\ensuremath{L}\xspace)}}{{\ensuremath{\delta}\xspace}^{\prime2}}, \label{M_deltaprime}
\end{equation}
where ${\ensuremath{\delta}\xspace}^{\prime}<1/2$, we obtain 
\[ \ensuremath{\mathcal{P}\left[{\ensuremath{\left|\left| {\bm{B}_1} \right|\right|}\xspace_\ensuremath{K}\xspace > 2{\ensuremath{\delta}\xspace}^{\prime}+2C_4{\ensuremath{\delta}\xspace}^{\prime}\sqrt{\log{\eta^{-1}}}+C_5\log{\eta^{-1}}\frac{\ensuremath{\delta}\xspace^{\prime2}\log{(\frac{1}{2}\ensuremath{M}\xspace\ensuremath{N}\xspace\ensuremath{L}\xspace\eta^{-1})}}{\log^5{(\ensuremath{N}\xspace\ensuremath{L}\xspace)}} }\right]}\xspace\le8\eta. \]

We observe now that
\[ \frac{\log(\frac{1}{2}\ensuremath{M}\xspace\ensuremath{N}\xspace\ensuremath{L}\xspace\eta^{-1})}{\log^5{(\ensuremath{N}\xspace\ensuremath{L}\xspace)}} \le \frac{2\log{(\ensuremath{N}\xspace\ensuremath{L}\xspace)}+\log{\eta^{-1}}}{\log^5{(\ensuremath{N}\xspace\ensuremath{L}\xspace)}}, \]
indicating that when $\eta \ge (\ensuremath{N}\xspace\ensuremath{L}\xspace)^{-\log^4{(\ensuremath{N}\xspace\ensuremath{L}\xspace)}}$, i.e., $\log{\eta^{-1}} \le \log^5{(\ensuremath{N}\xspace\ensuremath{L}\xspace)}$, then 
\[ \frac{\log(\frac{1}{2}\ensuremath{M}\xspace\ensuremath{N}\xspace\ensuremath{L}\xspace\eta^{-1})}{\log^5{(\ensuremath{N}\xspace\ensuremath{L}\xspace)}} \le C_6, \],
for a constant $C_6$. This inequality reduces our probability statement to
\[ \ensuremath{\mathcal{P}\left[{\ensuremath{\left|\left| {\bm{B}_1} \right|\right|}\xspace_\ensuremath{K}\xspace>2{\ensuremath{\delta}\xspace}^{\prime}+2C_4{\ensuremath{\delta}\xspace}^{\prime}\sqrt{\log{\eta^{-1}}}+C_5C_6\log{\eta^{-1}}\ensuremath{\delta}\xspace^{\prime2} }\right]}\xspace \le 8\eta. \]
For an arbitrary \ensuremath{\delta}\xspace with $0<\ensuremath{\delta}\xspace<1/2$, we let
\begin{gather}
\ensuremath{\delta}\xspace^{\prime} = \frac{\ensuremath{\delta}\xspace}{2C_4C_7\sqrt{\log{\eta^{-1}}}}, \label{deltaprime}
\end{gather}
resulting in 
\[ 2C_4\ensuremath{\delta}\xspace^{\prime}\sqrt{\log{\eta^{-1}}}=\frac{\ensuremath{\delta}\xspace}{C_7}, \]
\[ 2\ensuremath{\delta}\xspace^{\prime}=\frac{\ensuremath{\delta}\xspace}{C_4C_7\sqrt{\log{\eta^{-1}}}}\le \frac{\ensuremath{\delta}\xspace}{C_4C_7}, \]
\[ C_5C_6\log{\eta^{-1}}\ensuremath{\delta}\xspace^{\prime2}=C_5C_6\frac{\ensuremath{\delta}\xspace^2}{4C_4^2C_7^2}\le \frac{C_5C_6}{8C_4^2C_7^2}\ensuremath{\delta}\xspace. \]
Then we can see that if 
\[ C_7\ge \max\left\{3, \frac{3}{C_4}, \sqrt{\frac{3C_5C_6}{8C_4^2}}\right\}, \]
then
\[ \ensuremath{\mathcal{P}\left[{\ensuremath{\left|\left| {\bm{B}_1} \right|\right|}\xspace_\ensuremath{K}\xspace>\frac{\ensuremath{\delta}\xspace}{3}+\frac{\ensuremath{\delta}\xspace}{3}+\frac{\ensuremath{\delta}\xspace}{3}}\right]}\xspace<8\eta. \] 
Now by plugging (\ref{deltaprime}) into (\ref{M_deltaprime}), we know that when $(\ensuremath{N}\xspace\ensuremath{L}\xspace)^{-\ln^4{(NL)}}\le \eta<1/e$,
there exists a constant $C$ such that when
\begin{gather*}
    \ensuremath{M}\xspace \ge \frac{C\ensuremath{K}\xspace\mu^2{(\ensuremath{\bm{\Psi}}\xspace)}\log^5{(\ensuremath{N}\xspace\ensuremath{L}\xspace)\log{(\eta^{-1})}}}{\ensuremath{\delta}\xspace^2}, \nonumber
\end{gather*}
there is
\begin{gather*}
{\rm P}(\ensuremath{\left|\left| {\bm{B}_1} \right|\right|}\xspace_\ensuremath{K}\xspace>\ensuremath{\delta}\xspace) \le 8\eta,
\end{gather*}
which completes the proof.
\end{proof} 

\subsubsection{Lemmas for Theorem~\ref{thm:STMmulti}}

\begin{lemma}\label{tail_single}
 Suppose we have $n$ complex Gaussian random variables, $z_1,\ z_2, \cdots\ z_n$, $z_i=x_i+{\rm j}y_i$, where $x_i$ and $y_i$ denote the real and imaginary parts of $z_i$. Let $x_i$ and $y_i$ i.i.d Gaussian r.v.'s with mean 0 and variance $1/2M$. $\phi_1,\ \phi_2,\ \cdots\ \phi_n$ r.v.'s which are independent of $z_i$ for all $i$ and satisfy
$$
\sum_{i=1}^{n}|\phi_i|^2\le\mu_0^2,
$$
then for $w=\sum_{i=1}^{n}z_i\phi_i$, 
$$
\ensuremath{\mathcal{P}\left[{|w|^2>u}\right]}\xspace \le e^{-\frac{\ensuremath{M}\xspace u}{\mu_0^2}}.
$$
\end{lemma}

\begin{proof}
We use $x$ and $y$ to denote the real and imaginary parts of $w$, and $a_i$ and $b_i$ to denote the real and imaginary parts of $\phi_i$. We have
\begin{gather*}
  w = \sum_{i=1}^n(a_ix_i-b_iy_i)+{\rm j}\sum_{i=1}^n(a_iy_i+b_ix_i) x+{\rm j}y. \nonumber
\end{gather*}
Then conditioned on $a_i$ and $b_i$, $x$ and $y$ have distribution $\mathcal{N}(0, \frac{1}{2\ensuremath{M}\xspace}\sum_{i=1}^{n}(a_i^2+b_i^2))$.

The next step is to prove the conditional independence of $x$ and $y$. Since
\begin{gather*}
\mbox{Cov}(a_ix_i-b_iy_i, a_iy_i+b_ix_i|a_i,b_i) = \frac{a_ib_i}{2\ensuremath{M}\xspace}-\frac{a_ib_i}{2\ensuremath{M}\xspace} = 0, \nonumber
\end{gather*}
when $i\ne j$, $x_i$, $y_i$ are independent of $x_j$, $y_j$, and $x$ and $y$ are conditionally independent. Thus, conditioned on $a_i$ and $b_i$, 
\[\frac{2\ensuremath{M}\xspace}{\sum_{i=1}^n{(a_i^2+b_i^2)}}|w|^2, \]
is $\chi^2$ distributed.
We use $\chi_2$ to denote a two-degree $\chi^2$ distributed random variable. According to the results on $\chi^2$ distributions, we have
\begin{gather}
	\ensuremath{\mathcal{P}\left[{|w|^2>u|a_i, b_i}\right]}\xspace = \ensuremath{\mathcal{P}\left[{\chi_2>\frac{2\ensuremath{M}\xspace u}{\sum_{i=1}^n{(a_i^2+b_i^2)}} | a_i, b_i }\right]}\xspace  = e^{-\frac{\ensuremath{M}\xspace u}{\sum_{i=1}^n{(a_i^2+b_i^2)}}}  \le e^{-\frac{\ensuremath{M}\xspace u}{\mu_0^2}}. \label{chi_2_tail}
\end{gather}
Noticing that Equation~\eqref{chi_2_tail} holds for all possible values of $a_i$ and $b_i$ completes the proof.
\end{proof}

\begin{corollary}\label{tail_exp_max}
For $Q$ r.v.'s, $w_1,\ w_2,\ \cdots,\ w_Q$, let
$w_i=\sum_{l=1}^n z_{il}\phi_{il}$ and
$z_{il}=x_{il}+{\rm j}y_{il}$, where
$x_{il}$, $y_{il}$, $1\le i\le Q$, $1\le l\le n$ are i.i.d.
Gaussian distributed with mean 0 and variance $1/2\ensuremath{M}\xspace$. Suppose for any $i$, there is
\[ \sum_{l=1}^{n}|\phi_{il}|^2\le\mu_0^2. \]
And let $w_{\max }=\max_{1\le i \le Q}|w_i|$, then for $\eta>0$, we have
\[ \ensuremath{\mathcal{P}\left[{w_{\max}^2>\frac{\mu_0^2}{\ensuremath{M}\xspace}\log{\frac{Q}{\eta}}}\right]}\xspace\le \eta, \]
and
\[ \ensuremath{\mathbb{E}\left[{w_{\rm max}^2}\right]}\xspace \le \frac{\mu_0^2}{\ensuremath{M}\xspace}(\ln{Q}+1). \]
\end{corollary}

\begin{proof}
According to Lemma~\ref{tail_single} and by using union bound, we have
\[ \ensuremath{\mathcal{P}\left[{w_{\max}^2>u}\right]}\xspace\le Qe^{-\frac{\ensuremath{M}\xspace u}{\mu_0^2}}. \]
Let $\eta=Qe^{-\frac{\ensuremath{M}\xspace u}{\mu_0^2}}$ and there is
\[ \ensuremath{\mathcal{P}\left[{w_{\max}^2>\frac{\mu_0^2}{\ensuremath{M}\xspace}\log{\frac{Q}{\eta}}}\right]}\xspace\le \eta. \]

Then we have
\begin{eqnarray*}
	\ensuremath{\mathbb{E}\left[{w_{\max}^2}\right]}\xspace &= & \int_0^\infty\ensuremath{\mathcal{P}\left[{w_{\max}^2>u}\right]}\xspace du\nonumber\\
                            & \le &\int_0^{\frac{\mu_0^2}{\ensuremath{M}\xspace}\ln{Q}}1du+\int_{\frac{\mu_0^2}{\ensuremath{M}\xspace}\log{Q}}^\infty Qe^{-\frac{\ensuremath{M}\xspace u}{\mu_0^2}}du\nonumber\\
  & = & \frac{\mu_0^2}{\ensuremath{M}\xspace}(\log{Q}+1).\nonumber
\end{eqnarray*}
\end{proof}

\subsection{Proof of Low-rank Recovery}
\label{app:lowrank}

In this appendix we prove Theorem~\ref{thm:STMlowrank} where a low-rank input matrix \ensuremath{\bm{{S}}}\xspace can be recovered from the network state $\ensuremath{\bm{{x}}[N]}\xspace$ via nuclear norm optimization~\citep{candes2010power,candes2010matrix,recht2010guaranteed}. To prove this theorem we use the dual certificate approach used to prove similar results in~\citep{ahmed2013compressive,candes2011probabilistic}. In this methodology we seek a certificate \ensuremath{\bm{Y}}\xspace whose projections into and out of the space spanned by the singular vectors of \ensuremath{\bm{{S}}}\xspace are bounded appropriately. Specifically if we consider the singular value decomposition of \ensuremath{\bm{{S}}}\xspace as
\begin{gather*}
	\ensuremath{\bm{{S}}}\xspace = \ensuremath{\bm{Q}}\xspace\ensuremath{\bm{\Sigma}}\xspace\ensuremath{\bm{V}^{\ast}}\xspace \nonumber
\end{gather*}
and we consider the projection \ensuremath{\mathcal{P}_T}\xspace which projects a matrix into the space $T$ spanned by the left and right singular vectors, 
\begin{gather}
	\ensuremath{\mathcal{P}_T\left( \bm{W} \right)}\xspace = \ensuremath{\bm{Q}}\xspace\ensuremath{\bm{Q}^{\ast}}\xspace\bm{W} + \bm{W}\ensuremath{\bm{V}}\xspace\ensuremath{\bm{V}^{\ast}}\xspace - \ensuremath{\bm{Q}}\xspace\ensuremath{\bm{Q}^{\ast}}\xspace\bm{W}\ensuremath{\bm{V}}\xspace\ensuremath{\bm{V}^{\ast}}\xspace \label{eqn:projdef}
\end{gather}
the conditions for the dual certificate are that \ensuremath{\mathcal{{A}}}\xspace is injective on $T$ and there exists a matrix \ensuremath{\bm{Y}}\xspace which satisfies
\begin{eqnarray}
	\ensuremath{\left|\left| {\ensuremath{\mathcal{P}_T\left( \ensuremath{\bm{Y}}\xspace \right)}\xspace - \ensuremath{\bm{Q}}\xspace\ensuremath{\bm{V}}\xspace^H} \right|\right|}\xspace_F  & \leq & \frac{1}{2\sqrt{2}\gamma} \label{eqn:dual1}  \\
	\ensuremath{\left|\left| {\ensuremath{\mathcal{P}_{T^{\perp}}\left( \ensuremath{\bm{Y}}\xspace \right)}\xspace} \right|\right|}\xspace                  & \leq & \frac{1}{2}               \label{eqn:dual2}
\end{eqnarray}
where the projection \ensuremath{\mathcal{P}_{T^{\perp}}}\xspace is the projection onto the perpendicular space to $T$, 
\begin{gather*}
	\ensuremath{\mathcal{P}_{T^{\perp}}\left( \bm{W} \right)}\xspace = \left(\bm{I} - \ensuremath{\bm{Q}}\xspace\ensuremath{\bm{Q}^{\ast}}\xspace\right)\bm{W}\left(\bm{I} - \ensuremath{\bm{V}}\xspace\ensuremath{\bm{V}^{\ast}}\xspace\right) \nonumber
\end{gather*}

The remainder of this proof will be devoted to demonstrating that there does exist a certificate \ensuremath{\bm{Y}}\xspace by iteratively devising \ensuremath{\bm{Y}}\xspace via a golfing scheme~\citep{gross2011recovering,candes2011probabilistic,ahmed2014blind}. The golfing scheme essentially generates an iterative method which defined a series of certificate vectors \ensuremath{\bm{Y}_k}\xspace for $k\in[1,\cdots,{\kappa}]$ which converge to a certificate $\ensuremath{\bm{Y}}\xspace_{\kappa}$ which satisfies the necessary conditions. As in~\citep{ahmed2013compressive}, we can initialize the $0^{th}$ iterate to zero, and define the $k^{th}$ iterate in terms of the $\ensuremath{\bm{Y}}\xspace_{k-1}$ as
\begin{gather*}
	\ensuremath{\bm{Y}_k}\xspace = \ensuremath{\bm{Y}}\xspace_{k-1} +{\kappa}\ensuremath{\mathcal{{A}}^{\ast}}\xspace_k\ensuremath{\mathcal{{A}}}\xspace_k(\ensuremath{\bm{Q}}\xspace\ensuremath{\bm{V}^{\ast}}\xspace -\ensuremath{\mathcal{P}_T\left( \ensuremath{\bm{Y}}\xspace_{k-1} \right)}\xspace), \nonumber
\end{gather*}
where 
\begin{gather}
	\ensuremath{\mathcal{{A}}\left( \bm{W} \right)}\xspace = \mbox{vec}\left( \ensuremath{\langle {\ensuremath{\bm{{A}}_n}\xspace},{\bm{W}}\rangle}\xspace \right). \label{eqn:linopdef}
\end{gather}
We can see that since every iterate has $\ensuremath{\mathcal{{A}}^{\ast}}\xspace_k$ applied to it, every iteration is projected in to the range of \ensuremath{\mathcal{{A}}^{\ast}}\xspace, indicating that the final iteration \ensuremath{\bm{Y}}\xspace will also be in the range of \ensuremath{\mathcal{{A}}^{\ast}}\xspace. In~\citep{ahmed2013compressive}, Asif and Romberg define a simpler iteration 
\begin{gather*}
	\ensuremath{\widetilde{\bm{Y}}_k}\xspace = (\ensuremath{\mathcal{P}_T}\xspace - {\kappa}\ensuremath{\mathcal{P}_T}\xspace\ensuremath{\mathcal{{A}}^{\ast}}\xspace_k\ensuremath{\mathcal{{A}}}\xspace_k\ensuremath{\mathcal{P}_T}\xspace)\ensuremath{\widetilde{\bm{Y}}}\xspace_{k-1}, \nonumber
\end{gather*}
which is expressed in terms of the modified certificate
\begin{gather*}
	\ensuremath{\widetilde{\bm{Y}}_k}\xspace = \ensuremath{\mathcal{P}_T\left( \ensuremath{\bm{Y}_k}\xspace \right)}\xspace - \ensuremath{\bm{Q}}\xspace\ensuremath{\bm{V}^{\ast}}\xspace.   \nonumber
\end{gather*}

What remains now is to demonstrate that this iterative procedure converges, with high probability, to a certificate which satisfies the desired dual certificate conditions. We start by using Lemma~\ref{lem:lemma1} and observing that the Forbenious norm of the $k^{th}$ iterate is well bounded with probability $1-O((\ensuremath{L}\xspace\ensuremath{N}\xspace)^{-\beta})$ by 
\begin{eqnarray*}
	\ensuremath{\left|\left| {\ensuremath{\widetilde{\bm{Y}}_k}\xspace} \right|\right|}\xspace_F & \leq & \max_k \ensuremath{\left|\left| {\ensuremath{\mathcal{P}_T}\xspace - {\kappa}\ensuremath{\mathcal{P}_T}\xspace\ensuremath{\mathcal{{A}}^{\ast}}\xspace_k\ensuremath{\mathcal{{A}}}\xspace_k\ensuremath{\mathcal{P}_T}\xspace} \right|\right|}\xspace\ensuremath{\left|\left| {\ensuremath{\widetilde{\bm{Y}}}\xspace_{k-1}} \right|\right|}\xspace_F \nonumber \\
	& \leq & 2^{-k}\ensuremath{\left|\left| {\ensuremath{\widetilde{\bm{Y}}}\xspace_{0}} \right|\right|}\xspace_F  \nonumber \\
	& \leq & 2^{-k}\ensuremath{\left|\left| {\ensuremath{\bm{Q}}\xspace\ensuremath{\bm{V}^{\ast}}\xspace} \right|\right|}\xspace_F \nonumber \\
	& \leq & 2^{-k}\sqrt{\ensuremath{R}\xspace},        \nonumber 
\end{eqnarray*}
so long that $\ensuremath{M}\xspace \leq c\beta{\kappa}\ensuremath{R}\xspace(\ensuremath{N}\xspace + \mu^2_0\ensuremath{L}\xspace)\log^2(\ensuremath{L}\xspace\ensuremath{N}\xspace)$. As in~\citep{ahmed2013compressive} we observe that when we choose ${\kappa} \geq 0.5\log_2(8\gamma^2\ensuremath{R}\xspace)$, the bound for the Frobenious norm of $\ensuremath{\widetilde{\bm{Y}}}\xspace_{\kappa}$ is bounded by $\ensuremath{\left|\left| {\ensuremath{\widetilde{\bm{Y}}}\xspace_{\kappa}} \right|\right|}\xspace_F \leq (2\sqrt{2}\gamma)^{-1}$. 

To show that the second condition on the certificate is also satisfies, we apply Lemma~\ref{lem:lemma2}. We begin with writing the quantity we wish to bound in terms of the past golfing scheme iterate
\begin{eqnarray*}
	\ensuremath{\left|\left| {\ensuremath{\mathcal{P}_{T^{\perp}}\left( \ensuremath{\bm{Y}}\xspace_{{\kappa}} \right)}\xspace} \right|\right|}\xspace & \leq & \sum_{k=1}^{\kappa} \ensuremath{\left|\left| {\ensuremath{\mathcal{P}_{T^{\perp}}\left( {\kappa}\ensuremath{\mathcal{{A}}^{\ast}}\xspace_k\ensuremath{\mathcal{{A}}}\xspace_k\ensuremath{\widetilde{\bm{Y}}}\xspace_{k-1}  \right)}\xspace} \right|\right|}\xspace\nonumber \\
	& = & \sum_{k=1}^{\kappa} \ensuremath{\left|\left| {\ensuremath{\mathcal{P}_{T^{\perp}}\left( {\kappa}\ensuremath{\mathcal{{A}}^{\ast}}\xspace_k\ensuremath{\mathcal{{A}}}\xspace_k\ensuremath{\widetilde{\bm{Y}}}\xspace_{k-1} - \ensuremath{\widetilde{\bm{Y}}}\xspace_{k-1}  \right)}\xspace} \right|\right|}\xspace \nonumber \\
	& \leq & \sum_{k=1}^{\kappa} \ensuremath{\left|\left| {{\kappa}\ensuremath{\mathcal{{A}}^{\ast}}\xspace_k\ensuremath{\mathcal{{A}}}\xspace_k\ensuremath{\widetilde{\bm{Y}}}\xspace_{k-1} - \ensuremath{\widetilde{\bm{Y}}}\xspace_{k-1}} \right|\right|}\xspace         \nonumber \\
	& \leq & \sum_{k=1}^{\kappa} \ensuremath{\left|\left| {{\kappa}\ensuremath{\mathcal{{A}}^{\ast}}\xspace_k\ensuremath{\mathcal{{A}}}\xspace_k\ensuremath{\widetilde{\bm{Y}}}\xspace_{k-1} - \ensuremath{\widetilde{\bm{Y}}}\xspace_{k-1}} \right|\right|}\xspace_F       \nonumber \\
	& \leq & \sum_{k=1}^{\kappa} \max_{k\in[1,\ldots{\kappa}]} \ensuremath{\left|\left| {{\kappa}\ensuremath{\mathcal{{A}}^{\ast}}\xspace_k\ensuremath{\mathcal{{A}}}\xspace_k\ensuremath{\widetilde{\bm{Y}}}\xspace_{k-1} - \ensuremath{\widetilde{\bm{Y}}}\xspace_{k-1}} \right|\right|}\xspace_F  \nonumber \\
	& \leq & \sum_{k=1}^{\kappa} \frac{1}{2}2^{-k} \nonumber \\
	& \leq & \frac{1}{2}                                  \nonumber
\end{eqnarray*}
We use Lemma~\ref{lem:lemma2} to bound the maximum spectral norm of ${\kappa}\ensuremath{\mathcal{{A}}^{\ast}}\xspace_k\ensuremath{\mathcal{{A}}}\xspace_k\ensuremath{\widetilde{\bm{Y}}}\xspace_{k-1} - \ensuremath{\widetilde{\bm{Y}}}\xspace_{k-1}$ with probability $1-O((\ensuremath{L}\xspace\ensuremath{N}\xspace)^{1-\beta}$. Taking ${\kappa} \geq \log(\ensuremath{L}\xspace\ensuremath{N}\xspace)$ shows that the final certificate $\ensuremath{\bm{Y}}\xspace_{\kappa}$ satisfies all the desired properties, completing the proof.

\subsubsection{Matrix Bernstein Inequality and Olicz Norm}

The lemmas required in our main result depend heavily on the matrix Bernstein inequality~\citep{tropp2012user}. This inequality uses the variance measure and Oricz norm of a matrix to bound the largest singular value of the matrix. The matrix Bernstein inequality is summarized as
\begin{theorem}[Matrix Bernstein's Inequality]
	\label{thm:matbern}
        Let $\bm{X}_i\in\ensuremath{\mathbb{R}}\xspace^{\ensuremath{L}\xspace,\ensuremath{N}\xspace}$, $i\in[1,\dots,\ensuremath{M}\xspace]$ be \ensuremath{M}\xspace random matrices such that $\ensuremath{\mathbb{E}\left[{\bm{X}_i}\right]}\xspace = 0$ and $\ensuremath{\left|\left| {\bm{X}_i} \right|\right|}\xspace_{\psi_\alpha} < U_{\alpha} < \infty$ for some $\alpha \geq 1$. Then with probability $1-e^{-t}$, the spectral norm of the sum is bounded by
	\begin{gather*}
		\ensuremath{\left|\left| {\sum_{i=1}^\ensuremath{M}\xspace \bm{X}_i } \right|\right|}\xspace \leq C\max\left\{\sigma_X\sqrt{t+\log(\ensuremath{L}\xspace + \ensuremath{N}\xspace)}, U_{\alpha}\log^{1/\alpha}\left(\frac{\ensuremath{M}\xspace U_{\alpha}^2}{\sigma_X^2}\right)\left(t + \log(\ensuremath{L}\xspace + \ensuremath{N}\xspace)\right) \right\}, \nonumber
	\end{gather*}
	for some constant $C$ and the variance parameter defined by
	\begin{gather*}
		\sigma_X = \max\left\{\ensuremath{\left|\left| {\sum_{i=1}^\ensuremath{M}\xspace \ensuremath{\mathbb{E}\left[{\bm{X}_i\bm{X}_i^{\ast}}\right]}\xspace} \right|\right|}\xspace^{1/2}, \ensuremath{\left|\left| {\sum_{i=1}^\ensuremath{M}\xspace \ensuremath{\mathbb{E}\left[{\bm{X}_i^{\ast}\bm{X}_i}\right]}\xspace} \right|\right|}\xspace^{1/2} \right\}. \nonumber
	\end{gather*}
\end{theorem}
where Orlicz-$\alpha$ norm $\ensuremath{\left|\left| {X} \right|\right|}\xspace_{\psi_\alpha}$ is defined as
\begin{gather}
        \ensuremath{\left|\left| {X} \right|\right|}\xspace_{\psi_\alpha} = \inf\left\{ y > 0 | \ensuremath{\mathbb{E}\left[{e^{\ensuremath{\left|\left| {X} \right|\right|}\xspace^{\alpha}/y^{\alpha}}}\right]}\xspace \leq 2 \right\}. \label{eqn:onormdef}
\end{gather}

In particular we will use the matrix Bernstein inequality with the Orlicz-1 and Orlicz-2 norms, since subgaussian and subexponential random variables have bounded Orlicz-2 and -1 norms, respectively. To calculate these norms, we find the following lemmas from~\citep{tropp2012user,ahmed2013compressive} useful:

\begin{lemma}[Lemma 5.14,~\citealp{tropp2012user}]
	\label{lem:onorm1}
	A random variable $X$ is subgaussian iff $X^2$ is subexponential. Furthermore,
	\begin{gather*}
		\ensuremath{\left|\left| {X} \right|\right|}\xspace_{\psi_2}^2 \leq \ensuremath{\left|\left| {X^2} \right|\right|}\xspace_{\psi_1} \leq 2\ensuremath{\left|\left| {X} \right|\right|}\xspace_{\psi_2}^2. \nonumber 
	\end{gather*}
\end{lemma}

\begin{lemma}[Lemma 7,~\citealp{ahmed2013compressive}]
	\label{lem:onorm2}
	Let $X_1$ and $X_2$ be two subgaussian ranfom variables. Then the product $X_1X_2$ is a subexponential random variable with 
	\begin{gather*}
		\ensuremath{\left|\left| {X_1X_2} \right|\right|}\xspace_{\psi_1} \leq c\ensuremath{\left|\left| {X_1} \right|\right|}\xspace_{\psi_2}\ensuremath{\left|\left| {X_2} \right|\right|}\xspace_{\psi_2}.   \nonumber
	\end{gather*}
\end{lemma}

Lemma~\ref{lem:onorm1} relates the Orlicz-1 and -2 norms for a random variable and it's square. Lemma~\ref{lem:onorm2} allows us to factor an Orlicz-1 norm of a sub-exponential random variable as the product of two subgaussian random variables. Finally we find useful the following calculation for the Orlicz-1 norm of the norm of a random Gaussian vector \ensuremath{\bm{{z}}_n}\xspace with \emph{i.i.d.} zero-mean and variance $\sigma^2$ entries: 

\begin{eqnarray}
	\ensuremath{\left|\left| {\ensuremath{\left|\left| {\ensuremath{\bm{{z}}_n}\xspace} \right|\right|}\xspace_2^2} \right|\right|}\xspace_{\psi_1} & = & \inf \left\{y : \ensuremath{\mathbb{E}\left[{e^{\ensuremath{\left|\left| {\ensuremath{\bm{{z}}_n}\xspace} \right|\right|}\xspace_2^2/y}}\right]}\xspace \leq 2 \right\} \nonumber \\
	& = & \inf \left\{y : \frac{1}{\sqrt{2\pi\sigma}} \int_\ensuremath{\mathbb{R}}\xspace e^{-z_n^2(1/2\sigma^2 - 1/y)} dz_n \leq 2^{\frac{1}{\ensuremath{M}\xspace}} \right\} \nonumber \\
	& = & \frac{2\sigma^2}{1-4^{-\frac{1}{\ensuremath{M}\xspace}}}. \label{eqn:onormzvec}
\end{eqnarray}

\subsubsection{\texorpdfstring{Bound on $\ensuremath{\left|\left| {\kappa\ensuremath{\mathcal{P}_T}\xspace\ensuremath{\mathcal{{A}}}\xspace^{\ast}_k\ensuremath{\mathcal{{A}}}\xspace_k\ensuremath{\mathcal{P}_T}\xspace - \ensuremath{\mathcal{P}_T}\xspace} \right|\right|}\xspace$}{Proof of Lemma~\ref{lem:lemma1}}}

\begin{lemma}
	\label{lem:lemma1}
	Let $\ensuremath{\mathcal{P}_T}\xspace$ be defined as in Equation~\eqref{eqn:projdef} and $\ensuremath{\mathcal{{A}}}\xspace_k$ be the restricted measurement operator as defined in Equation~\eqref{eqn:linopdef}. Then if the number of nodes scale as 
	\begin{gather*}
		\ensuremath{M}\xspace \geq c\beta{\kappa}\ensuremath{R}\xspace\left(\ensuremath{N}\xspace + \mu_0^2\ensuremath{L}\xspace\right)\log^2(\ensuremath{L}\xspace\ensuremath{N}\xspace), \nonumber
	\end{gather*}
	for a constant $\beta > 1$, then with probability greater then $1-O({\kappa}(\ensuremath{L}\xspace\ensuremath{N}\xspace)^{-\beta}$, we have
	\begin{gather*}
		\max_{k\in[1,\dots,{\kappa}]}\ensuremath{\left|\left| {{\kappa}\ensuremath{\mathcal{P}_T}\xspace\ensuremath{\mathcal{{A}}^{\ast}}\xspace\ensuremath{\mathcal{{A}}}\xspace\ensuremath{\mathcal{P}_T}\xspace - \ensuremath{\mathcal{P}_T}\xspace} \right|\right|}\xspace \leq \frac{1}{2}. \nonumber
	\end{gather*}
\end{lemma}

\begin{proof}

Lemma~\ref{lem:lemma1} bounds the operator norm 
\[ \ensuremath{\left|\left| {\kappa\ensuremath{\mathcal{P}_T}\xspace\ensuremath{\mathcal{{A}}}\xspace^{\ast}_k\ensuremath{\mathcal{{A}}}\xspace_k\ensuremath{\mathcal{P}_T}\xspace - \ensuremath{\mathcal{P}_T}\xspace} \right|\right|}\xspace. \]
Since $\ensuremath{\mathbb{E}\left[{\ensuremath{\mathcal{{A}}}\xspace_k^{\ast}\ensuremath{\mathcal{{A}}}\xspace_k}\right]}\xspace = \frac{1}{{\kappa}}\mathcal{I}$, this norm is equivalent to
\begin{eqnarray*}
	{\kappa}\ensuremath{\mathcal{P}_T}\xspace\ensuremath{\mathcal{{A}}^{\ast}}\xspace_k\ensuremath{\mathcal{{A}}}\xspace_k\ensuremath{\mathcal{P}_T}\xspace - \ensuremath{\mathcal{P}_T}\xspace & = & \kappa\ensuremath{\mathcal{P}_T}\xspace\ensuremath{\mathcal{{A}}^{\ast}}\xspace_k\ensuremath{\mathcal{{A}}}\xspace_k\ensuremath{\mathcal{P}_T}\xspace - \ensuremath{\mathbb{E}\left[{{\kappa}\ensuremath{\mathcal{P}_T}\xspace\ensuremath{\mathcal{{A}}^{\ast}}\xspace_k\ensuremath{\mathcal{{A}}}\xspace_k\ensuremath{\mathcal{P}_T}\xspace}\right]}\xspace \nonumber \\
	& = & \kappa \sum_{n\in\Gamma_k}\left(\ensuremath{\mathcal{P}_T}\xspace(\ensuremath{\bm{{A}}_n}\xspace)\otimes\ensuremath{\mathcal{P}_T}\xspace(\ensuremath{\bm{{A}}_n}\xspace) - \ensuremath{\mathbb{E}\left[{\ensuremath{\mathcal{P}_T}\xspace(\ensuremath{\bm{{A}}_n}\xspace)\otimes\ensuremath{\mathcal{P}_T}\xspace(\ensuremath{\bm{{A}}_n}\xspace)}\right]}\xspace \right). \nonumber
\end{eqnarray*}

We can also define here $\mathcal{L}_n(\bm{C}) = \ensuremath{\langle {\ensuremath{\mathcal{P}_T}\xspace(\ensuremath{\bm{{A}}_n}\xspace)},{\bm{C}}\rangle}\xspace\ensuremath{\mathcal{P}_T}\xspace(\ensuremath{\bm{{A}}_n}\xspace)$ which has $\ensuremath{\left|\left| {\mathcal{L}_n} \right|\right|}\xspace = \ensuremath{\left|\left| {\ensuremath{\mathcal{P}_T}\xspace(\ensuremath{\bm{{A}}_n}\xspace)} \right|\right|}\xspace_F^2$ which gives us
\begin{gather*}
	{\kappa}\ensuremath{\mathcal{P}_T}\xspace\ensuremath{\mathcal{{A}}}\xspace^{\ast}_k\ensuremath{\mathcal{{A}}}\xspace_k\ensuremath{\mathcal{P}_T}\xspace - \ensuremath{\mathbb{E}\left[{{\kappa}\ensuremath{\mathcal{P}_T}\xspace\ensuremath{\mathcal{{A}}^{\ast}}\xspace_k\ensuremath{\mathcal{{A}}}\xspace_k\ensuremath{\mathcal{P}_T}\xspace}\right]}\xspace =  {\kappa}\sum_{n\in\Gamma_k}(\mathcal{L}_n - \ensuremath{\mathbb{E}\left[{\mathcal{L}_n}\right]}\xspace). \nonumber 
\end{gather*}

To calculate the variance, we can use the symmetry of $\mathcal{L}_n$ to only calculate
\begin{gather*}
	{\kappa}^2\ensuremath{\left|\left| {\sum_{n\in\Gamma_k}\ensuremath{\mathbb{E}\left[{\mathcal{L}_n^2}\right]}\xspace - \ensuremath{\mathbb{E}\left[{\mathcal{L}_n}\right]}\xspace^2} \right|\right|}\xspace \leq {\kappa}^2\ensuremath{\left|\left| {\sum_{n\in\Gamma_k} \ensuremath{\mathbb{E}\left[{\mathcal{L}_n^2}\right]}\xspace} \right|\right|}\xspace  =  {\kappa}^2\ensuremath{\left|\left| {\ensuremath{\mathbb{E}\left[{\sum_{n\in\Gamma_k} \|\ensuremath{\mathcal{P}_T}\xspace(\ensuremath{\bm{{A}}_n}\xspace)\|_F^2\mathcal{L}_n }\right]}\xspace} \right|\right|}\xspace. \nonumber
\end{gather*}
We now need to bound $\|\ensuremath{\mathcal{P}_T}\xspace(\ensuremath{\bm{{A}}_n}\xspace)\|_F^2$, which can be done by the following:
\begin{eqnarray*}
	\ensuremath{\left|\left| {\ensuremath{\mathcal{P}_T}\xspace(\ensuremath{\bm{{A}}_n}\xspace)} \right|\right|}\xspace_F^2 & = & \ensuremath{\langle {\ensuremath{\mathcal{P}_T}\xspace(\ensuremath{\bm{{A}}_n}\xspace)},{\ensuremath{\bm{{A}}_n}\xspace}\rangle}\xspace \nonumber \\
	& = & \ensuremath{\langle {\ensuremath{\bm{Q}}\xspace\ensuremath{\bm{Q}^{\ast}}\xspace\ensuremath{\bm{{z}}_n}\xspace\ensuremath{\bm{{f}}^{\ast}_n}\xspace},{\ensuremath{\bm{{z}}_n}\xspace\ensuremath{\bm{{f}}^{\ast}_n}\xspace }\rangle}\xspace + \ensuremath{\langle {\ensuremath{\bm{{z}}_n}\xspace\ensuremath{\bm{{f}}^{\ast}_n}\xspace\ensuremath{\bm{V}}\xspace\ensuremath{\bm{V}^{\ast}}\xspace},{\ensuremath{\bm{{z}}_n}\xspace\ensuremath{\bm{{f}}^{\ast}_n}\xspace}\rangle}\xspace - \ensuremath{\langle {\ensuremath{\bm{Q}}\xspace\ensuremath{\bm{Q}^{\ast}}\xspace\ensuremath{\bm{{z}}_n}\xspace\ensuremath{\bm{{f}}^{\ast}_n}\xspace\ensuremath{\bm{V}}\xspace\ensuremath{\bm{V}^{\ast}}\xspace},{\ensuremath{\bm{{z}}_n}\xspace\ensuremath{\bm{{f}}^{\ast}_n}\xspace}\rangle}\xspace \nonumber \\
	& = & \|\bm{f}_n\|_2^2\|\ensuremath{\bm{Q}^{\ast}}\xspace\ensuremath{\bm{{z}}_n}\xspace\|_2^2 + \|\ensuremath{\bm{{z}}_n}\xspace\|_2^2\|\ensuremath{\bm{V}^{\ast}}\xspace\ensuremath{\bm{{f}}_n}\xspace\|_2^2 - \|\ensuremath{\bm{Q}^{\ast}}\xspace\ensuremath{\bm{{z}}_n}\xspace\|_2^2\|\ensuremath{\bm{V}^{\ast}}\xspace\ensuremath{\bm{{f}}_n}\xspace\|_2^2 \nonumber \\
	& \leq & N\|\ensuremath{\bm{Q}^{\ast}}\xspace\ensuremath{\bm{{z}}_n}\xspace\|_2^2 + \|\ensuremath{\bm{{z}}_n}\xspace\|_2^2\|\ensuremath{\bm{V}^{\ast}}\xspace\ensuremath{\bm{{f}}_n}\xspace\|_2^2.  \nonumber
\end{eqnarray*}

Using this calculation we can write
\begin{eqnarray*}
	\ensuremath{\left|\left| {\ensuremath{\mathbb{E}\left[{\sum_{n\in\Gamma_k} \ensuremath{\left|\left| {\ensuremath{\mathcal{P}_T\left( \ensuremath{\bm{{A}}_n}\xspace \right)}\xspace} \right|\right|}\xspace_F^2\mathcal{L}_n}\right]}\xspace} \right|\right|}\xspace & \leq & \ensuremath{\left|\left| {\sum_{n\in\Gamma_k}\ensuremath{\mathbb{E}\left[{ \left(N\|\ensuremath{\bm{Q}^{\ast}}\xspace\ensuremath{\bm{{z}}_n}\xspace\|_2^2 + \|\ensuremath{\bm{{z}}_n}\xspace\|_2^2\|\ensuremath{\bm{V}^{\ast}}\xspace\ensuremath{\bm{{f}}_n}\xspace\|_2^2 \right)\mathcal{L}_n }\right]}\xspace} \right|\right|}\xspace \nonumber \\
	& \leq & N\left|\left| \sum_{n\in\Gamma_k}\mathbb{E}\left[\|\bm{Q}^{\ast}\bm{{z}}_n\|_2^2\mathcal{L}_n\right] \right|\right| \nonumber \\
        & & \qquad\qquad\qquad + \sup\|\bm{V}^{\ast}\bm{{f}}_n\|_{\infty}\left|\left| \sum_{n\in\Gamma_k}\mathbb{E}\left[{ \|\bm{{z}}_n\|_2^2\mathcal{L}_n }\right] \right|\right| \nonumber \\
	& \leq & \ensuremath{N}\xspace \ensuremath{\left|\left| {\sum_{n\in\Gamma_k}\ensuremath{\mathbb{E}\left[{\|\ensuremath{\bm{Q}^{\ast}}\xspace\ensuremath{\bm{{z}}_n}\xspace\|_2^2\mathcal{L}_n}\right]}\xspace} \right|\right|}\xspace + \ensuremath{R}\xspace\mu_0^2\ensuremath{\left|\left| {\sum_{n\in\Gamma_k}\ensuremath{\mathbb{E}\left[{\|\ensuremath{\bm{{z}}_n}\xspace\|_2^2\mathcal{L}_n }\right]}\xspace} \right|\right|}\xspace \nonumber
\end{eqnarray*}

We now need to bound these two quantities. First we look to bound the first quantity
\begin{eqnarray*}
	\ensuremath{\left|\left| {\sum_{n\in\Gamma_k}\ensuremath{\mathbb{E}\left[{\|\ensuremath{\bm{Q}^{\ast}}\xspace\ensuremath{\bm{{z}}_n}\xspace\|_2^2(\ensuremath{\mathcal{P}_T}\xspace(\ensuremath{\bm{{A}}_n}\xspace)\otimes\ensuremath{\mathcal{P}_T}\xspace(\ensuremath{\bm{{A}}_n}\xspace))}\right]}\xspace} \right|\right|}\xspace & \leq & \|\ensuremath{\mathcal{P}_T}\xspace\|\ensuremath{\left|\left| {\ensuremath{\mathbb{E}\left[{\|\ensuremath{\bm{Q}^{\ast}}\xspace\ensuremath{\bm{{z}}_n}\xspace\|_2^2(\ensuremath{\bm{{A}}_n}\xspace\otimes\ensuremath{\bm{{A}}_n}\xspace)}\right]}\xspace} \right|\right|}\xspace\|\ensuremath{\mathcal{P}_T}\xspace\| \nonumber \\
	& \leq & \ensuremath{\left|\left| {\ensuremath{\mathbb{E}\left[{\|\ensuremath{\bm{Q}^{\ast}}\xspace\ensuremath{\bm{{z}}_n}\xspace\|_2^2\{z_n[\alpha]z^{\ast}_n[\beta]\ensuremath{\bm{{f}}_n}\xspace\ensuremath{\bm{{f}}^{\ast}_n}\xspace\}_{\alpha,\beta} }\right]}\xspace} \right|\right|}\xspace. \nonumber
\end{eqnarray*}
Expanding, we have:
\begin{eqnarray*}
	\ast & = & \ensuremath{\left|\left| {\ensuremath{\mathbb{E}\left[{\left(\sum_{l=1}^L\|\bm{q}_l\|_2^2|z_n[l]|^2 + 2\sum_{l \neq m}\mbox{Re}(\langle\bm{q}_l,\bm{q}_m\rangle z_n[l]z^{\ast}_n[m])\right) z_n[\alpha]z_n^{\ast}[\beta]\bm{I}_\ensuremath{N}\xspace }\right]}\xspace_{\alpha,\beta}} \right|\right|}\xspace \nonumber \\
	& = & \ensuremath{\left|\left| {\left\{\frac{1}{\ensuremath{M}\xspace^2}\sum_{l=1}^L\|\bm{q}_l\|_2^2\bm{I}_\ensuremath{N}\xspace\delta_{\alpha=\beta} + \frac{2}{\ensuremath{M}\xspace^2}\ensuremath{\langle {\bm{q}_\alpha},{\bm{q}_\beta}\rangle}\xspace \bm{I}_\ensuremath{N}\xspace\delta_{\alpha \neq \beta} \right\}_{\alpha,\beta} } \right|\right|}\xspace \nonumber \\
	& = & \frac{1}{\ensuremath{M}\xspace^2}\left\|\left\{\|\ensuremath{\bm{Q}}\xspace\|_F^2\bm{I}_\ensuremath{N}\xspace\delta_{\alpha=\beta} + 2\ensuremath{\langle {\bm{q}_\alpha},{\bm{q}_\beta}\rangle}\xspace \bm{I}_\ensuremath{N}\xspace\delta_{\alpha \neq \beta} \right\}_{\alpha,\beta}\right\|, \nonumber
\end{eqnarray*}
giving us 
\begin{gather*}
	\ensuremath{\left|\left| {\sum_{n\in\Gamma_k}\ensuremath{\mathbb{E}\left[{\|\ensuremath{\bm{Q}^{\ast}}\xspace\ensuremath{\bm{{z}}_n}\xspace\|_2^2(\ensuremath{\mathcal{P}_T}\xspace(\ensuremath{\bm{{A}}_n}\xspace)\otimes\ensuremath{\mathcal{P}_T}\xspace(\ensuremath{\bm{{A}}_n}\xspace))}\right]}\xspace} \right|\right|}\xspace \leq \frac{1}{\ensuremath{M}\xspace{\kappa}} \|\ensuremath{\bm{Q}}\xspace\|_F^2 + \|\ensuremath{\bm{Q}}\xspace\ensuremath{\bm{Q}^{\ast}}\xspace\| \leq \frac{\ensuremath{R}\xspace+1}{\ensuremath{M}\xspace{\kappa}}. \nonumber
\end{gather*}

Similarly, for the second term we can take $\ensuremath{\bm{Q}}\xspace = \bm{I}_\ensuremath{L}\xspace$ to get 
\begin{gather*}
	\ensuremath{\left|\left| {\sum_{n\in\Gamma_k}\ensuremath{\mathbb{E}\left[{\|\ensuremath{\bm{{z}}_n}\xspace\|_2^2(\ensuremath{\mathcal{P}_T}\xspace(\ensuremath{\bm{{A}}_n}\xspace)\otimes\ensuremath{\mathcal{P}_T}\xspace(\ensuremath{\bm{{A}}_n}\xspace))}\right]}\xspace} \right|\right|}\xspace \leq \frac{1}{\ensuremath{M}\xspace{\kappa}} \|\bm{I}\|_F^2  = \frac{\ensuremath{L}\xspace}{\ensuremath{M}\xspace{\kappa}}.  \nonumber 
\end{gather*}

Putting the pieces together, we get 
\begin{gather*}
	\sigma_X^2 = {\kappa} \ensuremath{R}\xspace\frac{\ensuremath{N}\xspace+\mu_0^2\ensuremath{L}\xspace}{\ensuremath{M}\xspace}. \nonumber
\end{gather*}

To use the matrix Bernstein inequality, it now remains to bound the following Orlicz-1 norm: ${\kappa}\ensuremath{\left|\left| {\mathcal{L}_n - \ensuremath{\mathbb{E}\left[{\mathcal{L}_n}\right]}\xspace} \right|\right|}\xspace_{\psi_1}$. By the PSD quality of $\mathcal{L}_n$ and its expectation, 
\begin{gather*}
	\ensuremath{\left|\left| {\mathcal{L}_n - \ensuremath{\mathbb{E}\left[{\mathcal{L}_n}\right]}\xspace} \right|\right|}\xspace_{\psi_1} \leq \max\left\{\ensuremath{\left|\left| {\mathcal{L}_n} \right|\right|}\xspace_{\psi_1} - \ensuremath{\left|\left| {\ensuremath{\mathbb{E}\left[{\mathcal{L}_n}\right]}\xspace} \right|\right|}\xspace_{\psi_1} \right\}. \nonumber
\end{gather*}
The norm of $\ensuremath{\left|\left| {\ensuremath{\mathbb{E}\left[{\mathcal{L}_n}\right]}\xspace} \right|\right|}\xspace$ can be calculated via
\begin{gather*}
	\ensuremath{\left|\left| {\ensuremath{\mathbb{E}\left[{\mathcal{L}_n}\right]}\xspace} \right|\right|}\xspace = \ensuremath{\left|\left| {\ensuremath{\mathbb{E}\left[{\ensuremath{\mathcal{P}_T\left( \ensuremath{\bm{{A}}_n}\xspace \right)}\xspace}\right]}\xspace} \right|\right|}\xspace_F^2 = \ensuremath{\mathbb{E}\left[{\sum_m|z_n[m]|^2|f_n[m]|^2}\right]}\xspace = \frac{1}{\ensuremath{M}\xspace}\|\ensuremath{\bm{{f}}_n}\xspace\|_2^2 = \frac{\ensuremath{N}\xspace}{\ensuremath{M}\xspace},  \nonumber
\end{gather*}
indicating that the second term is simply $\ensuremath{\left|\left| {\ensuremath{\mathbb{E}\left[{\mathcal{L}_n}\right]}\xspace} \right|\right|}\xspace_{\psi_1} = \ensuremath{N}\xspace/(\ensuremath{M}\xspace\log(2))$. To calculate $\ensuremath{\left|\left| {\ensuremath{\mathbb{E}\left[{\mathcal{L}_n}\right]}\xspace} \right|\right|}\xspace_{\psi_1}$, we use the definition of the Orlitcz-1 norm in Equation~\eqref{eqn:onormdef} to see that 
\begin{gather*}
        \ensuremath{\left|\left| {\mathcal{L}_n} \right|\right|}\xspace_{\psi_1} =  \ensuremath{\left|\left| {\ensuremath{\left|\left| {\ensuremath{\mathcal{P}_T\left( \ensuremath{\bm{{A}}_n}\xspace \right)}\xspace} \right|\right|}\xspace_2^2} \right|\right|}\xspace_{\psi_1} \leq \ensuremath{N}\xspace\ensuremath{\left|\left| {\|\ensuremath{\bm{Q}^{\ast}}\xspace\ensuremath{\bm{{z}}_n}\xspace\|_2^2} \right|\right|}\xspace_{\psi_1} + R\ensuremath{N}\xspace\mu_0^2\ensuremath{\left|\left| {\|\ensuremath{\bm{{z}}_n}\xspace\|_2^2} \right|\right|}\xspace_{\psi_1},  \nonumber 
\end{gather*}
where the inequality stems form the fact that $\|\ensuremath{\bm{V}^{\ast}}\xspace\ensuremath{\bm{{f}}_n}\xspace\|_2^2 \leq R\ensuremath{N}\xspace\mu_0^2$ and $\|\ensuremath{\bm{{f}}_n}\xspace\|_2^2 \leq \ensuremath{N}\xspace$. Using the result in Equation~\eqref{eqn:onormzvec} with $\sigma^2 = 1/\ensuremath{M}\xspace$ in the first term and in $\sigma^2 = \ensuremath{R}\xspace/\ensuremath{M}\xspace$ in the second term yields 
\begin{eqnarray*}
	\ensuremath{\left|\left| {\ensuremath{\mathbb{E}\left[{\mathcal{L}_n}\right]}\xspace} \right|\right|}\xspace_{\psi_1} & \leq & \ensuremath{N}\xspace\ensuremath{\left|\left| {\|\ensuremath{\bm{Q}^{\ast}}\xspace\ensuremath{\bm{{z}}_n}\xspace\|_2^2} \right|\right|}\xspace_{\psi_1} + R\ensuremath{N}\xspace\mu_0^2\ensuremath{\left|\left| {\|\ensuremath{\bm{{z}}_n}\xspace\|_2^2} \right|\right|}\xspace_{\psi_1} \nonumber \\
	& \leq & \frac{1}{\ensuremath{M}\xspace}\left(\frac{\ensuremath{N}\xspace}{\ensuremath{M}\xspace(1-4^{-\frac{1}{\ensuremath{R}\xspace}})} + \frac{\ensuremath{R}\xspace\mu_0^2}{1-4^{-\frac{1}{\ensuremath{L}\xspace}}}\right) \nonumber \\  
	& \leq & \frac{2\ensuremath{R}\xspace\left(\ensuremath{N}\xspace + \ensuremath{L}\xspace\mu_0^2\right)}{\log(2)\ensuremath{M}\xspace}. \nonumber 
\end{eqnarray*}

We now have appropriate bounds on both the variance and Orliscz norm, permitting a bound on the largest singular value via the Matrix Bernstein inequality. Specifically, we can see that the first term in Theorem~\ref{thm:matbern} with $t = \beta\log(\ensuremath{L}\xspace\ensuremath{N}\xspace) > \log(\ensuremath{N}\xspace + \ensuremath{L}\xspace)$ is bounded as
\begin{gather*}
	\sigma_X\sqrt{t+\log(\ensuremath{L}\xspace + \ensuremath{N}\xspace)} \leq \sqrt{ 2{\kappa}\ensuremath{R}\xspace\beta\frac{\ensuremath{N}\xspace + \mu_0^2\ensuremath{L}\xspace}{\ensuremath{M}\xspace}\log(\ensuremath{L}\xspace\ensuremath{N}\xspace)}. \nonumber 
\end{gather*}

Likewise we can bound the second term:
\begin{eqnarray*}
    U_1\log\left(\frac{\ensuremath{M}\xspace U_1^2}{\sigma_X^2}\right)\left(t + \log(\ensuremath{L}\xspace + \ensuremath{N}\xspace)\right) & \leq & 2\beta U_1\log\left(\frac{4\Delta{\kappa}\ensuremath{R}\xspace(\ensuremath{N}\xspace + \mu_0^2\ensuremath{L}\xspace)}{\log^2(2)\ensuremath{M}\xspace}\right)\log(\ensuremath{L}\xspace\ensuremath{N}\xspace) \nonumber \\
    & \leq & c\frac{\beta{\kappa}\ensuremath{R}\xspace(\ensuremath{N}\xspace + \mu_0^2\ensuremath{L}\xspace)}{\ensuremath{M}\xspace}\log^2(\ensuremath{L}\xspace\ensuremath{N}\xspace). \nonumber
\end{eqnarray*}

Thus to appropriately bound
\begin{gather*}
	\ensuremath{\left|\left| {{\kappa}\ensuremath{\mathcal{P}_T}\xspace\ensuremath{\mathcal{{A}}^{\ast}}\xspace_k\ensuremath{\mathcal{{A}}}\xspace_k\ensuremath{\mathcal{P}_T}\xspace - \ensuremath{\mathcal{P}_T}\xspace} \right|\right|}\xspace \leq c\max\left\{ \sqrt{\frac{{\kappa}\ensuremath{R}\xspace\beta(\ensuremath{N}\xspace + \mu_0^2\ensuremath{L}\xspace)}{\ensuremath{M}\xspace}\log(\ensuremath{L}\xspace\ensuremath{N}\xspace)}, \frac{\beta{\kappa}\ensuremath{R}\xspace(\ensuremath{N}\xspace + \mu_0^2\ensuremath{L}\xspace)}{\ensuremath{M}\xspace}\log^2(\ensuremath{L}\xspace\ensuremath{N}\xspace)\right\}, \nonumber
\end{gather*}
we can see that we would need
\begin{gather*}
	\ensuremath{M}\xspace \geq C\beta{\kappa}\ensuremath{R}\xspace(\ensuremath{N}\xspace + \mu_0^2\ensuremath{L}\xspace)\log^2(\ensuremath{L}\xspace\ensuremath{N}\xspace). \nonumber 
\end{gather*}
Taking the union bound over the ${\kappa}$ partitions completes the proof of the lemma. 
\end{proof}

\subsubsection{\texorpdfstring{Bound on $\|(\mathcal{A}^{\ast}\mathcal{A} - \mathcal{I})(\bm{G})\|$}{Proof of Lemma~\ref{lem:lemma2}}}

\begin{lemma}
	\label{lem:lemma2}
	Let $\ensuremath{\mathcal{{A}}}\xspace_k$ be defined as in Equation~\eqref{eqn:linopdef}, ${\kappa} < \ensuremath{M}\xspace$ be the number of steps in the golfing scheme and assume that $\ensuremath{M}\xspace \leq \ensuremath{L}\xspace\ensuremath{N}\xspace$. Then as long as 
	\begin{gather*}
		\ensuremath{M}\xspace \geq c\beta{\kappa}\max\left(\ensuremath{N}\xspace + \mu_0^2\ensuremath{L}\xspace\right)\log^2(\ensuremath{N}\xspace\ensuremath{L}\xspace), \nonumber
	\end{gather*}
	where $\mu_k^2$ is the coherence term defined by
	\begin{gather}
                \mu_k^2 = \ensuremath{R}\xspace^{-1}\sup_{ \omega\in[0,2\pi] }\ensuremath{\left|\left| { \widetilde{\bm{Y}}_k^{\ast}\ensuremath{\bm{{f}}}\xspace_{\omega} } \right|\right|}\xspace_2^2, \label{eqn:coherencek}
	\end{gather}
	then with probability at least $1 - O(\ensuremath{M}\xspace(\ensuremath{L}\xspace\ensuremath{N}\xspace)^{-\beta})$, we have
	\begin{gather*}
		\max_k\ensuremath{\left|\left| {{\kappa}\ensuremath{\mathcal{{A}}^{\ast}}\xspace_k\ensuremath{\mathcal{{A}}}\xspace_k(\ensuremath{\widetilde{\bm{Y}}}\xspace_{k-1}) - \ensuremath{\widetilde{\bm{Y}}}\xspace_{k-1}} \right|\right|}\xspace \leq 2^{-(k+1)}. \nonumber
	\end{gather*}
\end{lemma}

\begin{proof}

Lemma~\ref{lem:lemma2} essentially bounds the operator norm of ${\kappa}\ensuremath{\mathcal{{A}}^{\ast}}\xspace\ensuremath{\mathcal{{A}}}\xspace - \mathcal{I}$. In particular, to prove Theorem 2, the reduced version with ${\kappa} = 1$ is needed. Lemma 2 essentially uses the matrix Bernstein inequality to accomplish this task, taking 
\[ X_n = {\kappa} (\ensuremath{\langle {\bm{G}},{\ensuremath{\bm{{A}}_n}\xspace}\rangle}\xspace\ensuremath{\bm{{A}}_n}\xspace - \ensuremath{\mathbb{E}\left[{\ensuremath{\langle {\bm{G}},{\ensuremath{\bm{{A}}_n}\xspace}\rangle}\xspace\ensuremath{\bm{{A}}_n}\xspace}\right]}\xspace), \]
and we just need to control $\ensuremath{\left|\left| {\sum \ensuremath{\mathbb{E}\left[{X_nX_n^{\ast}}\right]}\xspace} \right|\right|}\xspace$ and $\ensuremath{\left|\left| {\sum \ensuremath{\mathbb{E}\left[{X_n^{\ast}X_n}\right]}\xspace} \right|\right|}\xspace$.
To bound the second of these, we can calculate
\begin{eqnarray*}
	\ensuremath{\left|\left| {\sum_{n\in\Gamma_k} \ensuremath{\mathbb{E}\left[{X_n^{\ast}X_n}\right]}\xspace} \right|\right|}\xspace & \leq & {\kappa}^2\ensuremath{\left|\left| {\sum_{n\in\Gamma_k} \ensuremath{\mathbb{E}\left[{|\ensuremath{\langle {\bm{G}},{\ensuremath{\bm{{A}}_n}\xspace}\rangle}\xspace|^2\ensuremath{\bm{{A}}_n}\xspace\ensuremath{\bm{{A}}_n^{\ast}}\xspace}\right]}\xspace } \right|\right|}\xspace  \nonumber \\
	& = & {\kappa}^2\ensuremath{\left|\left| {\sum_{n\in\Gamma_k} \ensuremath{\mathbb{E}\left[{\ensuremath{\left|\left| {\ensuremath{\bm{{f}}_n}\xspace} \right|\right|}\xspace_2^2|\ensuremath{\langle {\bm{G}},{\ensuremath{\bm{{A}}_n}\xspace}\rangle}\xspace|^2\ensuremath{\bm{{z}}_n}\xspace\ensuremath{\bm{{z}}^{\ast}_n}\xspace}\right]}\xspace} \right|\right|}\xspace \nonumber\\
	& \leq & \frac{3\ensuremath{N}\xspace{\kappa}}{\ensuremath{M}\xspace}\ensuremath{\left|\left| {\bm{G}} \right|\right|}\xspace_F^2, \nonumber 
\end{eqnarray*}
where the second inequality is due to Lemma~\ref{lem:lemma3} and $\ensuremath{\left|\left| {\ensuremath{\bm{{f}}_n}\xspace} \right|\right|}\xspace_2^2\leq\ensuremath{N}\xspace$. For the other expectation
\begin{eqnarray*}
	\ensuremath{\left|\left| {\sum_{n\in\Gamma_k} \ensuremath{\mathbb{E}\left[{X_nX_n^{\ast}}\right]}\xspace} \right|\right|}\xspace & \leq & {\kappa}^2\ensuremath{\left|\left| {\sum_{n\in\Gamma_k} \ensuremath{\mathbb{E}\left[{|\ensuremath{\langle {\bm{G}},{\ensuremath{\bm{{A}}_n}\xspace}\rangle}\xspace|^2\ensuremath{\bm{{A}}_n^{\ast}}\xspace\ensuremath{\bm{{A}}_n}\xspace}\right]}\xspace} \right|\right|}\xspace  \nonumber \\
	& = & \frac{\ensuremath{L}\xspace{\kappa}^2}{\ensuremath{M}\xspace^2}\ensuremath{\left|\left| {\sum_{n\in\Gamma_k} \ensuremath{\mathbb{E}\left[{\|\bm{G}\ensuremath{\bm{{f}}_n}\xspace\|_2^2\ensuremath{\bm{{f}}_n}\xspace\ensuremath{\bm{{f}}^{\ast}_n}\xspace}\right]}\xspace} \right|\right|}\xspace \nonumber \\
        & \leq & \frac{\ensuremath{L}\xspace{\kappa}^2}{\ensuremath{M}\xspace^2}\sup_{\omega}(\|\bm{G}\bm{f}_{\omega}\|_2^2) \ensuremath{\left|\left| {\sum_{n\in\Gamma_k} \bm{1}_\ensuremath{N}\xspace} \right|\right|}\xspace \nonumber \\
        & \leq & \frac{\ensuremath{L}\xspace{\kappa}}{\ensuremath{M}\xspace}\sup_{\omega}(\|\bm{G}\bm{f}_\omega\|_2^2) \nonumber \\
	& = & \frac{\ensuremath{L}\xspace{\kappa}}{\ensuremath{M}\xspace}\mu^2\|\bm{G}\|_F^2. \nonumber 
\end{eqnarray*}

Using these bounds, we can write 
\begin{gather*}
	\sigma_X^2 \leq\frac{{\kappa}}{\ensuremath{M}\xspace}\|\bm{G}\|_F^2\max\left\{\mu_0\ensuremath{L}\xspace,3\ensuremath{N}\xspace \right\}, \nonumber
\end{gather*}
and to use Proposition 1 we just need to bound $\|X\|_{\psi_2}$. To start, we can see that
\begin{eqnarray*}
	U_1 & =    & \ensuremath{\left|\left| {X} \right|\right|}\xspace_{\psi_1} \leq 2{\kappa} \ensuremath{\left|\left| {\ensuremath{\langle {\bm{G}},{\ensuremath{\bm{{A}}_n}\xspace}\rangle}\xspace\ensuremath{\bm{{A}}_n}\xspace} \right|\right|}\xspace_{\psi_1} \nonumber \\
	& \leq & c{\kappa}\ensuremath{\left|\left| {\ensuremath{\langle {\bm{G}},{\ensuremath{\bm{{A}}_n}\xspace}\rangle}\xspace} \right|\right|}\xspace_{\psi_2} \ensuremath{\left|\left| {\ensuremath{\left|\left| {\ensuremath{\bm{{A}}_n}\xspace} \right|\right|}\xspace_F} \right|\right|}\xspace_{\psi_2} \nonumber \\
	& \leq & c{\kappa}\ensuremath{\left|\left| {\ensuremath{\langle {\bm{G}},{\ensuremath{\bm{{A}}_n}\xspace}\rangle}\xspace} \right|\right|}\xspace_{\psi_2} \sqrt{\ensuremath{\left|\left| {\ensuremath{\left|\left| {\ensuremath{\bm{{f}}_n}\xspace} \right|\right|}\xspace_2^2\ensuremath{\left|\left| {\ensuremath{\bm{{z}}_n}\xspace} \right|\right|}\xspace_2^2} \right|\right|}\xspace_{\psi_2}} \nonumber \\
	& = & c{\kappa}\sqrt{\frac{\ensuremath{N}\xspace}{\ensuremath{M}\xspace\left(1-4^{-1/\ensuremath{L}\xspace}\right)}}\ensuremath{\left|\left| {|\mbox{trace}(\ensuremath{\bm{{f}}_n}\xspace\ensuremath{\bm{{z}}^{\ast}_n}\xspace\bm{G})} \right|\right|}\xspace_{\psi_2} \nonumber \\
        & \leq & c{\kappa}\sqrt{\frac{\ensuremath{N}\xspace}{\ensuremath{M}\xspace^2\left(1-4^{-1/\ensuremath{L}\xspace}\right)}} \sum_{l=1}^L\ensuremath{\left|\left| {\langle\bm{g}_l,\bm{f}_{n}\rangle} \right|\right|}\xspace_{\psi_2} \nonumber \\
	& \leq & c{\kappa}\sqrt{\frac{\ensuremath{N}\xspace^2\ensuremath{L}\xspace\mu_0^2\|\bm{G}\|_F^2}{\ensuremath{M}\xspace^2}}. \nonumber 
\end{eqnarray*}

We can now apply the matrix Bernstein theorem with the calculated values of $U_1$ and $\sigma_X$. Again using $t = \beta\log(\ensuremath{L}\xspace\ensuremath{N}\xspace)$, the first portion of the bound is
\begin{gather*}
	\sigma_X\sqrt{t + \log(\ensuremath{L}\xspace + \ensuremath{N}\xspace)}  \leq  c\ensuremath{\left|\left| {\bm{G}} \right|\right|}\xspace_F\sqrt{\frac{{\kappa}\beta}{\ensuremath{M}\xspace}\max\{\mu_k^2\ensuremath{L}\xspace, \ensuremath{N}\xspace\}\log(\ensuremath{L}\xspace\ensuremath{N}\xspace)}, \nonumber
\end{gather*}
and the second portion of the bound is
\begin{eqnarray*}
        & & U_1 \log\left(\frac{\Delta U_1^2}{\sigma_X}\right)(t + \log(L + N)) \nonumber \\
        & & \qquad\qquad\qquad\qquad\leq c\left|\left| \bm{G} \right|\right|_F{\kappa}\sqrt{\frac{LN\mu_k^2}{M}}\log\left(\frac{LN\mu_k^2}{M^2}\frac{M c\left|\left|\bm{G}\right|\right|^2_F\kappa^2}{\kappa\left|\left|\bm{G}\right|\right|_F^2\max\{\mu_k^2L,N\}}\right)\beta\log(LN) \nonumber \\
	& & \qquad\qquad\qquad\qquad \leq c\left|\left| \bm{G} \right|\right|_F\kappa\sqrt{\frac{LN\mu_k^2}{M}}\log\left(\frac{c\Delta\kappa LN\mu_k^2}{M\max\{\mu_k^2L,N\}}\right)\beta\log(LN) \nonumber \\
	& & \qquad\qquad\qquad\qquad \leq c\beta\ensuremath{\left|\left| {\bm{G}} \right|\right|}\xspace_F{\kappa}\sqrt{\frac{\ensuremath{L}\xspace\ensuremath{N}\xspace\mu_k^2}{\ensuremath{M}\xspace}}\log\left(\min\{\mu_k^2\ensuremath{L}\xspace, \ensuremath{N}\xspace \}\right)\log(\ensuremath{L}\xspace\ensuremath{N}\xspace). \nonumber 
\end{eqnarray*}

We can now use Lemma~\ref{lem:lemma4} to bound $\mu_k^2 \leq \mu_0^2$ with probability $1-O(\ensuremath{M}\xspace(\ensuremath{L}\xspace\ensuremath{N}\xspace)^{-\beta})$ and Lemma~\ref{lem:lemma1} to bound $\ensuremath{\left|\left| {\bm{G}_k} \right|\right|}\xspace_F \leq 2^{-k}\sqrt{R}$, which gives us a bound of
\begin{gather*}
	\ensuremath{\left|\left| {(\ensuremath{\mathcal{{A}}^{\ast}}\xspace\ensuremath{\mathcal{{A}}}\xspace - \mathcal{I})\bm{G}} \right|\right|}\xspace \leq c2^{-k/2}\max\left\{\sqrt{\frac{{\kappa}\beta\ensuremath{R}\xspace\log(\ensuremath{L}\xspace\ensuremath{N}\xspace)}{\ensuremath{M}\xspace}\log(\max\{\mu_0^2\ensuremath{L}\xspace,\ensuremath{N}\xspace\})}, \right.\nonumber \\
	\qquad \qquad \left. \sqrt{\mu_0^2\ensuremath{L}\xspace\ensuremath{N}\xspace}\frac{\beta{\kappa}}{\ensuremath{M}\xspace}\log(\ensuremath{L}\xspace\ensuremath{N}\xspace)\log(\min\{\mu_0^2\ensuremath{L}\xspace,\ensuremath{N}\xspace \}) \right\}. \nonumber 
\end{gather*}
Simplifying the bound using $\ensuremath{R}\xspace \leq \min\{\ensuremath{L}\xspace,\ensuremath{N}\xspace\}$, 
\begin{gather*}
	\ensuremath{\left|\left| {(\ensuremath{\mathcal{{A}}^{\ast}}\xspace\ensuremath{\mathcal{{A}}}\xspace - \mathcal{I})\bm{G}} \right|\right|}\xspace \leq c2^{-k/2}\max\left\{\sqrt{\frac{{\kappa}\beta\ensuremath{R}\xspace\log(\ensuremath{L}\xspace\ensuremath{N}\xspace)}{\ensuremath{M}\xspace}\log(\ensuremath{L}\xspace\ensuremath{N}\xspace)}, \sqrt{\mu_0^2\ensuremath{L}\xspace\ensuremath{N}\xspace}\frac{\beta{\kappa}}{\ensuremath{M}\xspace}\log^2(\ensuremath{L}\xspace\ensuremath{N}\xspace)) \right\} \nonumber 
\end{gather*}
Taking 
\begin{gather*}
	\ensuremath{M}\xspace \geq c\beta{\kappa}\ensuremath{R}\xspace\max\{\ensuremath{N}\xspace,\ensuremath{L}\xspace\mu_0^2\}\log^2(\ensuremath{L}\xspace\ensuremath{N}\xspace), \nonumber
\end{gather*}
proves the lemma. To simplify the bound on the probability, we note that Lemma~\ref{lem:lemma4} holds with probability $1-O(\ensuremath{M}\xspace(\ensuremath{L}\xspace\ensuremath{N}\xspace)^{-\beta})$ and this lemma holds with probability $1-O({\kappa}(\ensuremath{L}\xspace\ensuremath{N}\xspace)^{-\beta})$. Since ${\kappa} < \ensuremath{M}\xspace$ and assuming that $\ensuremath{M}\xspace \leq \ensuremath{L}\xspace\ensuremath{N}\xspace$, we can write that the result holds with probability $1-O((\ensuremath{L}\xspace\ensuremath{N}\xspace)^{1-\beta})$. Additionally, since Lemma~\ref{lem:lemma4} holds when 
\begin{gather*}
	\ensuremath{M}\xspace \geq c\beta{\kappa}\ensuremath{R}\xspace\left(\ensuremath{N}\xspace + \ensuremath{L}\xspace\mu_0^2\right)\log^2(\ensuremath{L}\xspace\ensuremath{N}\xspace) \geq c\beta{\kappa}\ensuremath{R}\xspace\max\{\ensuremath{N}\xspace,\ensuremath{L}\xspace\mu_0^2\}\log^2(\ensuremath{L}\xspace\ensuremath{N}\xspace), \nonumber
\end{gather*}
Then both lemmas hold under the same condition.
\end{proof}

\subsubsection{\texorpdfstring{Bound on $\ensuremath{\mathbb{E}\left[{|\ensuremath{\langle {\bm{C}},{\ensuremath{\bm{{A}}_n}\xspace}\rangle}\xspace|^2\ensuremath{\bm{{z}}_n}\xspace\ensuremath{\bm{{z}}^{\ast}_n}\xspace}\right]}\xspace$}{Proof of Lemma~\ref{lem:lemma3}}}

Lemma~\ref{lem:lemma3} bounds the spectrum of the expected matrix 
\[ \ensuremath{\mathbb{E}\left[{|\ensuremath{\langle {\bm{G}},{\ensuremath{\bm{{A}}_n}\xspace}\rangle}\xspace|^2\ensuremath{\bm{{z}}_n}\xspace\ensuremath{\bm{{z}}^{\ast}_n}\xspace}\right]}\xspace: \]
\begin{lemma}
	\label{lem:lemma3}
	Suppose $\ensuremath{\bm{{A}}_n}\xspace=\ensuremath{\bm{{z}}_n}\xspace\ensuremath{\bm{{f}}^{\ast}_n}\xspace$ be defined as the outer product of an \emph{i.i.d.} random Gaussian vector $\ensuremath{\bm{{z}}_n}\xspace$ with zero mean and variance $1/M$ and a random Fourier vector $\ensuremath{\bm{{f}}_n}\xspace$. Then the operator  $|\ensuremath{\langle {\bm{C}},{\ensuremath{\bm{{A}}_n}\xspace}\rangle}\xspace|^2\ensuremath{\bm{{z}}_n}\xspace\ensuremath{\bm{{z}}^{\ast}_n}\xspace$ satisfies
	\begin{gather*}
		\ensuremath{\mathbb{E}_{z}\left[{|\ensuremath{\langle {\bm{C}},{\ensuremath{\bm{{A}}_n}\xspace}\rangle}\xspace|^2 \ensuremath{\bm{{z}}_n}\xspace\ensuremath{\bm{{z}}^{\ast}_n}\xspace}\right]}\xspace  \preceq \frac{3}{\ensuremath{M}\xspace^2}\|\bm{C}^{\ast}\ensuremath{\bm{{f}}_n}\xspace\|_2^2\bm{I}_\ensuremath{M}\xspace, \nonumber
	\end{gather*}
	and
	\begin{gather*}
		\ensuremath{\mathbb{E}_{z,f}\left[{|\ensuremath{\langle {\bm{C}},{\ensuremath{\bm{{A}}_n}\xspace}\rangle}\xspace|^2 \ensuremath{\bm{{z}}_n}\xspace\ensuremath{\bm{{z}}^{\ast}_n}\xspace}\right]}\xspace  \preceq \frac{3}{\ensuremath{M}\xspace^2}\|\bm{C}\|_F^2\bm{I}_\ensuremath{M}\xspace.  \nonumber
	\end{gather*}
\end{lemma}

\begin{proof}

To begin the proof, we look at the expectation of each element of the matrix. We first calculate the expectation with respect to $\ensuremath{\bm{{z}}_n}\xspace$,
\begin{eqnarray*}
	\ast  & = & \ensuremath{\mathbb{E}_{z}\left[{\left|\sum_{l=1}^Lz_{n}[l]\bm{c}_l^{\ast}\ensuremath{\bm{{f}}_n}\xspace \right|^2z_{n}[\alpha]z^{\ast}_n[\beta]}\right]}\xspace \nonumber  \\  
        & = &  \ensuremath{\mathbb{E}_{z}\left[{ \left(\sum_{l=1}^Lz_{n}[l]\bm{c}_l^{\ast}\ensuremath{\bm{{f}}_n}\xspace\right)^{\ast} \left(\sum_{l=1}^Lz_{n}[l]\bm{c}_l^{\ast}\ensuremath{\bm{{f}}_n}\xspace \right) z_{n}[\alpha]z^{\ast}_n[\beta] }\right]}\xspace \nonumber \\
	& = & \ensuremath{\mathbb{E}_{z}\left[{ \sum_{l=1}^L|z_{n}[l]|^2|\bm{c}_l^{\ast}\ensuremath{\bm{{f}}_n}\xspace|^2z_n[\alpha]z_n^{\ast}[\beta] + 2\sum_{k\neq l}\mbox{Re}\left(z_n^{\ast}[l]z_n[k] \ensuremath{\langle {\bm{c}_l^{\ast}\ensuremath{\bm{{f}}_n}\xspace},{\bm{c}_k^{\ast}\ensuremath{\bm{{f}}_n}\xspace}\rangle}\xspace \right) z_{n}[\alpha]z^{\ast}_n[\beta]}\right]}\xspace \nonumber \\
	& = & \left(\frac{3}{2\ensuremath{M}\xspace^2}|\bm{c}_\alpha^{\ast}\ensuremath{\bm{{f}}_n}\xspace|^2 + \frac{1}{\ensuremath{M}\xspace^2}\|\bm{C}\ensuremath{\bm{{f}}_n}\xspace\|_2^2\right)\delta_{\alpha = \beta} + \frac{2}{\ensuremath{M}\xspace^2} \ensuremath{\langle {\bm{c}_\alpha^{\ast}\ensuremath{\bm{{f}}_n}\xspace},{\bm{c}_\beta^{\ast}\ensuremath{\bm{{f}}_n}\xspace}\rangle}\xspace\delta_{\alpha\neq\beta}. \nonumber
\end{eqnarray*}

We can then use the matrix formulation
\begin{eqnarray*}
	\ensuremath{\mathbb{E}_{z}\left[{|\ensuremath{\langle {\bm{C}},{\ensuremath{\bm{{A}}_n}\xspace}\rangle}\xspace|^2 \ensuremath{\bm{{z}}_n}\xspace\ensuremath{\bm{{z}}^{\ast}_n}\xspace}\right]}\xspace & = & \frac{3}{2\ensuremath{M}\xspace^2}\mbox{diag}(\bm{C}\ensuremath{\bm{{f}}_n}\xspace\ensuremath{\bm{{f}}^{\ast}_n}\xspace\bm{C}^{\ast}) + \frac{1}{\ensuremath{M}\xspace^2}\|\bm{C}\ensuremath{\bm{{f}}_n}\xspace\|_2^2\bm{I}_\ensuremath{M}\xspace \nonumber \\
         & & \qquad \qquad + \frac{2}{\ensuremath{M}\xspace^2}\bm{C}\ensuremath{\bm{{f}}_n}\xspace\ensuremath{\bm{{f}}^{\ast}_n}\xspace\bm{C}^{\ast} + \frac{2}{\ensuremath{M}\xspace^2}\mbox{diag}(\bm{C}\ensuremath{\bm{{f}}_n}\xspace\ensuremath{\bm{{f}}^{\ast}_n}\xspace\bm{C}^{\ast}) \nonumber \\
	& = & \frac{1}{\ensuremath{M}\xspace^2}\|\bm{C}\ensuremath{\bm{{f}}_n}\xspace\|_2^2\bm{I}_\ensuremath{M}\xspace + 2\bm{C}\ensuremath{\bm{{f}}_n}\xspace\ensuremath{\bm{{f}}^{\ast}_n}\xspace\bm{C}^{\ast} - \frac{1}{2}\mbox{diag}(\bm{C}\ensuremath{\bm{{f}}_n}\xspace\ensuremath{\bm{{f}}^{\ast}_n}\xspace\bm{C}^{\ast}) \nonumber \\
	& \preceq & \frac{3}{\ensuremath{M}\xspace^2} \|\bm{C}\ensuremath{\bm{{f}}_n}\xspace\|_2^2\bm{I}_\ensuremath{M}\xspace,  \nonumber
\end{eqnarray*}
where to obtain the result we first use the linearity of the expectation along with the with the positive-semidefinite property of diag($\bm{C}\ensuremath{\bm{{f}}_n}\xspace\ensuremath{\bm{{f}}^{\ast}_n}\xspace\bm{C}^{\ast}$), proving the fist portion of the Lemma. To prove the second portion we simply take an expectation with respect to $\ensuremath{\bm{{f}}_n}\xspace$:
\begin{gather*}
	\ensuremath{\mathbb{E}_{z,f}\left[{|\ensuremath{\langle {\bm{C}},{\ensuremath{\bm{{A}}_n}\xspace}\rangle}\xspace|^2 \ensuremath{\bm{{z}}_n}\xspace\ensuremath{\bm{{z}}^{\ast}_n}\xspace}\right]}\xspace  \preceq  \ensuremath{\mathbb{E}_{f}\left[{\frac{3}{\ensuremath{M}\xspace^2}\|\bm{C}\ensuremath{\bm{{f}}_n}\xspace\|_2^2\bm{I}_\ensuremath{M}\xspace }\right]}\xspace \preceq \frac{3}{\ensuremath{M}\xspace^2} \|\bm{C}\|_F^2\bm{I},_\ensuremath{M}\xspace, \nonumber
\end{gather*}
completing the proof.
\end{proof}

\subsubsection{\texorpdfstring{Contractive Property of $\mu_k^2$}{Proof of Lemma~\ref{lem:lemma4}}}

\begin{lemma}
	\label{lem:lemma4}
	Let $\mu_k^2$ be the coherence factor as defined in Equation~\eqref{eqn:coherencek}, and additionally assume that $\ensuremath{L}\xspace > 1$ and that $\ensuremath{L}\xspace\ensuremath{N}\xspace > \ensuremath{R}\xspace\mu_0^4$. If 
	\begin{gather*}
		\ensuremath{M}\xspace \geq c\beta{\kappa}\ensuremath{R}\xspace\left(\ensuremath{N}\xspace + \ensuremath{L}\xspace\mu_0^2\right)\log^2(\ensuremath{L}\xspace\ensuremath{N}\xspace), \nonumber
	\end{gather*}
	then with probability at least $1-O( {\kappa}(\ensuremath{L}\xspace\ensuremath{N}\xspace)^{-\beta})$, 
	\begin{gather*}
		\mu_k^2 \leq 2^{-1}\mu_{k-1}^2, \nonumber 
	\end{gather*}
	for all $k \in [1,\cdots, {\kappa}]$.
\end{lemma}

\begin{proof}

        In Lemma~\ref{lem:lemma4} we show that the coherence term reduces at each golfing iteration. Observe that
        \begin{gather*}
	\mu_k^2 =  \frac{1}{\ensuremath{R}\xspace}\sup_{\omega} \sum_{l=1}^{\ensuremath{L}\xspace} \ensuremath{\langle {\ensuremath{\widetilde{\bm{Y}}}\xspace_k},{\bm{e}_l\ensuremath{\bm{{f}}^{\ast}}\xspace}\rangle}\xspace^2
        \end{gather*}
\begin{eqnarray*}
        & = & \frac{1}{\ensuremath{R}\xspace}\sup_{\omega} \sum_{l=1}^{\ensuremath{L}\xspace}\left(\sum_{n\in\Gamma_k} {\kappa}\ensuremath{\langle {\ensuremath{\mathcal{P}_T\left( \ensuremath{\bm{{A}}_n}\xspace \right)}\xspace},{\bm{e}_l\ensuremath{\bm{{f}}^{\ast}}\xspace}\rangle}\xspace\ensuremath{\langle {\ensuremath{\widetilde{\bm{Y}}}\xspace_{k-1}},{\ensuremath{\bm{{A}}_n}\xspace}\rangle}\xspace - \ensuremath{\langle {\ensuremath{\widetilde{\bm{Y}}}\xspace_{k-1}},{\bm{e}_l\ensuremath{\bm{{f}}^{\ast}}\xspace}\rangle}\xspace\right)^2 \nonumber \\
	& = & \frac{1}{\ensuremath{R}\xspace}\sup_{\omega} \sum_{l=1}^{\ensuremath{L}\xspace}\left(\sum_{n\in\Gamma_k} {\kappa}\ensuremath{\langle {\ensuremath{\mathcal{P}_T\left( \ensuremath{\bm{{A}}_n}\xspace \right)}\xspace},{\bm{e}_l\ensuremath{\bm{{f}}^{\ast}}\xspace}\rangle}\xspace\ensuremath{\langle {\ensuremath{\widetilde{\bm{Y}}}\xspace_{k-1}},{\ensuremath{\bm{{A}}_n}\xspace}\rangle}\xspace - \ensuremath{\mathbb{E}\left[{{\kappa}\ensuremath{\langle {\ensuremath{\mathcal{P}_T\left( \ensuremath{\bm{{A}}_n}\xspace \right)}\xspace},{\bm{e}_l\ensuremath{\bm{{f}}^{\ast}}\xspace}\rangle}\xspace\ensuremath{\langle {\ensuremath{\widetilde{\bm{Y}}}\xspace_{k-1}},{\ensuremath{\bm{{A}}_n}\xspace}\rangle}\xspace}\right]}\xspace \right)^2. \nonumber
\end{eqnarray*}
To bound this quantity we use the scalar Bernstein inequality on each of the inner quantities 
\begin{gather*}
	\sum_{n\in\Gamma_k}X_n = \sum_{n\in\Gamma_k} {\kappa}\ensuremath{\langle {\ensuremath{\mathcal{P}_T\left( \ensuremath{\bm{{A}}_n}\xspace \right)}\xspace},{\bm{e}_l\ensuremath{\bm{{f}}^{\ast}}\xspace}\rangle}\xspace\ensuremath{\langle {\ensuremath{\widetilde{\bm{Y}}}\xspace_{k-1}},{\ensuremath{\bm{{A}}_n}\xspace}\rangle}\xspace - \ensuremath{\mathbb{E}\left[{{\kappa}\ensuremath{\langle {\ensuremath{\mathcal{P}_T\left( \ensuremath{\bm{{A}}_n}\xspace \right)}\xspace},{\bm{e}_l\ensuremath{\bm{{f}}^{\ast}}\xspace}\rangle}\xspace\ensuremath{\langle {\ensuremath{\widetilde{\bm{Y}}}\xspace_{k-1}},{\ensuremath{\bm{{A}}_n}\xspace}\rangle}\xspace}\right]}\xspace. \nonumber
\end{gather*}
As in the matrix Bernstein formulation, we require both the variance and Orlicz norm. First we find the variance,
\begin{eqnarray*}
	\sum_{n\in\Gamma_k} \ensuremath{\mathbb{E}\left[{X_n X_n^{\ast}}\right]}\xspace & = & {\kappa}^2\sum_{n\in\Gamma_k} \ensuremath{\mathbb{E}\left[{|\ensuremath{\langle {\ensuremath{\mathcal{P}_T\left( \ensuremath{\bm{{A}}_n}\xspace \right)}\xspace},{\bm{e}_l\ensuremath{\bm{{f}}^{\ast}}\xspace}\rangle}\xspace|^2|\ensuremath{\langle {\ensuremath{\widetilde{\bm{Y}}}\xspace_{k-1}},{\ensuremath{\bm{{A}}_n}\xspace}\rangle}\xspace|^2}\right]}\xspace \nonumber \\
        & & \qquad\qquad\qquad  - |\ensuremath{\mathbb{E}\left[{\ensuremath{\langle {\ensuremath{\mathcal{P}_T\left( \ensuremath{\bm{{A}}_n}\xspace \right)}\xspace},{\bm{e}_l\ensuremath{\bm{{f}}^{\ast}}\xspace}\rangle}\xspace\ensuremath{\langle {\ensuremath{\widetilde{\bm{Y}}}\xspace_{k-1}},{\ensuremath{\bm{{A}}_n}\xspace}\rangle}\xspace }\right]}\xspace|^2 \nonumber \\
        & \leq & {\kappa}^2\sum_{n\in\Gamma_k} \mathbb{E}\left[|\ensuremath{\langle {\ensuremath{\bm{Q}}\xspace\ensuremath{\bm{Q}^{\ast}}\xspace\ensuremath{\bm{{z}}_n}\xspace\ensuremath{\bm{{f}}^{\ast}_n}\xspace},{\bm{e}_l\ensuremath{\bm{{f}}^{\ast}}\xspace}\rangle}\xspace + \ensuremath{\langle {\ensuremath{\bm{{z}}_n}\xspace\ensuremath{\bm{{f}}^{\ast}_n}\xspace\ensuremath{\bm{V}}\xspace\ensuremath{\bm{V}}\xspace},{\bm{e}_l\ensuremath{\bm{{f}}^{\ast}}\xspace}\rangle}\xspace\right. \nonumber \\
        & & \qquad\qquad\qquad \left.+ \ensuremath{\langle {\ensuremath{\bm{Q}}\xspace\ensuremath{\bm{Q}^{\ast}}\xspace\ensuremath{\bm{{z}}_n}\xspace\ensuremath{\bm{{f}}^{\ast}_n}\xspace\ensuremath{\bm{V}}\xspace\ensuremath{\bm{V}^{\ast}}\xspace},{\bm{e}_l\ensuremath{\bm{{f}}^{\ast}}\xspace}\rangle}\xspace|^2|\ensuremath{\langle {\ensuremath{\widetilde{\bm{Y}}}\xspace_{k-1}},{\ensuremath{\bm{{A}}_n}\xspace}\rangle}\xspace|^2\right] \nonumber \\
	& \leq & {\kappa}^2\sum_{n\in\Gamma_k} \ensuremath{\mathbb{E}}\xspace \left[\left(|\ensuremath{\langle {\ensuremath{\bm{Q}}\xspace\ensuremath{\bm{Q}^{\ast}}\xspace\ensuremath{\bm{{z}}_n}\xspace\ensuremath{\bm{{f}}^{\ast}_n}\xspace},{\bm{e}_l\ensuremath{\bm{{f}}^{\ast}}\xspace}\rangle}\xspace|^2 + |\ensuremath{\langle {\ensuremath{\bm{{f}}^{\ast}_n}\xspace\ensuremath{\bm{V}}\xspace\ensuremath{\bm{V}^{\ast}}\xspace},{\ensuremath{\bm{{f}}^{\ast}}\xspace}\rangle}\xspace z_n[l]|^2 \right.  \right. \nonumber \\
& &  \qquad\qquad\qquad \left. \left. + |\ensuremath{\langle {\ensuremath{\bm{Q}}\xspace\ensuremath{\bm{Q}^{\ast}}\xspace\ensuremath{\bm{{z}}_n}\xspace\ensuremath{\bm{{f}}^{\ast}_n}\xspace\ensuremath{\bm{V}}\xspace\ensuremath{\bm{V}^{\ast}}\xspace},{\bm{e}_l\ensuremath{\bm{{f}}^{\ast}}\xspace}\rangle}\xspace|^2\right)| \ensuremath{\langle {\ensuremath{\widetilde{\bm{Y}}}\xspace_{k-1}},{\ensuremath{\bm{{A}}_n}\xspace}\rangle}\xspace|^2 \right]. \nonumber
\end{eqnarray*}

This sum consists of three terms. The first of which can be bounded using Lemma~\ref{lem:lemma3},
\begin{eqnarray*}
	& & \sum_{n\in\Gamma_k}\ensuremath{\mathbb{E}\left[{|\ensuremath{\langle {\ensuremath{\bm{Q}}\xspace\ensuremath{\bm{Q}^{\ast}}\xspace\ensuremath{\bm{{z}}_n}\xspace\ensuremath{\bm{{f}}^{\ast}_n}\xspace},{\bm{e}_l\ensuremath{\bm{{f}}^{\ast}}\xspace}\rangle}\xspace|^2|\ensuremath{\langle {\ensuremath{\widetilde{\bm{Y}}}\xspace_{k-1}},{\ensuremath{\bm{{A}}_n}\xspace}\rangle}\xspace|^2}\right]}\xspace  \nonumber \\
        & & \qquad\qquad\qquad =  \sum_{n\in\Gamma_k}\ensuremath{\mathbb{E}\left[{|\ensuremath{\bm{{f}}^{\ast}_n}\xspace\ensuremath{\bm{{f}}}\xspace\ensuremath{\langle {\bm{q}_l},{\ensuremath{\bm{Q}^{\ast}}\xspace\ensuremath{\bm{{z}}_n}\xspace}\rangle}\xspace|^2|\ensuremath{\langle {\ensuremath{\widetilde{\bm{Y}}}\xspace_{k-1}},{\ensuremath{\bm{{A}}_n}\xspace}\rangle}\xspace|^2}\right]}\xspace \nonumber \\
	& & \qquad\qquad\qquad\leq  3\sum_{n\in\Gamma_k}\ensuremath{\mathbb{E}\left[{\ensuremath{\bm{{f}}^{\ast}}\xspace\ensuremath{\bm{{f}}_n}\xspace\ensuremath{\bm{{f}}^{\ast}_n}\xspace\ensuremath{\bm{{f}}}\xspace\bm{q}_l^{\ast}\ensuremath{\bm{Q}^{\ast}}\xspace\ensuremath{\left|\left| {\ensuremath{\widetilde{\bm{Y}}}\xspace_{k-1}\ensuremath{\bm{{f}}_n}\xspace} \right|\right|}\xspace_2^2\bm{I}_\ensuremath{L}\xspace\ensuremath{\bm{Q}}\xspace\bm{q}_l}\right]}\xspace \nonumber \\
	& & \qquad\qquad \qquad\leq  \frac{3\ensuremath{R}\xspace\mu_{k-1}^2}{\ensuremath{M}\xspace^3}\ensuremath{\left|\left| {\bm{q}_l} \right|\right|}\xspace_2^2\sum_{n\in\Gamma_k}\ensuremath{\bm{{f}}^{\ast}}\xspace\ensuremath{\mathbb{E}\left[{\ensuremath{\bm{{f}}_n}\xspace\ensuremath{\bm{{f}}^{\ast}_n}\xspace}\right]}\xspace\ensuremath{\bm{{f}}}\xspace \nonumber \\
	& & \qquad\qquad\qquad = \frac{3\ensuremath{N}\xspace\ensuremath{R}\xspace\mu_{k-1}^2}{{\kappa}\ensuremath{M}\xspace^2}\ensuremath{\left|\left| {\bm{q}_l} \right|\right|}\xspace_2^2. \nonumber
\end{eqnarray*}

For the second term we have 
\begin{eqnarray*}
	& & \sum_{n\in\Gamma_k}\ensuremath{\mathbb{E}\left[{|\ensuremath{\langle {\ensuremath{\bm{{f}}^{\ast}_n}\xspace\ensuremath{\bm{V}}\xspace\ensuremath{\bm{V}^{\ast}}\xspace},{\ensuremath{\bm{{f}}^{\ast}}\xspace}\rangle}\xspace|^2|z_n[l]|^2|\ensuremath{\langle {\ensuremath{\widetilde{\bm{Y}}}\xspace_{k-1}},{\ensuremath{\bm{{A}}_n}\xspace}\rangle}\xspace|^2}\right]}\xspace  \nonumber \\
        & & \qquad\qquad\qquad = \sum_{n\in\Gamma_k}\ensuremath{\mathbb{E}\left[{|\ensuremath{\langle {\ensuremath{\bm{V}^{\ast}}\xspace\ensuremath{\bm{{f}}_n}\xspace},{\ensuremath{\bm{V}^{\ast}}\xspace\ensuremath{\bm{{f}}}\xspace}\rangle}\xspace|^2|z_n[l]|^2|\ensuremath{\langle {\ensuremath{\widetilde{\bm{Y}}}\xspace_{k-1}},{\ensuremath{\bm{{A}}_n}\xspace}\rangle}\xspace|^2}\right]}\xspace \nonumber \\
	& & \qquad\qquad\qquad = \sum_{n\in\Gamma_k}\ensuremath{\mathbb{E}\left[{\ensuremath{\bm{{f}}^{\ast}}\xspace\ensuremath{\bm{V}}\xspace\ensuremath{\bm{V}^{\ast}}\xspace\ensuremath{\bm{{f}}_n}\xspace\ensuremath{\bm{{f}}^{\ast}_n}\xspace\ensuremath{\bm{V}}\xspace\ensuremath{\bm{V}^{\ast}}\xspace\ensuremath{\bm{{f}}}\xspace|\ensuremath{\langle {\ensuremath{\widetilde{\bm{Y}}}\xspace_{k-1}},{z_n[l]\ensuremath{\bm{{z}}_n}\xspace\ensuremath{\bm{{f}}^{\ast}_n}\xspace}\rangle}\xspace|^2}\right]}\xspace. \nonumber 
\end{eqnarray*}
Using the fact that $|z_n[l]|^2 = \bm{e}_l^{\ast}\ensuremath{\bm{{z}}_n}\xspace\ensuremath{\bm{{z}}^{\ast}_n}\xspace\bm{e}_l$ and Lemma~\ref{lem:lemma3}, we obtain
\begin{eqnarray*}
	& & \sum_{n\in\Gamma_k}\ensuremath{\mathbb{E}\left[{|\ensuremath{\langle {\ensuremath{\bm{{f}}^{\ast}_n}\xspace\ensuremath{\bm{V}}\xspace\ensuremath{\bm{V}^{\ast}}\xspace},{\ensuremath{\bm{{f}}^{\ast}}\xspace}\rangle}\xspace|^2|z_n[l]|^2|\ensuremath{\langle {\ensuremath{\widetilde{\bm{Y}}}\xspace_{k-1}},{\ensuremath{\bm{{A}}_n}\xspace}\rangle}\xspace|^2}\right]}\xspace  \nonumber \\
        & & \qquad\qquad\qquad \leq \frac{3}{\ensuremath{M}\xspace^2}\sum_{n\in\Gamma_k}\ensuremath{\mathbb{E}\left[{\ensuremath{\bm{{f}}^{\ast}}\xspace\ensuremath{\bm{V}}\xspace\ensuremath{\bm{V}^{\ast}}\xspace\ensuremath{\bm{{f}}_n}\xspace\ensuremath{\bm{{f}}^{\ast}_n}\xspace\ensuremath{\bm{V}}\xspace\ensuremath{\bm{V}^{\ast}}\xspace\ensuremath{\bm{{f}}}\xspace\bm{e}_l^{\ast}\ensuremath{\left|\left| {\ensuremath{\widetilde{\bm{Y}}}\xspace_{k-1}\ensuremath{\bm{{f}}_n}\xspace} \right|\right|}\xspace_2^2\bm{I}_\ensuremath{L}\xspace\bm{e}_l}\right]}\xspace \nonumber \\
	& & \qquad\qquad\qquad \leq \frac{3\ensuremath{R}\xspace\mu_{k-1}^2}{\ensuremath{M}\xspace^2}\sum_{n\in\Gamma_k}\ensuremath{\mathbb{E}\left[{\ensuremath{\bm{{f}}^{\ast}}\xspace\ensuremath{\bm{V}}\xspace\ensuremath{\bm{V}^{\ast}}\xspace\ensuremath{\bm{{f}}_n}\xspace\ensuremath{\bm{{f}}^{\ast}_n}\xspace\ensuremath{\bm{V}}\xspace\ensuremath{\bm{V}^{\ast}}\xspace\ensuremath{\bm{{f}}}\xspace}\right]}\xspace \nonumber \\
	& & \qquad\qquad\qquad \leq \frac{3\ensuremath{R}\xspace\mu_{k-1}^2}{\ensuremath{M}\xspace^2}|\Gamma_k|\ensuremath{\left|\left| {\ensuremath{\bm{V}^{\ast}}\xspace\ensuremath{\bm{{f}}}\xspace} \right|\right|}\xspace_2^2 \nonumber \\
	& & \qquad\qquad\qquad \leq \frac{3\ensuremath{R}\xspace^2\mu_{k-1}^2\mu_0^2}{{\kappa}\ensuremath{M}\xspace^1}. \nonumber \\
\end{eqnarray*}

Finally, for the third term, we have
\begin{eqnarray*}
	& & \sum_{n\in\Gamma_k}\ensuremath{\mathbb{E}\left[{|\ensuremath{\langle {\ensuremath{\bm{Q}}\xspace\ensuremath{\bm{Q}^{\ast}}\xspace\ensuremath{\bm{{z}}_n}\xspace\ensuremath{\bm{{f}}^{\ast}_n}\xspace\ensuremath{\bm{V}}\xspace\ensuremath{\bm{V}^{\ast}}\xspace},{\bm{e}_l\ensuremath{\bm{{f}}^{\ast}}\xspace}\rangle}\xspace|^2|\ensuremath{\langle {\ensuremath{\widetilde{\bm{Y}}}\xspace_{k-1}},{\ensuremath{\bm{{A}}_n}\xspace}\rangle}\xspace|^2}\right]}\xspace  \nonumber \\
        & & \qquad\qquad\qquad = \sum_{n\in\Gamma_k}\ensuremath{\mathbb{E}\left[{|\ensuremath{\langle {\ensuremath{\bm{V}^{\ast}}\xspace\ensuremath{\bm{{f}}_n}\xspace},{\ensuremath{\bm{V}^{\ast}}\xspace\ensuremath{\bm{{f}}}\xspace}\rangle}\xspace|^2|\ensuremath{\langle {\bm{q}_l},{\ensuremath{\bm{Q}^{\ast}}\xspace\ensuremath{\bm{{z}}_n}\xspace}\rangle}\xspace|^2|\ensuremath{\langle {\ensuremath{\widetilde{\bm{Y}}}\xspace_{k-1}},{\ensuremath{\bm{{A}}_n}\xspace}\rangle}\xspace|^2}\right]}\xspace \nonumber \\
	& & \qquad\qquad\qquad = \sum_{n\in\Gamma_k}\ensuremath{\mathbb{E}\left[{\ensuremath{\bm{{f}}^{\ast}}\xspace\ensuremath{\bm{V}}\xspace\ensuremath{\bm{V}^{\ast}}\xspace\ensuremath{\bm{{f}}_n}\xspace\ensuremath{\bm{{f}}^{\ast}_n}\xspace\ensuremath{\bm{V}}\xspace\ensuremath{\bm{V}^{\ast}}\xspace\ensuremath{\bm{{f}}}\xspace|\ensuremath{\langle {\bm{q}_l},{\ensuremath{\bm{Q}^{\ast}}\xspace\ensuremath{\bm{{z}}_n}\xspace}\rangle}\xspace|^2|\ensuremath{\langle {\ensuremath{\widetilde{\bm{Y}}}\xspace_{k-1}},{\ensuremath{\bm{{A}}_n}\xspace}\rangle}\xspace|^2}\right]}\xspace \nonumber \\
	& & \qquad\qquad\qquad \leq \sum_{n\in\Gamma_k}\ensuremath{\mathbb{E}\left[{\ensuremath{\left|\left| {\ensuremath{\bm{V}^{\ast}}\xspace\ensuremath{\bm{{f}}}\xspace} \right|\right|}\xspace_2^2\bm{q}_l^{\ast}\ensuremath{\bm{Q}^{\ast}}\xspace\ensuremath{\bm{{z}}_n}\xspace\ensuremath{\bm{{z}}^{\ast}_n}\xspace|\ensuremath{\langle {\ensuremath{\widetilde{\bm{Y}}}\xspace_{k-1}},{\ensuremath{\bm{{A}}_n}\xspace}\rangle}\xspace|^2\ensuremath{\bm{Q}}\xspace\bm{q}_l  }\right]}\xspace \nonumber \\
	& & \qquad\qquad\qquad \leq \frac{3}{\ensuremath{M}\xspace^2}\ensuremath{\left|\left| {\bm{q}_l} \right|\right|}\xspace_2^2\ensuremath{\left|\left| {\ensuremath{\bm{V}^{\ast}}\xspace\ensuremath{\bm{{f}}}\xspace} \right|\right|}\xspace_2^2\sum_{n\in\Gamma_k}\ensuremath{\mathbb{E}\left[{\ensuremath{\left|\left| {\ensuremath{\widetilde{\bm{Y}}}\xspace_{k-1}\ensuremath{\bm{{f}}_n}\xspace} \right|\right|}\xspace_2^2  }\right]}\xspace \nonumber \\
	& & \qquad\qquad\qquad\leq \frac{3\ensuremath{R}\xspace^2\mu_0^2\mu_{k-1}^2}{{\kappa}\ensuremath{M}\xspace}\ensuremath{\left|\left| {\bm{q}_l} \right|\right|}\xspace_2^2. \nonumber
\end{eqnarray*}

Summing the three bounds and using $\ensuremath{\left|\left| {\bm{q}_l} \right|\right|}\xspace \leq 1$ yields
\begin{gather*}
	\sigma_X^2 \leq 9{\kappa}\left(\frac{\ensuremath{N}\xspace\ensuremath{R}\xspace\mu_{k-1}^2}{\ensuremath{M}\xspace^2}\ensuremath{\left|\left| {\bm{q}_l} \right|\right|}\xspace_2^2 + 2\frac{\ensuremath{R}\xspace^2\mu_0^2\mu_{k-1}^2}{\ensuremath{M}\xspace}\right). \nonumber
\end{gather*}

To use the Bernstein inequality it remains to find the Orlicz-1 norm of $X_n$. From Lemma~\ref{lem:lemma2} we have
\begin{gather*}
	\ensuremath{\left|\left| {\ensuremath{\langle {\ensuremath{\widetilde{\bm{Y}}}\xspace_{k-1}},{\ensuremath{\bm{{A}}_n}\xspace}\rangle}\xspace} \right|\right|}\xspace_{\psi_2}^2 = \ensuremath{\left|\left| {\ensuremath{\bm{{z}}^{\ast}_n}\xspace\ensuremath{\widetilde{\bm{Y}}}\xspace_{k-1}\ensuremath{\bm{{f}}_n}\xspace} \right|\right|}\xspace_{\psi_2}^2  \leq  \ensuremath{\left|\left| {\ensuremath{\left|\left| {\ensuremath{\bm{Q}^{\ast}}\xspace\ensuremath{\bm{{z}}_n}\xspace} \right|\right|}\xspace_2\ensuremath{\left|\left| {\bm{\Lambda_{k-1}}\ensuremath{\bm{V}^{\ast}}\xspace\ensuremath{\bm{{f}}_n}\xspace} \right|\right|}\xspace_2} \right|\right|}\xspace_{\psi_2}^2  \leq c\frac{\ensuremath{R}\xspace\mu^2_{k-1}}{\ensuremath{M}\xspace}. \nonumber
\end{gather*}

For the first term we have 
\begin{eqnarray*}
	\ensuremath{\left|\left| {\ensuremath{\bm{{f}}^{\ast}_n}\xspace\ensuremath{\bm{{f}}}\xspace\ensuremath{\langle {\bm{q}_l},{\ensuremath{\bm{Q}^{\ast}}\xspace\ensuremath{\bm{{z}}_n}\xspace}\rangle}\xspace} \right|\right|}\xspace_{\psi_2}^2 & \leq & \ensuremath{\left|\left| {\ensuremath{\left|\left| {\ensuremath{\bm{{f}}^{\ast}_n}\xspace} \right|\right|}\xspace_2\ensuremath{\left|\left| {\ensuremath{\bm{{f}}}\xspace} \right|\right|}\xspace_2|\ensuremath{\langle {\bm{q}_l},{\ensuremath{\bm{Q}^{\ast}}\xspace\ensuremath{\bm{{z}}_n}\xspace}\rangle}\xspace} \right|\right|}\xspace_{\psi_2}^2 \nonumber \\
	& \leq & \ensuremath{N}\xspace^2\ensuremath{\left|\left| {\bm{q}_l^{\ast}\ensuremath{\bm{Q}^{\ast}}\xspace\ensuremath{\bm{{z}}_n}\xspace} \right|\right|}\xspace_{\psi_2}^2 \nonumber \\
	& \leq & c\frac{\ensuremath{N}\xspace^2\ensuremath{\left|\left| {\bm{q}_l} \right|\right|}\xspace_2^2}{\ensuremath{M}\xspace^2}. \nonumber
\end{eqnarray*}

For the second term we have 
\begin{gather*}
	\ensuremath{\left|\left| {\ensuremath{\langle {\ensuremath{\bm{V}^{\ast}}\xspace\ensuremath{\bm{{f}}_n}\xspace},{\ensuremath{\bm{V}^{\ast}}\xspace\ensuremath{\bm{{f}}}\xspace}\rangle}\xspace z_n[l]} \right|\right|}\xspace_{\psi_2}^2 \leq \ensuremath{\left|\left| {\ensuremath{\left|\left| {\ensuremath{\bm{V}^{\ast}}\xspace\ensuremath{\bm{{f}}_n}\xspace} \right|\right|}\xspace_2\ensuremath{\left|\left| {\ensuremath{\bm{V}^{\ast}}\xspace\ensuremath{\bm{{f}}}\xspace} \right|\right|}\xspace_2z_n[l]} \right|\right|}\xspace_{\psi_2}^2 \leq c\frac{\ensuremath{R}\xspace^2}{\ensuremath{M}\xspace}\mu_0^4. \nonumber 
\end{gather*}
Similarly, for the final term we have 
\begin{gather*}
	\ensuremath{\left|\left| {\ensuremath{\langle {\ensuremath{\bm{V}^{\ast}}\xspace\ensuremath{\bm{{f}}_n}\xspace},{\ensuremath{\bm{V}^{\ast}}\xspace\ensuremath{\bm{{f}}}\xspace}\rangle}\xspace\ensuremath{\langle {\bm{q}_l},{\ensuremath{\bm{Q}^{\ast}}\xspace\ensuremath{\bm{{z}}_n}\xspace}\rangle}\xspace} \right|\right|}\xspace_{\psi_2}^2 \leq \ensuremath{\left|\left| {\ensuremath{\left|\left| {\ensuremath{\bm{V}^{\ast}}\xspace\ensuremath{\bm{{f}}_n}\xspace} \right|\right|}\xspace_2\ensuremath{\left|\left| {\ensuremath{\bm{V}^{\ast}}\xspace\ensuremath{\bm{{f}}}\xspace} \right|\right|}\xspace_2 \bm{q}_l^{\ast}\ensuremath{\bm{Q}^{\ast}}\xspace\ensuremath{\bm{{z}}_n}\xspace} \right|\right|}\xspace_{\psi_2}^2 \leq c\frac{\ensuremath{R}\xspace^2}{\ensuremath{M}\xspace}\mu_0^4\ensuremath{\left|\left| {\bm{q}_l} \right|\right|}\xspace_2^2. \nonumber 
\end{gather*}

Now we can calculate the total Orlicz norm as
\begin{eqnarray*}
	\ensuremath{\left|\left| {X_n} \right|\right|}\xspace_{\psi_1}^2 & \leq & c{\kappa}^2\frac{\ensuremath{R}\xspace\mu^2_{k-1}}{\ensuremath{M}\xspace}\left(\frac{\ensuremath{N}\xspace^2\ensuremath{\left|\left| {\bm{q}_l} \right|\right|}\xspace_2^2}{\ensuremath{M}\xspace^2} + \frac{\ensuremath{R}\xspace^2}{\ensuremath{M}\xspace}\mu_0^4 + \frac{\ensuremath{R}\xspace^2}{\ensuremath{M}\xspace}\mu_0^4\ensuremath{\left|\left| {\bm{q}_l} \right|\right|}\xspace_2^2\right)  \nonumber \\
	 & \leq & c\frac{{\kappa}^2\ensuremath{R}\xspace\ensuremath{N}\xspace^2}{\ensuremath{M}\xspace^3}\ensuremath{\left|\left| {\bm{q}_l} \right|\right|}\xspace_2^2\mu^2_{k-1} + c\frac{2{\kappa}^2\ensuremath{R}\xspace^3}{\ensuremath{M}\xspace^2}\mu_0^4\mu_{k-1}^2.  \nonumber 
\end{eqnarray*}

Since we wish to bound the square of the sum of terms, we calculate the square values of the two terms in the Bernstein inequality. The first term is bounded by 
\begin{gather*}
        t\sigma_X^2 \leq c\beta{\kappa}\frac{\ensuremath{R}\xspace}{\ensuremath{M}\xspace}\mu_{k-1}^2\left(\frac{\ensuremath{N}\xspace}{\ensuremath{M}\xspace}\ensuremath{\left|\left| {\bm{q}_l} \right|\right|}\xspace_2^2 + 2\ensuremath{R}\xspace\mu_0^2\right)\log(\ensuremath{L}\xspace\ensuremath{N}\xspace), \nonumber
\end{gather*}
and the second term is bounded by
\begin{gather*}
	t^2U_{\alpha}^2\log^2\left(\frac{|\Gamma_k|U_{\alpha}^2}{\sigma_X^2}\right) \leq t^2U_{\alpha}^2\log^2\left(c\frac{|\Gamma_k|\ensuremath{M}\xspace{\kappa}^2\ensuremath{R}\xspace\mu^2_{k-1}\left(\frac{\ensuremath{N}\xspace^2}{\ensuremath{M}\xspace}\ensuremath{\left|\left| {\bm{q}_l} \right|\right|}\xspace_2^2 + 2\ensuremath{R}\xspace^2\mu_0^4 \right)}{\ensuremath{M}\xspace^2{\kappa}\ensuremath{R}\xspace\mu_{k-1}^2\left(\frac{\ensuremath{N}\xspace}{\ensuremath{M}\xspace}\ensuremath{\left|\left| {\bm{q}_l} \right|\right|}\xspace_2^2 + 2\ensuremath{R}\xspace\mu_0^2\right) }\right) 
\end{gather*}
\begin{eqnarray*}
	& \leq & t^2c{\kappa}^2\frac{\ensuremath{R}\xspace}{\ensuremath{M}\xspace^2}\mu_{k-1}^2\left(\frac{\ensuremath{N}\xspace^2}{\ensuremath{M}\xspace}\ensuremath{\left|\left| {\bm{q}_l} \right|\right|}\xspace_2^2 + 2\ensuremath{R}\xspace^2\mu_0^4 \right)\log^2\left(c\frac{\ensuremath{N}\xspace^2\ensuremath{\left|\left| {\bm{q}_l} \right|\right|}\xspace_2^2 + 2\ensuremath{M}\xspace\ensuremath{R}\xspace^2\mu_0^4 }{\ensuremath{N}\xspace\ensuremath{\left|\left| {\bm{q}_l} \right|\right|}\xspace_2^2 + 2\ensuremath{R}\xspace\ensuremath{M}\xspace\mu_0^2 }\right) \nonumber \\
	& \leq & c\beta^2{\kappa}^2\frac{\ensuremath{R}\xspace}{\ensuremath{M}\xspace^2}\mu_{k-1}^2\left(\frac{\ensuremath{N}\xspace^2}{\ensuremath{M}\xspace}\ensuremath{\left|\left| {\bm{q}_l} \right|\right|}\xspace_2^2 + 2\ensuremath{R}\xspace^2\mu_0^4 \right)\log^4\left(\ensuremath{L}\xspace\ensuremath{N}\xspace \right), \nonumber
\end{eqnarray*}
where the last step assumes $\ensuremath{L}\xspace > 1$ and $\ensuremath{L}\xspace\ensuremath{N}\xspace > \ensuremath{R}\xspace\mu_0^4$.
Each summand is then bounded by the maximum of these two quantities with probability $1-O(|\Gamma_k|(\ensuremath{L}\xspace\ensuremath{N}\xspace)^{-\beta})$, the $|\Gamma_k|$ term coming from the union bound over all terms in each inner sum. 

Using this bound on each summand, we obtain the total bound by taking a union bound, summing over $l \in [1,\cdots,\ensuremath{L}\xspace]$, and dividing by $\ensuremath{R}\xspace$, yielding a bound of the maximum of
\begin{gather*}
	t\sigma_X^2 \leq c\beta{\kappa}\frac{\ensuremath{R}\xspace}{\ensuremath{M}\xspace}\mu_{k-1}^2\left(\frac{\ensuremath{N}\xspace}{\ensuremath{M}\xspace} + 2\ensuremath{L}\xspace\mu_0^4\right)\log(\ensuremath{L}\xspace\ensuremath{N}\xspace), \nonumber
\end{gather*}
and
\begin{gather*}
	 t^2U_{\alpha}^2\log^2\left(\frac{|\Gamma_k|U_{\alpha}^2}{\sigma_X^2}\right) \leq  \beta^2c{\kappa}^2\frac{\ensuremath{R}\xspace}{\ensuremath{M}\xspace^2}\mu_{k-1}^2\left(\frac{\ensuremath{N}\xspace}{\ensuremath{M}\xspace} + 2\ensuremath{R}\xspace\ensuremath{L}\xspace\mu_0^2 \right)\log^4\left(\ensuremath{L}\xspace\ensuremath{N}\xspace \right), \nonumber
\end{gather*}
with probability $1 - O(\ensuremath{M}\xspace(\ensuremath{N}\xspace\ensuremath{L}\xspace)^{-\beta})$. 
To complete the proof, we note that if we have
\begin{gather*}
	\ensuremath{M}\xspace \geq c\beta{\kappa}\ensuremath{R}\xspace\left(\ensuremath{N}\xspace + \ensuremath{L}\xspace\mu_0^2\right)\log^2(\ensuremath{L}\xspace\ensuremath{N}\xspace), \nonumber
\end{gather*}
then both terms in this bound are less than $\mu_{k-1}^2$. 
\end{proof}



\bibliographystyle{natbib}           


\end{document}